\documentclass[a4paper,onecolumn,11pt,accepted=2026-07-14]{quantumarticle}
\pdfoutput=1
\usepackage[utf8]{inputenc}
\usepackage{graphicx}
\usepackage{amsmath}
\usepackage{amssymb}
\usepackage{hyperref}
\usepackage{amsfonts}
\usepackage{amsthm}
\usepackage{bbold}
\usepackage{indentfirst}
\usepackage[dvipsnames,usenames]{xcolor}
\usepackage{mathrsfs}
\usepackage{apptools}
\usepackage{multirow}
\usepackage{array}
\usepackage{tabularx}
\usepackage{ragged2e}
\usepackage[dvipsnames]{xcolor}
\usepackage{comment}
\AtAppendix{\counterwithin{theorem}{section}}
\usepackage{bm}
\usepackage[numbers,sort&compress]{natbib}

\newcommand{\norm}[1]{{\left\|#1\right\|}}

\newtheorem{theorem}{Theorem}
\newtheorem{lemma}{Lemma}

\usepackage{tikz}
\usetikzlibrary{quantikz}

\begin{document}

\title{Quantum Simulation of Nuclear Dynamics in First Quantization}

\author{Luca~Spagnoli}
\affiliation{Physics Department, University of Trento, Via Sommarive 14, I-38123 Trento, Italy}
\affiliation{INFN-TIFPA Trento Institute of Fundamental Physics and Applications, Trento, Italy}
\author{Chiara~Lissoni}
\affiliation{Physics Department, University of Trento, Via Sommarive 14, I-38123 Trento, Italy}
\author{Alessandro~Roggero}
\affiliation{Physics Department, University of Trento, Via Sommarive 14, I-38123 Trento, Italy}
\affiliation{INFN-TIFPA Trento Institute of Fundamental Physics and Applications, Trento, Italy}

\begin{abstract}
The study of real time dynamics of nuclear systems is of great importance to provide theoretical predictions of cross sections relevant for both terrestrial experiments as well as applications in astrophysics. First principles simulations of these dynamical processes is however hindered by an exponential cost in classical resources and the possibility of performing scalable simulations using quantum computers is currently an active field of research. In this work we provide the first complete characterization of the resource requirements for studying nuclear dynamics with the full Leading Order (LO) pionless EFT Hamiltonian in first quantization employing simulation strategies using both product formulas as well as Quantum Signal Processing. In particular, we show that time evolution of such an Hamiltonian can be performed with polynomial resources in the number of particles, and logarithmic resources in the number of single-particle basis states. This result provides an exponential improvement compared with previous work on the same Hamiltonian model in second quantization. We find that interesting simulations for low energy nuclear scattering could be achievable with tens of millions of $T$ gates and few hundred logical qubits suggesting that the study of simple nuclear reactions could be amenable for early fault tolerant quantum platforms.
\end{abstract}

\date{August 27, 2026}

\maketitle

\tableofcontents

\section{Introduction}

Accurate simulations of dynamical properties of nuclear systems are of central importance for the description of a large number of precesses in low-energy nuclear physics ranging from fusion reactions of light elements in stellar nucleosynthesis~\cite{RevModPhys.83.195,acharya2024solarfusioniiinew,articleQuantStar}, to precision experiments for fundamental physics tackling problems like neutrino properties~\cite{10.1093/ptep/ptv061,dunecollaboration2016longbaselineneutrinofacilitylbnf,huber2022snowmassneutrinofrontierreport,Alvarez_Ruso_2025} and neutrino-less double beta decay~\cite{PhysRevC.97.014606,Cirigliano_2022}. A prominent role in the field is played by ab-initio approaches to nuclear physics where a clear link is maintained between the low-energy Hamiltonians describing the interactions among nucleons and the underlying theory of Quantum Chromodynamics (QCD)~\cite{RevModPhys.81.1773} allowing for a sensible quantification of the theoretical uncertainties~\cite{Epelbaum2015,PhysRevLett.125.202702}. The past decade has witnessed a remarkable progress in the development of ab-initio techniques to study the structure of atomic nuclei and the static properties of nuclear matter with a large part of the nuclear chart now amenable to an ab-initio description~\cite{10.3389/fphy.2020.00379,10.3389/fphy.2023.1129094}. 

The ab-initio description of reactions involving nuclei has also seen a similar blossoming of computational techniques~\cite{Navrátil_2016,Johnson_2020,annurev:/content/journals/10.1146/annurev-nucl-102419-033316} with, for example, the description of quasi-elastic scattering in atomic nuclei having made important steps in recent years~\cite{PhysRevC.97.022502,PhysRevC.98.025501,PhysRevC.100.062501,PhysRevX.10.031068,PhysRevC.101.044612,PhysRevC.105.014002,PhysRevLett.127.072501,PhysRevC.109.044314,PhysRevC.110.064004}.
The complexity of the full reaction problem is however still hindering the availability of general, controllable and scalable simulation strategies for ab-initio nuclear dynamics in all the regimes relevant to experiments.

In recent years, the prospect of employing quantum computing to simulate nuclear systems in a scalable way has motivated an increasing number of research groups around the world to develop quantum algorithms to study both structure and dynamics of nuclear systems(see~\cite{cloët2019opportunitiesnuclearphysics,Klco_2022,PRXQuantum.4.027001,Ayral2023,beck2023quantuminformationsciencetechnology} for recent reviews). Given the great success of classical computations to study a variety of static properties of atomic nuclei, it is reasonable to expect that the most important near term impact of quantum simulations will be on the study of dynamical processes involving nucleons.

In this work we develop a quantum-simulation protocol for studying nuclear processes at small energies by considering a system of interacting nucleons in pionless effective field theory~\cite{VANKOLCK1999273,KAPLAN1998390}. This theory is appropriate for nuclear systems with typical momenta smaller than the pion mass, and has been successfully employed to describe both the structure and reactions of light nuclei (see~\cite{RevModPhys.92.025004} for a recent review). In this framework, nucleons are described as non-relativistic point particles with both spin and isospin quantum numbers while the strong force between nucleons is modeled by contact interactions and their derivatives. 
At leading order, and assuming Wigner's SU(4) symmetry~\cite{PhysRev.51.106}, nucleons interact through two and three-body interactions according to the following Hamiltonian \cite{Lee_2009}
\begin{equation}
    \begin{split}
        H &= T + V_2  + V_3 \\
        &= -\sum_{i=0}^{\eta-1} \frac{\hbar^2\nabla^2_i}{2\mu} + \frac{1}{2}\sum_{i=0}^{\eta-1}\sum_{j\ne i}^{\eta-1} C\delta (\Vec{r}_i - \Vec{r}_j) + \frac{1}{6}\sum_{i=0}^{\eta-1}\sum_{j\ne i}^{\eta-1}\sum_{k\ne j\ne i}^{\eta-1} G\delta (\Vec{r}_i - \Vec{r}_j)\delta (\Vec{r}_j - \Vec{r}_k)
    \end{split}
    \label{equation:Hamiltonian}
\end{equation}
where $T$ is the kinetic energy, $\mu$ the nucleon mass, $V_2$ the two-body potential, $V_3$ the three-body potential and we have denoted with $\eta$ the number of particles in the system. An explicit potential term that distinguishes singlet and triplet spin pairs, and thus breaks the $SU(4)$ invariance of the model down to $SU(2)$, can be added in the two-body potential. In this first work we do not explicitly include this term {\ but describe the general strategy to extend our construction in Appendix.~\ref{app:su4_break}. }

The model Hamiltonian in Eq.~\eqref{equation:Hamiltonian} requires a regularization of the contact terms, the standard approach is to introduce an ultra-violet (UV) cutoff to smear the delta functions. This can be done in continuous space using explicit regulators~\cite{RevModPhys.81.1773,MACHLEIDT20111} but in this work we choose instead to work on a spatial lattice as in Lattice EFT~\cite{Lee_2009,lahde2019nuclear}. In this case the role of this UV cutoff is played by the lattice spacing $a$. As a result of the regularization, the low energy constants $C$ and $G$ will acquire an explicit $a$ dependence and will need to be renormalized in order to describe low energy nuclear physics correctly. The single particle basis in the calculation is formed by taking a regular lattice in $d$ spatial dimensions and associating $4$ separate spin/isospin states for every site. We will consider a lattice with periodic boundary conditions and $M$ sites per direction, leading to a total number of single particle states given by $\Omega = 4\cdot M^{d}$. In order to study nuclear reactions we need to prepare the initial nuclear states and then perform time-evolution generated by the Hamiltonian in Eq.~\eqref{equation:Hamiltonian}, in this work we focus on the second step only. The total spatial volume of the lattice $V=(aM)^d$ needs to be sufficiently large to be able to extrapolate the simulated results to cross sections in the continuum. Notably, in the nuclear physics context the limit $V\to\infty$ is considered by fixing lattice spacing $a$ first and then taking $M\to\infty$. The UV cutoff dependence is then controlled by performing the procedure for multiple separate lattice spacings. Since the relevant nuclear Hamiltonians, including the one from Eq.~\eqref{equation:Hamiltonian}, conserve the total number of particles $\eta$ and for nuclear reactions $\eta$ is fixed by the process being studied, these calculations are then performed in the dilute limit $\eta/\Omega\ll1$. This is to be contrasted with physical simulations of bulk systems where instead the ratio $\eta/\Omega$ is kept fixed as the volume is increased.

The time evolution generated by the pionless Hamiltonian considered here has been already studied in detail using second quantization in Refs.~\cite{Roggero_2020,watson2023quantumalgorithmssimulatingnuclear}. These works found that the total gate cost to approximate the time-evolution operator $e^{itH}$ for a fixed total time $t$ and error $\epsilon$ scales at least linearly in the volume as $O(\Omega)$. In particular, using a second order product formula the authors of Ref.~\cite{watson2023quantumalgorithmssimulatingnuclear} have shown that the number of $T$-gates required for the simulation with a given number of particles $\eta$ scales as
\begin{equation}
    O\left(\frac{t^{3/2} \eta^{1/2} }{\epsilon^{1/2}} \Omega \log\left( \frac{t\eta \Omega}{\epsilon} \right) \right)\;,
    \label{equation:second_quantization_tr2}
\end{equation}
which is sublinear in $\eta$ and linear in $\Omega$. In addition, since in second quantization the number of qubits scales with the size of the single particle space, the required number of qubits scales as $O(\Omega)$, irrespective of the number of particles $\eta$ in the system. {\  We note that quantifying the computational cost of a quantum algorithm in terms of non-Clifford gates, like the $T$ gate, is a common choice for fault-tolerant compilation due to the obstructions imposed by the Eastin-Knill theorem~\cite{PhysRevLett.102.110502} which make the implementation of these operations particularly expensive using standard error correcting codes. The recent discovery of efficient protocols for $T$-gate synthesis like magic state cultivation~\cite{gidney2024magic,rosenfeld2025magic} have significantly reduced the cost of these operation. In this work we will use the $T$-gate complexity as a proxy for the total gate cost for two main reasons: first, this choice makes it possible to more easily compare with the available literature; second, we expect in practice that the number of Clifford gates required will scale similarly to the number of $T$-gates. }

The linear scaling with $\Omega$ poses an important challenge for simulations of nuclear processes in the early fault-tolerant era since, even for light nuclei with small $\eta$, spatial lattices with at least $M\approx8-10$ are typically required in order to obtain controllable results~\cite{LU2019134863,PhysRevLett.132.062501}. In second quantization this requirement translates into at least $2048$ qubits with a memory efficient encoding such as Jordan-Wigner~\cite{jw} and to $3072$ qubits with a local encoding such as the Verstraete-Cirac~\cite{Verstraete_2005} employed to obtain the gate count of Eq.~\eqref{equation:second_quantization_tr2} above~\cite{watson2023quantumalgorithmssimulatingnuclear}. These substantial memory requirements have motivated the use of a representation of the nuclear system in first quantization for the triton toy model introduced in Ref.~\cite{Roggero_2020} and employed in the first demonstrations on current quantum hardware~\cite{Roggero_2020,PhysRevD.105.074503,PhysRevD.111.034504}. Quantum simulation methods in first quantization have a long history in quantum chemistry~\cite{PhysRevLett.79.2586,Kassal_2008,Jones_2012,Babbush_2019,Su2021,georges2025quantum}. The first proposal of a quantum simulation of nuclear dynamics in first quantization was introduced recently in Ref.~\cite{weiss2025solvingreactiondynamicsquantum}. The total number of qubits required for the algorithm scales as $O(\eta\log(\Omega))$ while they provide an implementation of a single Trotter step in a pionless EFT without three-body interactions with a gate count scaling as $O(\eta^2\log(\Omega))$.  
In this work we provide the first complete characterization of the resource requirements for studying nuclear dynamics with the full Leading Order (LO) pionless EFT Hamiltonian employing simulation strategies using both product formulas~\cite{Suzuki1991,Lloyd1996} as well as Quantum Signal Processing~\cite{Low2017, Low_2019} and Generalized Quantum Signal Processing~\cite{motlagh2023generalized}. In particular we find that indeed the gate cost scales polynomially in $\eta$ and {\  poly}-logarithmically in $\Omega$. For example, a second order product formula implementation has a T-gate cost scaling as
\begin{equation}
\label{eq:2ndexample}
\widetilde{O}\left(\frac{t^{3/2}\eta^{3/2}}{\epsilon^{1/2}}\log\left(\frac{t\eta\Omega}{\epsilon}\right)\left(\eta^2+\log\left(\frac{t\Omega}{\epsilon}\right)\right)\right)\;,
\end{equation}
which is exponentially better in the volume $\Omega$ than the second quantization result in Eq.~\eqref{equation:second_quantization_tr2} at the expense of being polynomially worse in the number of particles $\eta$.

This is achieved by associating a qubit register of size $\log_2(\Omega)$ qubits to every particle out of which: $d\log_2(M)$ qubits are used to encode the position, and 2 qubits encoding the $4$ possible states of spin and isospin. In first quantization, the total wave function carries the statistics of particles, which means that it has to be initialized to be anti-symmetric with respect to the exchange of two fermions. Once the system is initialized in the right symmetry sector, since the Hamiltonian commutes with the permutation operator, the statistics will be preserved. The cost of antisymmetrizing the wavefunction is $O(\eta\log\eta \log M)$ (see e.g.~\cite{Berry_2018}), but it has to be done only once at the beginning of the algorithm. In this work we consider this antisymmetrization step to be included in the initial state preparation {\  which is not considered in our explicit gate counts. For completeness, we also provide in Sec.~\ref{sec:state_prep} some additional considerations on the expected cost of this procedure and summarize the scaling of three schemes proposed in the literature in Table~\ref{tab:state_prep}. The only algorithm maintaining a logarithmic scaling in $\Omega$ is the one from Ref.~\cite{PRXQuantum.6.020319} which relies on a basis of $N_B\ll\Omega$ localized function to express the single particle orbitals.}  A complete summary of the results obtained in this work is provided in Table~\ref{tab:comparision_results} where we quote the scaling of the T gate count $C$ and total number of qubits $Q$ for different time-evolution algorithms. We note in particular that the advantageous scaling from Eq.~\eqref{eq:2ndexample} above is maintained if one goes to higher order product formulas and for QSP.

\begin{table}[]
    \centering
\begin{tabular}{|cc|c|}
\hline
 \multicolumn{2}{|c|}{cost} & Ref. \\ \hline
 $C=O\left(\eta\Omega\right)$&$Q=O(\eta)$ & \cite{babbush2023quantum} \\ \hline
 $C=O\left(\eta^2\sqrt{\Omega}\right)$&$Q=O\left(\sqrt{\Omega}\right)$ & \cite{Rule2026recursivealgorithm} \\ \hline
 $C=O\left(\eta^2N_B^{3/2}\log(\Omega)\right)$&$ Q=O\left(\eta\log(\eta)+\sqrt{N_B}\right)$ & \cite{PRXQuantum.6.020319} \\ \hline
\end{tabular}
\caption{{\  The $T$-gate count $C$ and the number of qubits $Q$ required to prepare an anti-symmetric initial state in first quantization by a variety of quantum algorithms. The number of particles in the system are $\eta$ and the size of the single particle space $\Omega$ (proportional to the spatial volume). For the method of Ref.~\cite{PRXQuantum.6.020319} $N_B$ denotes the number of Gaussian basis functions used to represent the single particle orbitals.}}
\label{tab:state_prep}
\end{table}

\begin{table}[]
    \centering
    \small
    \setlength{\tabcolsep}{4pt}
\begin{tabular}{|c|c|c|c|}
\hline
 Encoding  & Method & cost & Ref. \\ \hline
2nd quant.  & PF(2) & $\begin{array}{cc}
   C=O\left(\frac{t^{3/2}}{\epsilon^{1/2}} \eta^{3/2}\Omega^2\log\left(t\Omega\eta/\epsilon\right)\right)   \\
    Q=\Omega
\end{array}$ & \cite{Roggero_2020} \\ \hline
2nd quant.  & PF(2) & $\begin{array}{cc}
C=O\left(\frac{t^{3/2}}{\epsilon^{1/2}} \eta^{1/2}\Omega \log\left( t\eta \Omega/\epsilon \right) \right)\\
Q=O(\Omega)
\end{array}$ & \cite{watson2023quantumalgorithmssimulatingnuclear} \\ \hline
2nd quant.  & PF(2p) & $\begin{array}{cc}
C=O\left(\frac{t^{1+1/(2p)}}{\epsilon^{1/(2p)}} \eta^{1/(2p)}\Omega \log\left( t\eta \Omega/\epsilon \right) \right)\\
Q=O(\Omega)
\end{array}$ & \cite{watson2023quantumalgorithmssimulatingnuclear} \\ \hline
2nd quant.  & QSP & $\begin{array}{cc}
C=\widetilde{O}\left(\left(\Omega t+\log\left(1/\epsilon\right)\right)\Omega\log\left(t\Omega/\epsilon\right)\right)\\
Q=O(\Omega)
\end{array}$ & \cite{Roggero_2020} \\ \hline\hline
1st quant.  & PF(2)& $\begin{array}{cc}
C=\widetilde{O}\left(\frac{t^{3/2}}{\epsilon^{1/2}} \eta^{3/2}\log\left(t\Omega\eta/\epsilon\right)\left(\eta^2+\log\left(\Omega\right)\right)\right)\\
Q=\widetilde{O}\left(\eta\log(\Omega)+\log^2(\Omega)+\log\left(t/\epsilon\right)\right)
\end{array}$ & (this work) \\ \hline
1st quant.  & PF(2p)& $\begin{array}{cc}
C=\widetilde{O}\left(\frac{t^{1+1/(2p)}}{\epsilon^{1/(2p)}} \eta^{1+1/(2p)}\log\left(t\Omega\eta/\epsilon\right)\left(\eta^2+\log\left(\Omega\right)\right)\right)\\
Q=\widetilde{O}\left(\eta\log(\Omega)+\log^2(\Omega)+\log\left(t/\epsilon\right)\right)
\end{array}$ & (this work) \\ \hline
1st quant.  & QSP/GQSP& $\begin{array}{cc}
C=O\left(\left(\eta t+\log\left(\frac{\eta t}{\epsilon}\right)\right)\log(\Omega) \left( \eta  + \log\left( \frac{\eta t \log(\Omega)}{\epsilon} \right) \right)\right)\\
Q=O\left(\eta\log(\Omega)+\log\left(t/\epsilon\right)\right)
\end{array}$ & (this work) \\ \hline
\end{tabular}
\caption{The $T$-gate count $C$ and the number of qubits $Q$ required by a variety of quantum algorithms for time evolution with the pionless EFT Hamiltonian Eq.~\eqref{equation:Hamiltonian}. The total time is $t$, the target error $\epsilon$, the number of particles $\eta$ and the size of the single particle space $\Omega$ (proportional to the spatial volume). In all of these we consider the number of spatial dimensions $d = O(1)$.}
\label{tab:comparision_results}
\end{table}

The organization of the paper is as follows: in Sec.~\ref{sec:firstq} we first introduce in more detail the representation of the Hamiltonian $H$ from Eq.~\ref{equation:Hamiltonian}, which is defined in the continuum, in our discrete single particle basis, we then discuss the implementation of the time evolution operator with product formulas in Sec.~\ref{sec:Exponential_of_Hamiltonian_terms} and Sec.~\ref{ssec:pf} and present our construction of a QSP based simulation in Sec.~\ref{sec:qspimpl}. {\  We provide a discussion of the cost of state preparation for initial states relevant to study nuclear reactions in Sec.~\ref{sec:state_prep}.} We then present some numerical result employing realistic parameters in Sec.~\ref{sec:results} and conclude in Sec.~\ref{sec:conclusions}. Many of the more technical details and proofs are left for the Appendices. {\  In order to facilitate the use of our results, we also provide code implementing the cost estimates in Python in Ref.~\cite{spagnoli_2026_20041384}.}

\section{First quantization algorithms}
\label{sec:firstq}

In this section, we will provide a detailed derivation of algorithms that approximate the time evolution operator for the Hamiltonian defined in Eq.~\ref{equation:Hamiltonian} in a discrete first quantization encoding. As discussed in the introduction, the single particle basis we employ here corresponds to considering, for the spatial part, a finite lattice with spacing $a$ and using periodic boundary conditions while the spin/isospin part is encoded into two qubits. 
For example, in a system with $d=3$ spatial dimensions, the single particle basis is spanned by states of the form
\begin{equation}
\ket{\vec{r},s,t} = \ket{r_x}\otimes\ket{r_y}\otimes\ket{r_z}\otimes\ket{s}\otimes\ket{t}\;,
\end{equation}
where $\ket{r_k}$, for $k=x,y,z$ (in $d=3$), encode the coordinates of a nucleon in register of $m=\lceil\log_2(M)\rceil$ qubits, while the spin $\ket{s}$ and isospin $\ket{t}$ are both encoded into registers with one qubit each. From here on we will assume that the number of lattice sites per spatial dimension $M$ is chosen as a power of $2$ and remove the explicit ceiling in the definition of $m$.
For a system with $\eta$ nucleon we therefore require a register with
\begin{equation}
Q = (dm+2)\eta\;,
\end{equation}
qubits to encode the state of the entire system.

We can now proceed to encode the full Hamiltonian which was presented in Eq.~\eqref{equation:Hamiltonian} for a continuous system on our lattice basis. As mentioned above in the introduction, we employ here Periodic Boundary Conditions (PBC) for our spatial lattice which implies that the particle momenta are quantized in units of
\begin{equation}
\frac{2\pi}{V^{1/d}} = \frac{2\pi}{aM} = \frac{2\pi}{a2^m}\;,
\end{equation}
where $a$ is the lattice spacing and $V$ the total volume. For a single particle, the kinetic energy is therefore mapped to a diagonal operator in momentum space
\begin{equation}
-\frac{\hbar^2\nabla^2}{{2\mu}} \quad\to\quad \frac{\hbar^2}{2\mu a^2} \left(\text{QFT}^{\otimes d}\right)^{\dagger}\left( \sum_{p} E_k(\Vec{p}) \ket{\Vec{p}}\bra{\Vec{p}} \right) \left(\text{QFT}^{\otimes d}\right)\;,
\end{equation}
where $\left(\text{QFT}^{\otimes d}\right)$ and $\left(\text{QFT}^{\otimes d}\right)^\dagger$ denote the $d$ dimensional Quantum Fourier Transform (QFT) and its inverse while $E_k(\Vec{p})$ is the energy of a plane wave state with momentum $\Vec{p}$. Since the standard QFT on $m$ qubits maps states in the interval $[0,2^{m}-1]$ into states on the same interval, but we want to describe momenta with components in $[-2^m/2,2^m/2-1]$, we use the periodicity to introduce for each  direction $w=x,y,z$ the components
\begin{equation}
    q_w = \begin{cases}
        p_w  &0\le p_w \le 2^{m-1}-1 \\
        p_w-2^m  &2^{m-1}\le p_w \le 2^m-1 \\
    \end{cases}   
\end{equation}
such that the full kinetic energy can be written as
\begin{equation}
\label{eq:Tdiscrete}
    T = K \left(\text{QFT}^{\otimes d\eta}\right)^{\dagger} \sum_{i=0}^{\eta -1} \sum_{w=0}^{d-1} \sum_{p=0}^{2^{m}-1}  q_w^2 \Pi_{w,i}(p) \left(\text{QFT}^{\otimes d\eta}\right)\;,
\end{equation}
where the overall numerical coefficient $K$ is given by
\begin{equation}
\label{eq:kkin}
K=\frac{\hbar^2}{2\mu a^2}\left(\frac{2\pi}{2^m}\right)^2\;,
\end{equation}
the full $QFT$ now act on all $d\times\eta$ registers of the system and we denoted the projector in the state $\ket{p}$ on the $w$ register for particle $i$ with $\Pi_{w,i}(p)$. Note that the contribution to the energy given by $q_w^2$ can be computed by squaring the right-most $m-1$ bits on the momentum register.

The potential energy contribution in our basis can then be written as
\begin{equation}
\label{eq:Vdiscrete}
\begin{split}
    V &= V_2 + V_3\\
    V_2 &= \frac{C}{2}\sum_{i=0}^{\eta-1}\sum_{j\ne i}^{\eta-1} \sum_{\Vec{r}_i,\Vec{r}_j} \delta_{\Vec{r}_i,\Vec{r}_j} \Pi_{i}(\Vec{r}_i)\Pi_{j}(\Vec{r}_j) \\
    V_3 &=\frac{G}{6}\sum_{i=0}^{\eta-1}\sum_{j\ne i}^{\eta-1}\sum_{k\ne j\ne i}^{\eta-1} 
    \sum_{\Vec{r}_i,\Vec{r}_j,\Vec{r}_k}\delta_{\Vec{r}_i,\Vec{r}_j}\delta_{\Vec{r}_i,\Vec{r}_k} \Pi_{i}(\Vec{r}_i)\Pi_{j}(\Vec{r}_j)\Pi_{k}(\Vec{r}_k)\;.
\end{split}
\end{equation}
In the expressions above, the sum over $\Vec{r}_a$, for $a=i,j,k$, run over all the $2^{dm}$ momentum states and we used the following compact notation for projector on particle states and Kronecker delta for vectors
\begin{equation}
\Pi_{i}(\Vec{x}) = \prod_{w=0}^{d-1}\Pi_{w,i}(x_w)\quad\quad\quad\delta_{\Vec{x},\Vec{y}}=\prod_{w=0}^{d-1}\delta_{x_w,y_w}\;.
\end{equation}

{\  These potential operators and the kinetic term are diagonal matrices in the position and momentum basis respectively.} This property is useful for performing the time evolution operator in first quantization on quantum computers, as was realized early on (see e.g.~\cite{Kassal_2008}) and it is also the main motivation to use these two basis sets. In general it is also possible to only work in coordinate space and approximate differential operators using finite differences. This is the typical setup used in classical Lattice EFT calculations~\cite{Lee_2009} and in the first implementations of pionless EFT on quantum computers~\cite{Roggero_2020,watson2023quantumalgorithmssimulatingnuclear}. Efficient schemes to employ finite differences for quantum simulations in first quantization have already appeared in the past (see e.g.~\cite{Kivlichan_2017}) but for this first work we don't explore the possibility further.

\subsection{Exponential of Hamiltonian terms}
\label{sec:Exponential_of_Hamiltonian_terms}

In this section we provide results for the direct approximation of the exponentials of both the kinetic energy $T$ and the total potential energy $V=V_2+V_3$. In particular our aim is to design unitaries $U_T(t)$ and $U_V(t)$ such that
\begin{equation}
    \norm{e^{-iTt} - U_T(t)} \le \epsilon_T \quad\quad
    \norm{e^{-iVt} - U_{V}(t)} \le \epsilon_V \;.
\end{equation}

We note in passing that, since $V_2$ and $V_3$ commute with each other, their exponentials could be implemented separately even though, as we will see below, there might be more efficient strategies.

We begin with the following result for the implementation of the unitary $U_T$.
\begin{lemma}[Exponential of Kinetic Energy]
    Consider the kinetic energy term of Eq.~\eqref{eq:Tdiscrete} for a system with $\eta$ particles on a $d$ dimensional lattice with $M=2^{m}$ sites per direction. Let $\lambda_T \ge \norm{T}$ an upper bound on its spectral norm. For any $t\in \mathbb{R}$ and $\epsilon > 0$ {\  we can construct a quantum circuit implementing} a unitary operator $U_T(t)$ such that $\norm{U_T(t) - e^{-iTt}} < \epsilon$ with a number of T gates 
    given by:
    \begin{equation}
    \begin{split}
        T_T(\epsilon, t) &= d\eta \left( 2T_{SQU}(m-1) + T_{DIAG}\left(\frac{\epsilon}{3d\eta}, t\right) + 2T_{QFT} \left(m,\frac{\epsilon}{3d\eta}\right) \right) \\
        &= \widetilde{O}\left[\eta m\left( m + \log\left( \frac{t\eta\lambda_T}{\epsilon}\right) \right) \right]\\
    \end{split}
    \end{equation}
    and a total number of ancilla qubits given by
    \begin{equation}
    b_{QFT}+b_{DIAG}+\max \left( 2b_{QFT}-1,b_{DIAG}, m(m-1) \right)\;,
    \end{equation}
    where we have defined
    \begin{equation}
    \begin{split}
b_{DIAG}&=\left\lceil\log_2\left(\frac{6d\eta t\lambda_T}{\epsilon}\log_2\left(\frac{6d\eta t\lambda_T}{\epsilon}\right)\right)\right\rceil\\
b_{QFT}&=\min\left(m-1,\left\lceil\log_2\left(\frac{3d\eta m}{\epsilon}\right)\right\rceil\right)+1\;.
    \end{split}
    \end{equation}
    The explicit definition of the various gate counts are in Lemma~\ref{lemma:phkick} and Theorems \ref{theorem:multiplication_better} and \ref{theorem:QFT} in the Appendix.
    \label{theorem:kinetic_energy_operator}
\end{lemma}

\begin{proof}
    As shown in Eq.~\eqref{eq:Tdiscrete} in Fourier space $T$ is a diagonal operator with $K\norm{q}^2$ on the diagonal. That matrix can be implemented by using Lemma~\ref{lemma:phkick}, where the oracle $O_{b_{\lambda}}$ is the operator such that $O_{b_{\lambda}}\ket{p}\ket{0}_{b_{\lambda}} = \ket{p}\ket{p^2}_{b_{\lambda}}$. This means that the oracle $O_{b_{\lambda}}$ is nothing but a circuit that performs the multiplication of a bitstring of size $(m-1)$ with itself, and its cost is given by $T_{SQU}(m-1)$ in Theorem \ref{theorem:multiplication_better} and uses a total of $(m-1)(m-2)$ ancilla qubits. Since this operation is performed for one direction at a time, the result can be stored into a register with $b_\lambda=2m-2$ qubits.

    From Lemma~\ref{lemma:phkick} it follows that, including the cost of the $2$ QFTs and repeating the circuit for every particle in the system and for every spatial dimension, we need a total number of $T$ gates equal to 
    \begin{equation}
        T_T(\epsilon, t) = d\eta \left( 2T_{SQU}(m-1) + T_{DIAG}(\epsilon_1, t) + 2T_{QFT} (m,\epsilon_2) \right)
    \end{equation}
    where $T_{QFT}(m,\epsilon)$ is given in Theorem~\ref{theorem:QFT}, and $d\eta(\epsilon_1 + 2\epsilon_2) = \epsilon$. For simplicity we can choose $\epsilon_1 = \epsilon_2 = \epsilon/(3d\eta)$ even if, in general, one should take the minimum cost evaluated over all possible choices of $\epsilon_1$ and $\epsilon_2$. With this choice, and since we are computing the eigenvalues exactly, we can use the expression in Equation \ref{eq:b_exact} from Lemma~\ref{lemma:phkick} to estimate the size of the phase register leading to
    \begin{equation}
    b_{DIAG}=\left\lceil\log_2\left(\frac{6d\eta t\lambda_T}{\epsilon}\log_2\left(\frac{6d\eta t\lambda_T}{\epsilon}\right)\right)\right\rceil\;,
    \end{equation}
    together with a number of $T$ gates given by
    \begin{equation}
    T_{DIAG}\left(\frac{\epsilon}{3d\eta}, t\right)=8b\min\left(\left\lceil\frac{2m-1}{2}\right\rceil,w_H\right)\;,
    \end{equation}
    where $w_H$ is the Hamming weight of the constant $\gamma=t\lambda_V/(2\pi)$ approximated with $2m$ bits. For implementing the $QFT$ instead we need a phase register of size
    \begin{equation}
    b_{QFT}=\min\left(m-1,\left\lceil\log_2\left(\frac{3d\eta m}{\epsilon}\right)\right\rceil\right)+1\;.
    \end{equation}
    
    As for the number of qubits we need
    \begin{enumerate}
        \item $(m-1)(m-2)+b_\lambda=m(m-1)$ qubits for squaring and storing the eigenvalue
        \item $b_{DIAG}$ qubits to apply the diagonal unitary using Lemma~\ref{lemma:phkick}
        \item $2b_{QFT}-1$ qubits to implement the approximate QFT using Theorem~\ref{theorem:QFT}
    \end{enumerate}
    plus two phase registers of sizes $b_{DIAG}$ and $b_{QFT}$ respectively. Since we can reuse the ancilla space for the different steps the overall number of ancilla qubits needed to implement this circuit is
    \begin{equation}
        b_{QFT}+b_{DIAG}+\max \left( 2b_{QFT}-1,b_{DIAG}, m(m-1) \right)
    \end{equation}
    where we have also included the space needed for the two phase registers, one for the $QFT$ and one for the diagonal gate.
\end{proof}

A simple expression for a suitable upperbound on the norm of $T$ is for instance
    \begin{equation}
    \label{eq:lambdaT}
        \norm{T} \le \lambda_T = dK\eta2^{2m-2}=\frac{\hbar^2}{2\mu a^2}d\eta\pi^2\;.
    \end{equation}
For this choice of $\lambda_T$, and using the fact that $m=\log_2(M)=O(\log(\Omega))$, the $T$ gate cost of evolving under the kinetic term scales as
\begin{equation}
T_T(\epsilon, t) =  \widetilde{O}\left[\eta \log(\Omega) \log\left(\frac{t\Omega\eta}{\epsilon}\right) \right]\;.
\end{equation}

We now turn our attention to the potential energy term. The first result we obtain is the following. 
\begin{lemma}[Exponential of  Potential Energy]
    Consider the two and three-body potential terms of Eq.~\eqref{eq:Vdiscrete} for a system with $\eta$ particles on a $d$ dimensional lattice with $M=2^{m}$ sites per direction. For any $t\in \mathbb{R}$ and $\epsilon > 0$ {\  we can construct a quantum circuit implementing} $U_{V_2}(t)$ and $U_{V_3}(t)$ such that $\norm{U_{V_2}(t) - e^{-iV_2t}}<\epsilon_{V_2}$ and $\norm{U_{V_3}(t)-e^{-iV_3t}}<\epsilon_{V_3}$ with a number of $T$ gates given by:
    \begin{equation}
        \begin{split}
            T_{V_2}(\epsilon_{V_2}) &= \frac{\eta(\eta-1)}{2} T_{MCRZ} \left(\frac{2\epsilon_{V_2}}{\eta(\eta-1)}, dm-1 \right) \\
            T_{V_3}(\epsilon_{V_3}) &= \frac{\eta(\eta-1)(\eta-2)}{6} T_{MCRZ} \left( \frac{6\epsilon_{V_3}}{\eta(\eta-1)(\eta-2)}, 2dm-1 \right) \\
        \end{split}
    \end{equation}
    where $T_{MCRZ} (\epsilon)$ is given in Theorem~\ref{theorem:MCRZ}. The number of ancilla qubits needed is equal to the number of ancilla needed to implement $MCRZ$ from Theorem~\ref{theorem:MCRZ}.
    \label{theorem:Trotter_potential_energy}
\end{lemma}

\begin{proof}
    The action of the two-body potential $V_2$ and the three-body potential $V_3$ is the following:
    \begin{equation}
    \begin{split}
        e^{-iV_2t}\ket{\Vec{r}_i, \Vec{r}_j} &= e^{-iCt \delta_{\Vec{r}_i,\Vec{r}_j}}\ket{\Vec{r}_i, \Vec{r}_j} \\
        e^{-iV_3t}\ket{\Vec{r}_i, \Vec{r}_j, \Vec{r}_k} &= e^{-iGt \delta_{\Vec{r}_i,\Vec{r}_j}\delta_{\Vec{r}_j,\Vec{r}_k}}\ket{\Vec{r}_i, \Vec{r}_j, \Vec{r}_k} \\
    \end{split}
    \end{equation}
    Let us start with the $2$-body potential for simplicity. We have to apply a phase $e^{-iC_0t}$ if and only if the two particles are in the same position. This can be done with a $Z$ rotation 
    \begin{equation}
    R_Z(2Ct) = \begin{pmatrix}
        1 & 0 \\
        0 & e^{-iCt} \\
    \end{pmatrix}
    \end{equation}
    that has to be applied on a qubit which is $1$ if and only if the two positions are equal. This can be easily done summing the two position registers and then using a multi-controlled $R_Z$ gate. We show an explicit construction for this circuit in Fig.~\ref{fig:pair_exp_circuit} for a simple case where $m=1$ and $d=3$, so that every particle position is represented by registers composed by $3$ qubits, one for each spatial dimension.
   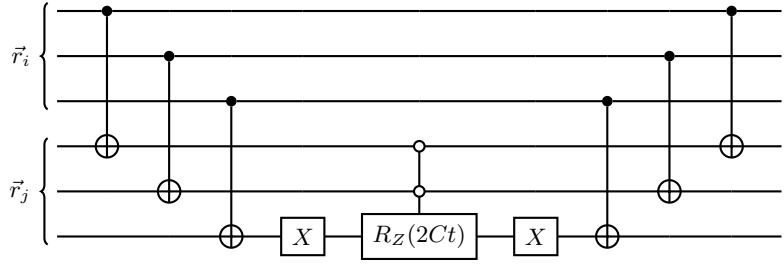
\begin{figure}
\centering
    \begin{quantikz}[row sep={0.6cm,between origins}]
    \lstick[3]{$\Vec{r}_i$} & \ctrl{3} & \qw     & \qw    & \qw    & \qw            & \qw    & \qw    & \qw & \ctrl{3} & \qw \\
    & \qw                              & \ctrl{3}& \qw    &  \qw   &\qw             & \qw    & \qw    & \ctrl{3} & \qw & \qw \\
    & \qw                              & \qw     &\ctrl{3}&  \qw   &\qw             & \qw    &\ctrl{3}& \qw & \qw & \qw \\
    \lstick[3]{$\Vec{r}_j$}& \targ{}   & \qw     & \qw    & \qw    & \octrl{1}      & \qw    & \qw    & \qw & \targ{} & \qw \\
    & \qw                              & \targ{} & \qw    & \qw    & \octrl{1}      & \qw    & \qw    & \targ{} & \qw & \qw \\
    & \qw                              & \qw     & \targ{}&\gate{X}&\gate{R_Z(2Ct)}&\gate{X}& \targ{}& \qw & \qw & \qw \\
    \end{quantikz}
    \caption{Circuit implementation of the exponential of a single two body term for $m=1$ and $d=3$.}
    \label{fig:pair_exp_circuit}
\end{figure}
    
    The summation of two $dm$-qubits registers will require $2dm$ CNOTs, while the cost of a $(dm-1)$-controlled $Z$-rotation is $T_{MCRZ}(\epsilon, dm-1)$, given in Theorem \ref{theorem:MCRZ}. Since we have $\eta(\eta-1)/2$ pairs, we will need to apply this circuit for every one of them, thus the total number of $T$-gates required to apply $V_2$ is
    \begin{equation}
        T_{V_2}(\epsilon_{V_2}) = \frac{\eta(\eta-1)}{2} T_{MCRZ} \left(\frac{2\epsilon_{V_2}}{\eta(\eta-1)}, dm-1 \right)
    \end{equation}
    For the $3$-body potential, we can generalize the $V_2$ construction. This time we need to check $\Vec{r}_i=\Vec{r}_j$ and also $\Vec{r}_j=\Vec{r}_k$, which means that we will need a $C^{2dm-1}R_Z(\theta)$ gate for every triplet of particles (which are in total $\eta(\eta-1)(\eta-2)/6$). This means that the number of $T$-gates needed to implement $V_3$ are:
    \begin{equation}
        T_{V_3}(\epsilon_{V_3}) = \frac{\eta(\eta-1)(\eta-2)}{6} T_{MCRZ} \left( \frac{6\epsilon_{V_3}}{\eta(\eta-1)(\eta-2)}, 2dm-1 \right)
    \end{equation}
\end{proof}

As mentioned above in the beginning of this section, it is actually advantageous to combine the two potential terms together. In addition, we can apply the $Z$ rotations on ancillas instead saving a factor of two in the rotation count. Putting all together we have the following result

\begin{lemma}[Compact Exponential of  Potential Energy]
    Consider the full potential $V$ from Eq.~\eqref{eq:Vdiscrete} for a system with $\eta$ particles on a $d$ dimensional lattice with $M=2^{m}$ sites per direction. For any $t\in \mathbb{R}$ and $\epsilon > 0$ {\  we can construct a quantum circuit implementing} a unitary $U_{V}(t)$ such that $\norm{U_{V}(t) - e^{-iVt}}<\epsilon$ with a number of $T$ gates given by:
    \begin{equation}
    \begin{split}
        T_V(\epsilon) &= 4\frac{\eta(\eta-1)(\eta-2)}{3} + \Gamma \left( 2T_{MCX}(dm) + T_{ROT}\left(\frac{\epsilon}{\Gamma} \right) \right) \\
        &= O\left[ \eta^3 \left( m + \log\left( \frac{\eta}{\epsilon} \right) \right) \right]
    \end{split}
    \end{equation}
    where $\Gamma$ is defined as follows:
    \begin{equation}
        \Gamma = \frac{\eta(\eta-1)(\eta-2)}{6} + \frac{\eta(\eta - 1)}{2}\;,
    \end{equation}
    $T_{MCX} $ and $T_{ROT}$ are given in Theorems~\ref{theorem:Z_rotations} and \ref{theorem:MCX}, and a number of ancilla qubits equal to $dm+1$.
    \label{theorem:complact_Trotter_potential_energy}
\end{lemma}

\begin{proof}
    Consider the implementation of $V_2$ and $V_3$ described in Lemma~\ref{theorem:Trotter_potential_energy}. Instead of implementing them separately, we can merge the two implementation, and apply $V_2$ for free while applying a modified version of $V_3$. Consider $3$ particles, on which we want to apply the $3$-body potential. In Lemma~\ref{theorem:Trotter_potential_energy}, after doing the proper summation between registers corresponding to particle positions, we applied a $(2dm-1)$-controlled $Z$-rotation, which is equivalent to evaluate $\delta_{\Vec{r}_i,\Vec{r}_j}\delta_{\Vec{r}_i,\Vec{r}_k}$ directly. Now, consider the variation where instead we save $\delta_{\Vec{r}_i,\Vec{r}_j}$ and $\delta_{\Vec{r}_i,\Vec{r}_k}$ individually in two ancilla qubits. We can then use the circuit shown in Fig.~\ref{fig:pair_exp_circuit_anc} to perform the evolution corresponding to the first delta term $\delta_{\Vec{r}_i,\Vec{r}_j}$.
   \begin{figure}
\centering
    \begin{quantikz}[row sep={0.6cm,between origins}]
    \lstick[3]{$\Vec{r}_i$} & \ctrl{3} & \qw & \qw & \qw & \qw & \qw & \qw & \qw & \ctrl{3} & \qw & \qw \\
    & \qw & \ctrl{3} & \qw & \qw & \qw & \qw & \qw & \ctrl{3} & \qw & \qw & \qw \\
    & \qw & \qw & \ctrl{3} & \qw & \qw & \qw & \ctrl{3} & \qw & \qw & \qw & \qw \\
    \lstick[3]{$\Vec{r}_j$} & \targ{} & \qw & \qw & \octrl{1} & \qw & \octrl{1} & \qw & \qw & \targ{} & \qw & \qw \\
    & \qw & \targ{} & \qw & \octrl{1} & \qw & \octrl{1} & \qw & \targ{} & \qw & \qw & \qw \\
    & \qw & \qw & \targ{} & \octrl{1} & \qw & \octrl{1} & \targ{} & \qw & \qw & \qw & \qw \\
    && \lstick{$\ket{0}$} & \qw & \targ{} & \gate{R_Z(2Ct)} & \targ{} & \qw & \qw \\
    \end{quantikz}
    \caption{Alternative implementation of the exponential of a single two body term for $m=1$ and $d=3$.}
    \label{fig:pair_exp_circuit_anc}
\end{figure}
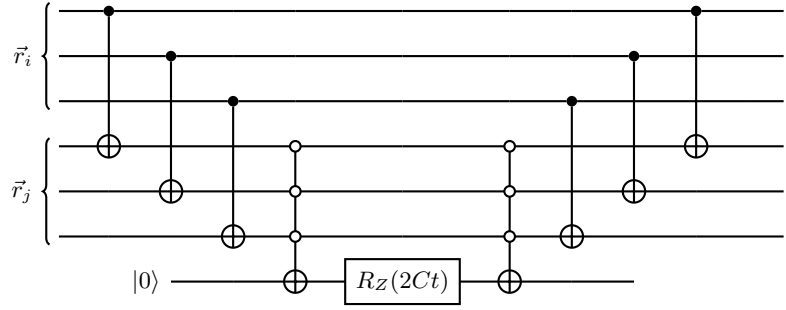

    This will be repeated for every pair $(i,j)$ so that we can apply the two-body potential with the rotation on the ancilla qubit. Then, we can add a subroutine to this circuit to evaluate the three-body potential. An example of the full circuit, for $m=1$ and $d=3$, is shown in Fig~\ref{fig:triple_exp_circuit_anc}.
   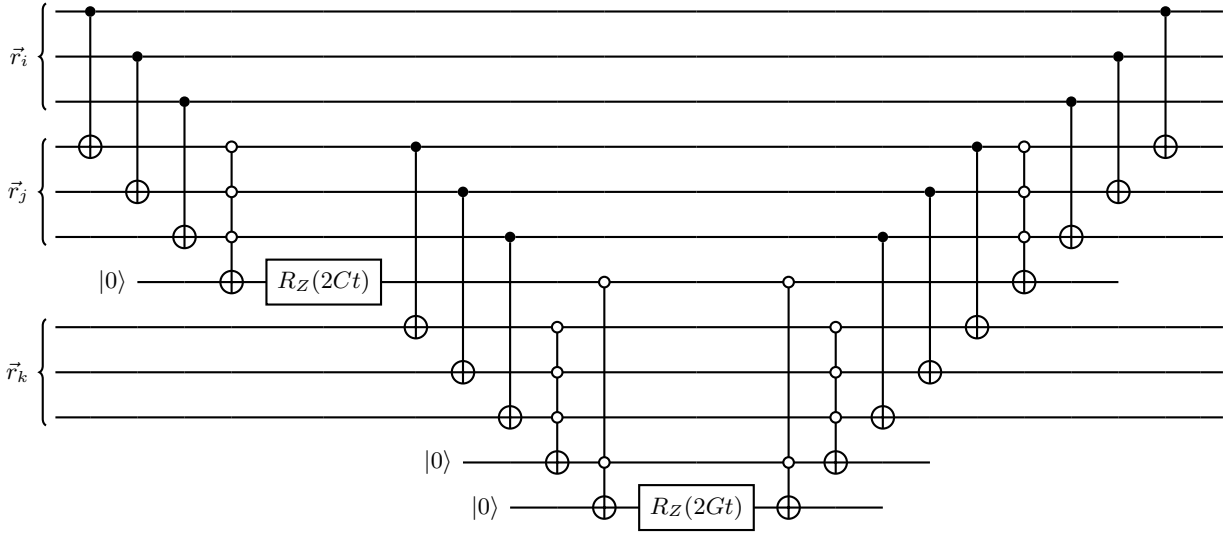
\begin{figure}
\centering
    \begin{quantikz}[row sep={0.6cm,between origins},column sep={0.2cm}]
    \lstick[3]{$\Vec{r}_i$} &\ctrl{3}& \qw    & \qw    & \qw     &\qw&\qw     & \qw    & \qw    & \qw & \qw & \qw & \qw & \qw & \qw & \qw & \qw & \qw & \qw & \qw & \ctrl{3} & \qw & \qw \\
    & \qw                            &\ctrl{3}& \qw    & \qw     &\qw& \qw    & \qw    & \qw    & \qw & \qw & \qw & \qw & \qw & \qw & \qw & \qw & \qw & \qw & \ctrl{3} & \qw & \qw & \qw \\
    & \qw                            & \qw    &\ctrl{3}& \qw     &\qw& \qw    & \qw    &\qw     & \qw & \qw & \qw & \qw & \qw & \qw & \qw & \qw & \qw & \ctrl{3} & \qw & \qw & \qw & \qw \\
    \lstick[3]{$\Vec{r}_j$} & \targ{}& \qw    & \qw    &\octrl{1}&\qw&\ctrl{4}& \qw    & \qw    & \qw & \qw & \qw & \qw & \qw & \qw & \qw & \ctrl{4} & \octrl{1} & \qw & \qw & \targ{} & \qw & \qw \\
    & \qw                            & \targ{}& \qw    &\octrl{1}&\qw& \qw    &\ctrl{4}& \qw    & \qw & \qw & \qw & \qw & \qw & \qw & \ctrl{4} & \qw & \octrl{1} & \qw & \targ{} & \qw & \qw & \qw \\
    & \qw                            & \qw    & \targ{}&\octrl{1}&\qw& \qw    & \qw    &\ctrl{4}& \qw & \qw & \qw & \qw & \qw & \ctrl{4} & \qw & \qw & \octrl{1} & \targ{} & \qw & \qw & \qw & \qw \\
    && \lstick{$\ket{0}$} & \qw & \targ{} & \gate{R_Z(2Ct)} & \qw & \qw & \qw & \qw & \octrl{4} & \qw & \octrl{4} & \qw & \qw & \qw & \qw & \targ{} & \qw & \qw \\
    \lstick[3]{$\Vec{r}_k$} & \qw    & \qw    & \qw    & \qw     &\qw& \targ{}& \qw    & \qw    & \octrl{1} & \qw & \qw & \qw & \octrl{1} & \qw & \qw & \targ{} & \qw & \qw & \qw & \qw & \qw & \qw \\
    & \qw                            & \qw    & \qw    & \qw     &\qw& \qw    & \targ{}& \qw    & \octrl{1} & \qw & \qw & \qw & \octrl{1} & \qw & \targ{} & \qw & \qw & \qw & \qw & \qw & \qw & \qw \\
    & \qw                            & \qw    & \qw    & \qw     &\qw& \qw    & \qw    &\targ{} & \octrl{1} & \qw & \qw & \qw & \octrl{1} & \targ{} & \qw & \qw & \qw & \qw & \qw & \qw & \qw & \qw \\
    &&&&&&& \lstick{$\ket{0}$} & \qw & \targ{} & \octrl{1} & \qw & \octrl{1} & \targ{} & \qw & \qw \\
    &&&&&&&& \lstick{$\ket{0}$}& \qw & \targ{} & \gate{R_Z(2Gt)} & \targ{} & \qw & \qw \\
    \end{quantikz}
    \caption{Circuit implementation of the exponential of a single two body term, for the pair $(i,j)$, and a single three-body term, for the triple $(i,j,k)$, for $m=1$ and $d=3$.}
    \label{fig:triple_exp_circuit_anc}
\end{figure}
    With this implementation we can apply the three-body potential with, as in the two body case, a single rotation on the ancilla qubit. To gain intuition about the full scheme employed here, we first note that the diagonal matrix element of the two- and tree-body potentials can be written as follows
    \begin{equation}
        V_2(\Vec{r}_0,\dots,\Vec{r}_{\eta-1}) = C\sum_{i<j}^{\eta-1} \delta_{\Vec{r}_i,\Vec{r}_j} 
    \end{equation}
    \begin{equation}
        V_3(\Vec{r}_0,\dots,\Vec{r}_{\eta-1}) = G \sum_{i<j<k}^{\eta-1}\delta_{\Vec{r}_i,\Vec{r}_j}\delta_{\Vec{r}_j,\Vec{r}_k} = G \sum_{i<j}^{\eta-1}\delta_{\Vec{r}_i,\Vec{r}_j}\left(\sum_{k>j}^{\eta-1}\delta_{\Vec{r}_j,\Vec{r}_k}\right)\;.
    \end{equation}
    The circuit in Fig.~\ref{fig:pair_exp_circuit_anc} is used to apply the evolution corresponding to $\delta_{\Vec{x}_i,\Vec{x}_j}$ and by repeating it for every pair $i,j$ with $i<j$ we compute the evolution undet the full sum $\sum_{i=0}^{\eta-1} \sum_{j>i}^{\eta-1} \delta_{\Vec{r}_i,\Vec{r}_j}$ appearing in both $V_2$ and $V_3$. In order to deal with the full three-body term, for every one of those circuits we need then to also compute the inner sum over $k$. For a fixed $k$, the inner delta is computed by the portion of circuit we added when going from Fig.~\ref{fig:pair_exp_circuit_anc} to Fig.~\ref{fig:triple_exp_circuit_anc}. To do the sum over $k$ we can repeat that inner portion of the circuit but we do not need to compute again the delta between $i,j$ every time.

    Schematically, the algorithm is as follows:
    \begin{enumerate}
        \item for each pair $(i,j)$ store $\ket{1}$ in an ancilla if $\Vec{r}_i=\Vec{r}_j$. This can be done, as sown in Fig.~\ref{fig:pair_exp_circuit_anc} above, using $dm$ CNOT gates and a single $C^{dm}X$ gate with the ancilla as target
        \item apply a rotation gate $R_Z(2Ct)$ to the ancilla to apply the two-body potential
        \item for each $k>j$ do the following:
        \begin{itemize}
            \item store  $\ket{1}$ in a new ancilla if $\Vec{r}_j=\Vec{r}_k$ using $dm$ CNOT gates and one $C^{dm}X$ gate
            \item apply the rotation for the three body potential. To do that, we can either apply a controlled-rotation between the ancilla storing $\delta_{\Vec{r}_i,\Vec{x}_j}$ and the one storing $\delta_{\Vec{r}_j,\Vec{r}_k}$, or apply a Toffoli gate controlled by the two ancillae and a simple rotation on a third ancilla, target of the Toffoli gate that will have to be uncomputed
            \item uncompute the state of the second ancilla using $dm$ CNOT gates and one $C^{dm}X$ gate
        \end{itemize}
    \end{enumerate}

    The total cost of this circuit will be the cost of the circuit in Fig.~\ref{fig:pair_exp_circuit_anc} (used to store $\delta_{\Vec{r}_i,\Vec{r}_j}$ and apply the two-body potential) times the number of pairs in the system, plus the inner portion of the circuit in Fig.\ref{fig:triple_exp_circuit_anc} (used to store $\delta_{\Vec{r}_j,\Vec{r}_k}$ and apply the three-body potential) times the number of triples in the system:
    \begin{equation}
        T_V(\epsilon_V) = 8\frac{\eta(\eta-1)(\eta-2)}{6} + \Gamma \left( 2T_{MCX}(dm) + T_{ROT}\left(\frac{\epsilon_V}{\Gamma} \right) \right) 
    \end{equation}
    where we used the strategy from Ref.~\cite{Gidney_2018} to implement each Toffoli with 4 $T$ gates and we introduced
    \begin{equation}
        \Gamma = \frac{\eta(\eta-1)(\eta-2)}{6} + \frac{\eta(\eta - 1)}{2}\;.
    \end{equation}
    As for the number of ancilla qubits, we will need $3$ ancilla plus the ones required to apply the multi-controlled $X$ operation. Considering that we can use the same ancilla many times, the maximum number of {\  ancillae} that we need at the same time is $dm+1$.
    
\end{proof}

\subsection{Product formula implementation}
\label{ssec:pf}

Using the implementation of the unitaries $U_T(t)$ and $U_V(t)$ in the previous subsection, we are now ready to implement the full time evolution using product formulas. In order to bound the errors in the product formulas for a fixed number of particles we use fermionic semi-norms. {\  This idea is not new, for instance it was called {\it physical norm} in Ref.~\cite{Roggero_2020} that studied the current nuclear model in second-quantization, while} in Ref.~\cite{Su_2021seminorm} the fermionic semi-norm was formally defined by considering projectors onto the particle-preserving subspace in second quantization. In Appendix \ref{appendix:norm} we define an equivalent fermionic semi-norm in first quantization. In this setting operators trivially conserve the number of particles, while they do not carry the fermionic statistics since that is encoded instead in the wave-function. The equivalence with the second quantization formulation is then done by considering projectors onto the anti-symmetric subspace. Here we briefly show the calculation for the commutator of the first-order product formula using this construction, while in Appendix \ref{appendix:norm} we provide the full definition of this new semi-norm, followed by the calculation of the various upper bounds that will be used in all the product formulas. 

First, we can define the fermionic semi-norm of an operator $O$ as $\norm{O}_A = \norm{\Pi_A O \Pi_A}$, where $\Pi_A$ is the projector onto the anti-symmetric subspace (its explicit definition can be found in Appendix~\ref{app:fermionic_seminorm}). As a first example we can start by calculating the following quantity
\begin{equation}
\label{eq:maxi_norm_pot}
    \max_i\norm{\sum_{j\ne i}^{\eta-1} V_{ij}+\sum_{j\ne i}^{\eta-1}\sum_{k\ne i,j}^{\eta-1} V_{ijk}}_A\;,
\end{equation}
where we have denoted with $V_{ij}$ the pair potential from $V_2$ acting on the pair of particles $(i,j)$ and with $V_{ijk}$ the three-body potential from $V_3$ acting on the triplet $(i,j,k)$.
For a fixed particle $i$, the norm of the potential will depend on how many particles are occupying the same site as particle $i$. Due to the fact that we are considering anti-symmetric states, and that in our system there are only $4$ types of fermions distinguished by the spin and isospin quantum numbers, we will have a maximum of $4$ fermions in the same site at a time. This means that we can enumerate all possible arrangements of particles directly. In Table \ref{tab:potential_per_site} we enumerate the number of pairs and triples in which particle $i$ is involved simultaneously, in terms of how many fermions there are in that site.

\begin{table}[h]
    \centering
    \begin{tabular}{|c|c|c|}
        \hline
         fermion occupation & pairs & triples \\
         \hline
         1 & 0 & 0 \\
         \hline
         2 & 1 & 0 \\
         \hline
         3 & 2 & 1 \\
         \hline
         4 & 3 & 3 \\
         \hline
    \end{tabular}
    \caption{Number of pairs and triples on a spatial site as a function of the fermion occupation.}
    \label{tab:potential_per_site}
\end{table}

We can then directly bound the quantity from Eq.~\eqref{eq:maxi_norm_pot} by maximizing over all possible arrangements:
\begin{equation}
\max_i\norm{\sum_{j\ne i}^{\eta-1} V_{ij}+\sum_{j\ne i}^{\eta-1}\sum_{k\ne i,j}^{\eta-1} V_{ijk}}_A \le \max\left\{\left| \frac{C}{2} \right|, \left| C+{\  \frac{G}{3}}\right|, \left| \frac{3C}{2} + {\ G}\right|\right\}\;.
\end{equation}
The first term corresponds to a state with occupation $1$ on the site, the second with occupation $2$ and the last with $3$. Of course if only the $i$-th particle is on the site the contribution is zero.

Considering the kinetic energy, in the same way, we could say that there are at most $4$ fermions with the same momentum at a time. A simpler bound can be found directly noticing that, for a fixed $i$, we will have
\begin{equation}
    \max_i\norm{T_i}_A \le dK2^{2m-2}=\lambda_T/\eta\;,
\end{equation}
where we set the momentum of particle $i$ at its maximum possible value of $\pi/a$ for each cartesian direction. We are now ready to compute the norm of the commutator between the kinetic and potential energy:
\begin{equation}
\begin{split}
    \norm{[T,V]}_A &\le 2 \sum_{i=0}^{\eta-1} \left(\max_i\norm{T_i}_A \right)\left(\max_i\norm{2\sum_{j\ne i}^{\eta-1} V_{ij}+3\sum_{j\ne i}^{\eta-1}\sum_{k\ne i,j}^{\eta-1} V_{ijk}}_A \right)\\
    &\le 2\lambda_T\max\left\{\left| C \right|, \left| 2C+{\ G}\right|, \left| 3C + {\ 3G} \right|\right\}\;.
\end{split}
\end{equation}
with $\lambda_T$ from Eq.~\eqref{eq:lambdaT}. For the full calculation, refer to Appendix \ref{appendix_fermionic_seminorm_Trotter_error}. 

Then, for simplicity, we will define $\alpha_1$ as half of the bound we just computed
\begin{equation}
    \alpha_1 = \lambda_T\max\left\{\left| C \right|, \left| 2C+{\ G}\right|, \left| 3C + {\ 3G} \right|\right\}=O(\eta)\;,
    \label{eq:alpha1}
\end{equation}
so that $\norm{[T, V]}_A \le 2\alpha_1$. For the second and fourth order Trotter, we will define the upper bound to the error with $\alpha_2$ and $\alpha_4$ respectively, and their expressions are given at the end of Appendix \ref{appendix_fermionic_seminorm_Trotter_error}.

\begin{theorem}[First-order Trotter]
    Consider a system with $\eta$ particles on a $d$ dimensional lattice with $M=2^{m}$ sites per direction described by the pionless Hamiltonian $H = T+V$ with
    the kinetic and potential terms
    given in Eq.~\eqref{eq:Tdiscrete} and Eq.~\eqref{eq:Vdiscrete}.
    For any $t\in \mathbb{R}$ and $\epsilon > 0$ {\  we can construct a quantum circuit implementing} a unitary operator $U(t)$ such that $\norm{U(t) - e^{-iHt}} < \epsilon$ with a $T$-gate count equal to
    \begin{equation}
    \label{eq:tcount_trotter1st}
    \begin{split}
        T^{(1)}_{tr} (\epsilon, t) &= r\left( T_T \left( \frac{\epsilon}{3r}, \frac{t}{r} \right) + T_V \left( \frac{\epsilon}{3r} \right) \right) \\
        &= O\left[ \frac{t^2 \eta^2}{\epsilon}\left(m+\log \left( \frac{\eta t}{\epsilon} \right)\right) \left(\eta^2+m\right) \right]\\
        \end{split}
    \end{equation}
    where $T_T$ and $T_V$ are given by Lemma~\ref{theorem:kinetic_energy_operator} and Lemma~\ref{theorem:complact_Trotter_potential_energy} respectively.
    The number of ancilla qubits needed is
 \begin{equation}
    b_{QFT}+b_{DIAG}+\max \left( 2b_{QFT}-1,b_{DIAG}, m(m-1), dm-1 \right)\;,
    \end{equation}
    where we have defined
    \begin{equation}
    \begin{split}
b_{DIAG}&=\left\lceil\log_2\left(\frac{18d\eta t\lambda_T}{\epsilon}\log_2\left(\frac{18d\eta t\lambda_T}{\epsilon}\right)\right)\right\rceil\\
b_{QFT}&=\min\left(m-1,\left\lceil\log_2\left(\frac{9d\eta mr}{\epsilon}\right)\right\rceil\right)+1\;.
    \end{split}
    \end{equation}    
    In all expressions, $r = \left \lceil \frac{3t^2}{\epsilon}\alpha_1 \right \rceil$ and $\alpha_1$ is defined in Eq.~\eqref{eq:alpha1}. 
    For the asymptotic behavior in Eq.~\eqref{eq:tcount_trotter1st} we used the bound $\lambda_T$ from Eq.~\eqref{eq:lambdaT} for the spectral norm of the kinetic energy.
    \label{theorem:first_order_Trotter}
\end{theorem}

\begin{proof}
    Consider the first order Trotter formula, so that
    \begin{equation}
        U(t) = U_T(t)U_V(t)
    \end{equation}
    where $U_T(t)$ and $U_V(t)$ implement $e^{itT}$ and $e^{itV}$ with error $\epsilon_T$ and $\epsilon_V$ respectively. Then, we can write the total error of the Trotter formula as:
    \begin{equation}
        \norm{e^{itH} - \left(U\left(\frac{t}{r}\right)\right)^r} \le r\left( \epsilon_{Trotter}^{(1)} + \epsilon_T + \epsilon_V \right) \le \epsilon
    \end{equation}
    where $r$ is the number of Trotter steps and $\epsilon_{Trotter}^{(1)}$ an upperbound on the Trotter error at the first order given by (see Appendix \ref{appendix_fermionic_seminorm_Trotter_error} for the full calculation of the norm of the commutator):
    \begin{equation}
        \epsilon_{Trotter}^{(1)} = \left( \frac{t}{r} \right)^2 \frac{\norm{ [T, V]}}{2}, \quad\quad \frac{\norm{ [T, V] }}{2} \le \alpha_1 = \lambda_T\max\left\{\left| C \right|, \left| 2C+{\ G}\right|, \left| 3C + {\ 3G} \right|\right\}\;.
        \label{eq:first_order_Trotter_error}
    \end{equation}
    For simplicity, let us assume now that $\epsilon_T = \epsilon_V = \epsilon/3r$, so that the error inequality becomes:
    \begin{equation}
        r\epsilon_{Trotter}^{(1)} \le \frac{\epsilon}{3}
    \end{equation}
    From here we can find $r$ by setting it equal to the minimum possible value it can assume:
    \begin{equation}
        r = \left \lceil \frac{3t^2}{\epsilon}\alpha_1 \right \rceil
    \end{equation}
    Now that we have the number of oracle calls, we can evaluate the cost of the algorithm. The oracle $U(\tau) = U_T(\tau)U_V(\tau)$ will be applied $r$ times, and $\tau = t/r$. So the total number of $T$-gates is equal to:
    \begin{equation}
        T^{(1)}_{tr} (\epsilon, t) = r\left( T_T \left( \frac{\epsilon}{3r}, \frac{t}{r} \right) + T_V \left( \frac{\epsilon}{3r} \right) \right)
    \end{equation}

    The number of ancilla qubits needed to implement this algorithm is equal to the maximum between the number of ancilla qubits required to implement every single piece (assuming we can use many times the same set of ancilla qubits). It then follows that the number of ancilla qubits required is:
    \begin{equation}
    b_{QFT}+b_{DIAG}+\max \left( 2b_{QFT}-1,b_{DIAG}, m(m-1), dm-1 \right)\;,
    \end{equation}
    where the size of the phase register $b_{DIAG}$ for the diagonal unitary is now given by
    \begin{equation}
    b_{DIAG}=\left\lceil\log_2\left(\frac{18d\eta t\lambda_T}{\epsilon}\log_2\left(\frac{18d\eta t\lambda_T}{\epsilon}\right)\right)\right\rceil\;,
    \end{equation}
    while in order to implement the $QFT$ using Theorem~\ref{theorem:QFT} we require a phase register of size
    \begin{equation}
b_{QFT}=\left\lceil\log_2\left(\frac{9d\eta mr}{\epsilon}\right)\right\rceil+1\;.
    \end{equation}      
\end{proof}

Using the same implementations of the exponential of the kinetic and potential energy from Lemma~\ref{theorem:kinetic_energy_operator} and Lemma~\ref{theorem:complact_Trotter_potential_energy} we can also devise higher order Trotter decompositions. In particular for second and fourth order we have the following
\begin{theorem}[Second and fourth-order Trotter-Suzuki]
 Consider a system with $\eta$ particles on a $d$ dimensional lattice with $M=2^{m}$ sites per direction described by the pionless Hamiltonian $H = T+V$ with    the kinetic and potential terms
    given in Eq.~\eqref{eq:Tdiscrete} and Eq.~\eqref{eq:Vdiscrete}. 
    For any $t\in \mathbb{R}$ and $\epsilon > 0$ {\  we can construct a quantum circuit implementing} a unitary operator $U(t)$ such that $\norm{U(t) - e^{-iHt}} < \epsilon$ with a $T$-gate count equal to
    \begin{equation}
        T^{(2)}_{tr}(\epsilon, t) = r_2T_V \left( \frac{\epsilon}{4r_2} \right) + (r_2-1) T_T \left( \frac{\epsilon}{4r_2}, \frac{t}{r_2} \right) + 2T_T\left( \frac{\epsilon}{8r_2}, \frac{t}{2r_2} \right)\;,
    \end{equation}
    and
    \begin{equation}
    T^{(4)}_{tr}(\epsilon, t) = 5r_4T_V \left( \frac{\epsilon}{12r_4} \right) + 6r_4 T_T \left( \frac{\epsilon}{12r_4}, \frac{t}{r_4} \frac{1}{4-4^{1/3}}\right)\;,
    \end{equation}    
    where $T_T$ and $T_V$ are given by Lemma~\ref{theorem:kinetic_energy_operator} and Lemma~\ref{theorem:complact_Trotter_potential_energy} respectively. The number of steps for the two implementations are
    \begin{equation}
    r_2 = \left \lceil \sqrt{\frac{4t^3}{\epsilon}\alpha_2} \;\right \rceil \quad\text{and}\quad r_4 = \left \lceil \left( \frac{12t^5}{\epsilon}\alpha_4 \right)^{1/4} \right \rceil\;,
    \end{equation}
    with $\alpha_2$ and $\alpha_4$ from Eq.~\eqref{eq:alpha2} and Eq.~\eqref{eq:alpha4} respectively. The number of ancilla qubits required is given by 
    \begin{equation}
    b_{QFT}^{(k)} + b_{DIAG}^{(k)} + \max \left( 2b_{QFT}^{(k)}-1, b_{DIAG}^{(k)}, m(m-1), dm-1 \right)\;,
    \end{equation}
    where $k=2$ or $k=4$ for the two implementations respectively, and
    \begin{equation}
    b_{DIAG}^{(2)}=\left\lceil\log_2\left(\frac{24d\eta t\lambda_T}{\epsilon}\log_2\left(\frac{24d\eta t\lambda_T}{\epsilon}\right)\right)\right\rceil\;,
    \end{equation}
    \begin{equation}
    b_{DIAG}^{(4)}=\left\lceil\log_2\left(\frac{6d\eta t\lambda_T}{\epsilon}\frac{12}{4-4^{1/3}}\log_2\left(\frac{6d\eta t\lambda_T}{\epsilon}\frac{12}{4-4^{1/3}}\right)\right)\right\rceil\;,
    \end{equation}
    \begin{equation}
    b_{QFT}^{(k)} = \min\left( m-1, \left\lceil\log_2\left(\frac{12(k-1)d\eta mr}{\epsilon}\right)\right\rceil \right) +1\;,
    \end{equation}
    \label{theorem:2nd_and_fourth_Trotter}
\end{theorem}
\begin{proof}
    See Lemma~\ref{lemma:2nd_order_trotter} and Lemma~\ref{lemma:4th_order_trotter} in Appendix~\ref{appendix:higher_order_Trotter}.
\end{proof}

Using the spectral norm upperbound estimate $\lambda_T$ from Eq.~\eqref{eq:lambdaT} and reintroducing the size of the single particle basis $\Omega$ we have then found that the gate count scales as
\begin{equation}
\label{eq:tcount_scaling_trotter}
T_{tr}^{(k)}(\epsilon,t) = \widetilde{O}\left(\frac{t^{1+1/k}\eta^{1+1/k}}{\epsilon^{1/k}}\log\left(\frac{t\eta\Omega}{\epsilon}\right)\left(\eta^2+\log\left(\Omega \right)\right)\right)\;,
\end{equation}
where $k=1,2,4$ for the first, second and fourth order product formula respectively. For all of these implementations, the total number of qubits needed scales as
\begin{equation}
Q(t,\epsilon) = \widetilde{O}\left(\log\left(\frac{t}{\epsilon}\right)+\eta\log(\Omega)+\log^2(\Omega)\right)\;.
\end{equation}

We can in principle generalize the results presented above to any product formula of even order $2p$. By defining $S_{2p}(t)$ the $2p$-th order oracle, it can be shown that (see Ref.~\cite{PhysRevX.11.011020})
\begin{equation}
    \norm{e^{-itH} - S_{2p}(t)} \le t^{2p+1} \delta_{2p}
\end{equation}
where $\delta_{2p}$ is a linear combination of norms of commutators between the kinetic energy $T$ and the potential energy $V$. Since there is always at least one $T$ in these commutators, and since $T$ is a one-body operator, one can show in general that $\delta_{2p} = O(\eta)$. By introducing the time discretization we can find the number of oracle calls needed in terms of the maximum admitted error $\epsilon$
\begin{equation}
    \norm{e^{-itH} - S_{2p}\left( \frac{t}{r} \right)^r} \le r\left(\frac{t}{r}\right)^{2p+1} \delta_{2p} \le \epsilon \quad\quad \Rightarrow \quad\quad r = O\left( \frac{t^{1+\frac{1}{2p}} \eta^{\frac{1}{2p}}}{\epsilon^{\frac{1}{2p}}} \right)
\end{equation}
The cost of this algorithm will be $r$ times the cost of implementing $e^{-itT}$ and $e^{-itV}$ and we then recover the same asymptotic scaling from Eq.~\eqref{eq:tcount_scaling_trotter} for any $k=2p$ for integer $p$.
However, it is important to notice that the prefactor in the gate count is proportional to the number of times the oracle $S_{2p}$ contains $e^{-itT}$ and $e^{-itV}$, and this grows exponentially in the order $p$. Therefore the cost of product formulas will have a minimum, which usually is around order $4$, and that's why we stopped at the fourth order.

\subsection{Quantum Signal Processing implementation}
\label{sec:qspimpl}

In this section we present our implementation of the evolution operator for pionless EFT on a Lattice using Quantum Signal Processing~\cite{Low2017, Low_2019} and we will use the "Linear Combination of Unitaries" (LCU) method to block-encode the Hamiltonian~\cite{childs2012hamiltonian, Low_2019,Gily_n_2019}.

We start by providing our implementations of the block encoding of both the kinetic energy and the potential energy. For the kinetic energy operator $T$ we use the same strategy employed in Ref.~\cite{Su2021}. We start by writing the kinetic energy term $T$ as follows
\begin{equation}
\label{eq:Tdiscrete_forLCU}
    T = K \left(\text{QFT}^{\otimes d\eta}\right)^{\dagger} \sum_{i=0}^{\eta -1} \sum_{w=0}^{d-1} \sum_{p=0}^{2^{m}-1}\sum_{r=0}^{m-2}\sum_{s=0}^{m-2} 2^{r+s} p_rp_s \Pi_{w,i}(p) \left(\text{QFT}^{\otimes d\eta}\right)\;,
\end{equation}
where, in constrast to Eq.~\eqref{eq:Tdiscrete} we have explicitly written out the squared momentum in terms of the bits of $p$, denoted with $p_r$ and $p_s$. We then aim at preparing a state proportional to
\begin{equation}
\label{eq:lcu_prep_statet_T}
    \ket{+}_b \sum_{i=0}^{\eta-1} \ket{i}_f \sum_{\omega=0}^{d-1} \ket{\omega}_g \sum_{r=0}^{m-2} 2^{r/2}\ket{r}_h \sum_{s=0}^{m-2} 2^{s/2}\ket{s}_k\;.
\end{equation}
Using this state, we can then implement a block-encoding of $T$ by flipping the $b$ register state to $\ket{-}_b$ unless the $r$ and $s$ bits of the $\omega$ component of the momentum of particle $i$ are both equal to one. In this way we have effectively turned the (non unitary) projection operators $\Pi_{w,i}(p)$ from Eq.~\eqref{eq:Tdiscrete_forLCU} into controlled Z operations on the $b$ ancilla and can thus now perform a LCU. We thus have

\begin{lemma}(LCU Kinetic Energy)
    Consider the kinetic energy operator $T$ of Eq.~\eqref{eq:Tdiscrete_forLCU}. {\  We can construct a quantum circuit implementing} a block encoding $U_T$ of $T$ such that $\norm{ {}_a\langle0\lvert U_T\rvert0\rangle_a - T/\lambda_T}<\epsilon$ for
    \begin{equation}
    \lambda_T = dK\eta 2^{2(m-1)}\;,
    \end{equation}
    using $Q_a=\eta+2m+6$ qubits for the block encoding, and additional
    \begin{equation}
    b_r+b_{QFT}+\max\left(2b_{QFT}-1,m+6\right)\quad\text{with}\quad b_{QFT}=\min\left(m-1,\left\lceil\log_2\left(\frac{4dm}{\epsilon}\right)\right\rceil\right)+1\;,
    \end{equation}
    ancilla qubits, including the $b_{QFT}$ qubits needed to store the phase-gradient state used to implement the $QFT$ using the construction from Theorem.~\ref{theorem:QFT}. The gate count in terms of $T$-gates is instead given by:
    \begin{equation}
    \begin{split}
         T^{LCU}_T (\epsilon) 
         &=2dT_{QFT}\left(m,\frac{\epsilon}{4d}\right) + 24n_\eta+32b_r+18m + 4(2\eta-1) dm-4d-68\;,
    \end{split}
    \end{equation}
where $T_{QFT}$ is defined in Theorem.~\ref{theorem:QFT} in the Appendix and we introduced
\begin{equation}
n_\eta=\lceil\log_2(\eta)\rceil\quad\quad b_r = \left\lceil\frac{1}{2}\log_2\left(\frac{9\pi^2}{\epsilon}\right)\right\rceil \;.
\end{equation}
\label{lemma:LCU_T}
\end{lemma}

\begin{proof}
As a first step we describe how to prepare the state in Eq.~\eqref{eq:lcu_prep_statet_T} with a high success probability following Ref.~\cite{Su2021}. The $\ket{+}$ state in register $b$ can be prepared with a single Hadamard. The rest can be prepared with a large success probability following the same construction as in Sec.II.B of Ref.~\cite{Su2021} using a total of $\eta+2m+6$ qubits, out of which $4$ are flagging the successful preparation being in $\ket{0}$, and two are used for rotations and left in $\ket{0}$ upon success. The small failure probability in the preparation of the $h$ and $k$ register leads to an increase in the normalization of the block encoding which becomes~\cite{Su2021}
\begin{equation}
\lambda_T = dK\eta 2^{2(m-1)}\;,
\end{equation}
upon success of the other state prep. The total number of Toffoli required for this state preparation is $3n_\eta+4b_r+2m-16$, where $n_\eta=\lceil\log_2(\eta)\rceil$, and we can choose the bits of precision $b_r$ for rotations in order to control the final error in the block encoding due to subnormalization. Using the bound on the failure probability from Appendix.~J of Ref~\cite{Su2021} we can choose
\begin{equation}
b_r = \left\lceil\frac{1}{2}\log_2\left(\frac{9\pi^2}{\epsilon}\right)\right\rceil 
\end{equation}
in order to guarantee that we commit an error less then $\epsilon$ in block-encoding $T/\lambda_T$ provided the error for $SELECT$ is also bounded by $\epsilon/2$. The $SELECT$ operation $S_T$ instead acts as
\begin{equation}
    S_T \ket{+}_b \ket{k}_f \ket{w}_g \ket{r}_h \ket{s}_i \ket{p_{w}} = Z^{q_{w,r}q_{w,s}\oplus 1} \ket{b}_b \ket{k}_f \ket{\omega}_g \ket{r}_h \ket{s}_i \ket{p_{w}}
\end{equation}
and can be performed by applying the following steps (see Ref.~\cite{Su2021} for additional details)
\begin{itemize}
    \item controlled on the register $f$ we swap the register of particle $1$ with the one of particle $i$ using $(\eta-1) d m$ controlled SWAP gates (since we can skip the case corresponding to $i=0$ that requires no SWAP). The total cost is then $(\eta-1) d m$ Toffoli gates.
    \item we apply a $m$-bit $QFT$ on the $d$ registers for the first particle 
    \item controlled on the value $w$ on register $g$, we copy the leftmost $m-1$ bits of $p_{q,0}$ in an ancilla. This costs $d(m-1)$ Toffoli gates
    \item controlled on the $h$ register, we copy the $r$ bit of the momentum in an ancilla bit for a cost of $m-1$ Toffoli
    \item controlled on the $k$ register, we copy the $s$ bit of the momentum in an ancilla bit for a cost of $m-1$ Toffoli
    \item controlled on the last two ancilla bit, as well as anticontrolled on the $4$ ancilla bits used to flag success in the state preparation, apply a $Z$ gate to the $a$ register unless both ancilla are $1$. This can be done applying first a $Z$ gate without controls and then a $C^6Z$ gate. Using Theorem~\ref{theorem:Z_rotations} this requires $20$ T gates and $5$ additional ancilla qubits
    \item uncompute all the $m+1$ ancilla bits employed during the previous steps. This can be done with measurements and Clifford gates
    \item we finish by undoing the $QFT$ and the controlled $SWAP$
\end{itemize}

The full $SELECT$ protocol is depicted in the circuit in Fig.~\ref{fig:T_LCU_complete_circuit}, where the uncomputation of the ancilla in the next to last step is indicated with a $RST$ operation. The $4$ flag qubits used to denoted a successful state preparation are denoted by $ \ket{f}_f$, $ \ket{f}_g$, $ \ket{f}_h$ and $ \ket{f}_k$. With a slight abuse of notation we denoted as simple controls even controls on registers with multiple qubits, this has to be understood by doing a unary iteration on the individual qubits of the control. In order to be more explicit, we give a direct implementation of the controlled $SWAP$ network in Fig.~\ref{fig:controlled_SN_example_circuit} which is implemented with the controlled SWAP circuit from Fig.~\ref{fig:cSWAP_circuit}.

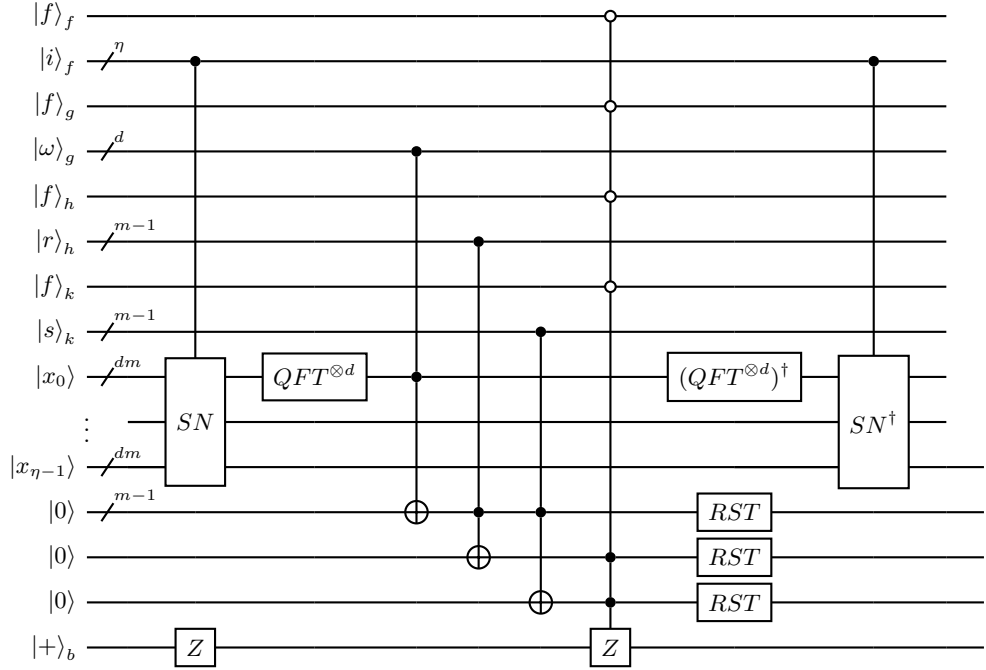
\begin{figure}
\centering
    \begin{quantikz}[row sep={0.6cm,between origins}]
    \lstick{$ \ket{f}_f$} & \qw & \qw                  & \qw      & \qw    & \qw    & \qw    &\octrl{2}&  \qw & \qw & \qw \\
    \lstick{$ \ket{i}_f$} & \qwbundle{\eta} & \ctrl{7} & \qw      & \qw    & \qw    & \qw    & \qw     &  \qw & \ctrl{7} & \qw \\
    \lstick{$\ket{f}_g$} & \qw & \qw                   & \qw      & \qw    & \qw    & \qw    &\octrl{2}&  \qw & \qw & \qw \\
    \lstick{$\ket{\omega}_g$} & \qwbundle{d} & \qw     & \qw      &\ctrl{5}& \qw    & \qw    & \qw     &  \qw & \qw & \qw \\
    \lstick{$\ket{f}_h$} & \qw & \qw                   & \qw      & \qw    & \qw    & \qw    &\octrl{2}&  \qw & \qw & \qw \\
    \lstick{$\ket{r}_h$} & \qwbundle{m-1} & \qw        & \qw      & \qw    &\ctrl{6}& \qw    & \qw     &  \qw & \qw & \qw \\
    \lstick{$\ket{f}_k$} & \qw & \qw                   & \qw      & \qw    & \qw    & \qw    &\octrl{6}&  \qw & \qw & \qw \\
    \lstick{$\ket{s}_k$} & \qwbundle{m-1} & \qw        & \qw      & \qw    & \qw    &\ctrl{4}& \qw     &  \qw & \qw & \qw \\
    \lstick{$\ket{x_0}$} & \qwbundle{dm} & \gate[3]{SN}&\gate{QFT^{\otimes d}}&\ctrl{3}& \qw    & \qw    & \qw     &  \gate{(QFT^{\otimes d})^\dag} & \gate[3]{SN^\dag} & \qw \\
    \vdots                                   & & \qw   & \qw      & \qw    & \qw    & \qw    & \qw     & \qw      &  \qw & \qw \\
    \lstick{$\ket{x_{\eta-1}}$} & \qwbundle{dm}& \qw   & \qw      & \qw    & \qw    & \qw    & \qw     & \qw      &  \qw & \qw & \qw \\
    \lstick{$\ket{0}$} & \qwbundle{m-1}        & \qw   & \qw      & \targ{}&\ctrl{1}&\ctrl{2}& \qw     &\gate{RST}&  \qw & \qw & \qw \\
    \lstick{$\ket{0}$} & \qw & \qw                     & \qw      & \qw    & \targ{}& \qw    & \ctrl{1}&\gate{RST}&  \qw & \qw & \qw \\
    \lstick{$\ket{0}$} & \qw & \qw                     & \qw      & \qw    & \qw    & \targ{}& \ctrl{1}&\gate{RST}&  \qw & \qw & \qw \\
    \lstick{$\ket{+}_b$} & \qw            & \gate{Z}   & \qw      & \qw    & \qw    & \qw    &\gate{Z} & \qw & \qw & \qw & \qw \\
    \end{quantikz}
    \caption{Circuit implementation of the $SELECT$ unitary for the kinetic energy operator.}
    \label{fig:T_LCU_complete_circuit}
\end{figure}

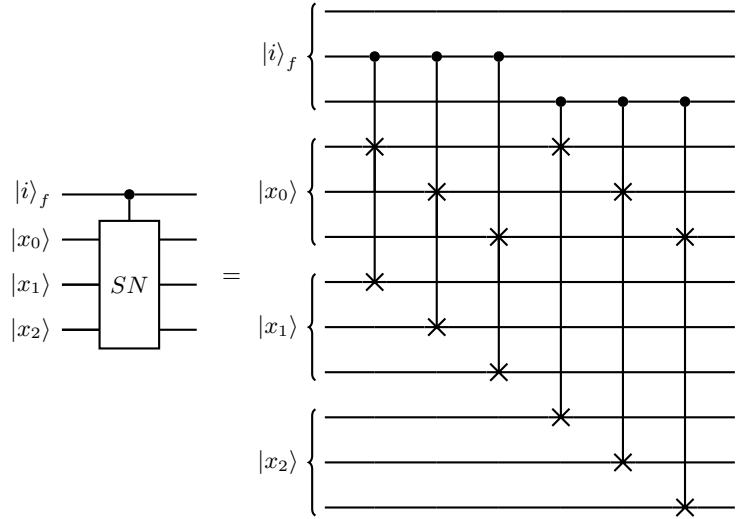
\begin{figure}
\centering
    \begin{quantikz}[row sep={0.6cm,between origins}]
    \lstick{$\ket{i}_f$} & \ctrl{1}&  \qw   \\
    \lstick{$\ket{x_0}$} &\gate[3]{SN}&  \qw   \\
    \lstick{$\ket{x_1}$} & \qw&  \qw   \\
    \lstick{$\ket{x_2}$} & \qw     &  \qw   \\
    \end{quantikz}\quad=     \begin{quantikz}[row sep={0.6cm,between origins}]
    \lstick[3]{$\ket{i}_f$} &  \qw    &  \qw   &  \qw   &  \qw   &  \qw   &  \qw   &  \qw\\
    & \ctrl{3}                        &\ctrl{4}&\ctrl{5}&  \qw   &  \qw   &  \qw   &  \qw\\
    &  \qw                            &  \qw   & \qw    &\ctrl{2}&\ctrl{3}&\ctrl{4}& \qw\\
    \lstick[3]{$\ket{x_0}$} & \swap{3}&  \qw   &  \qw   &\swap{6}&  \qw   &  \qw   &  \qw\\
    & \qw                             &\swap{3}&\qw     &  \qw   &\swap{6}&  \qw   &  \qw\\
    & \qw                             &\qw     &\swap{3}&  \qw   &  \qw   &\swap{6}&  \qw   \\
    \lstick[3]{$\ket{x_1}$} & \targX{}&  \qw   &  \qw   &  \qw   &  \qw   &  \qw   &  \qw\\
    & \qw                             &\targX{}&  \qw   &  \qw   &  \qw   &  \qw   &  \qw\\
    & \qw                             &  \qw   &\targX{}&  \qw   &  \qw   &  \qw   &  \qw\\
    \lstick[3]{$\ket{x_2}$} & \qw     &  \qw   &  \qw   &\targX{}&  \qw   &  \qw   &  \qw\\
    & \qw                             &  \qw   &  \qw   &  \qw   &\targX{}&  \qw   &  \qw\\
    & \qw                             &  \qw   &  \qw   &  \qw   &  \qw   &\targX{}&  \qw\\
    \end{quantikz}
    \caption{Circuit implementation of the controlled $SN$ unitary for a case with $\eta=3$ and $dm=3$.}
    \label{fig:controlled_SN_example_circuit}
\end{figure}

\begin{figure}
    \centering
    \begin{quantikz}[row sep={0.6cm,between origins}]
    \lstick{$\ket{a}$} & \ctrl{2}&  \qw   \\
    \lstick{$\ket{b}$} &\swap{1}&  \qw   \\
    \lstick{$\ket{c}$} & \targX{}&  \qw   \\
    \end{quantikz}\quad=    \begin{quantikz}[row sep={0.6cm,between origins}]
    \lstick{$\ket{a}$} &\qw      & \ctrl{1}&  \qw&  \qw   \\
    \lstick{$\ket{b}$} &\targ{}  &\ctrl{1} &\targ{}&  \qw   \\
    \lstick{$\ket{c}$} &\ctrl{-1}& \targ{}&\ctrl{-1}&  \qw   \\
    \end{quantikz}
    \caption{Circuit implementation of a controlled SWAP with a single Toffoli gate.}
    \label{fig:cSWAP_circuit}
\end{figure}
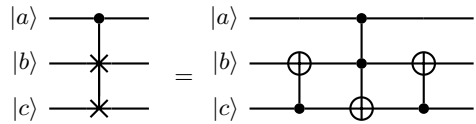

The total number of ancilla qubits required is thus $\eta+2m+6$ for the $PREPARE$ while for the $SELECT$ we need
\begin{equation}
    b_{QFT}+\max\left(2b_{QFT}-1,m+6\right)\quad\text{with}\quad b_{QFT}=\min\left(m-1,\left\lceil\log_2\left(\frac{4dm}{\epsilon}\right)\right\rceil\right)+1\;,
\end{equation}
ancilla qubits, including the $b_{QFT}$ qubits needed to store the phase-gradient state used to implement the $QFT$ using the construction from Theorem.~\ref{theorem:QFT}. Accounting for the need to apply the PREPARE and its inverse, for the gate count instead we need a number of Toffoli gates given by
\begin{equation}
6n_\eta+8b_r+4m-32 + 2(\eta-1)dm+(d+2)(m-1)
\end{equation}
which when converted into $T$ gates and added to the cost of 2d $QFT$ and the $C^6X$ gate gives a total of
\begin{equation}
    T^{LCU}_T(\epsilon) = 2dT_{QFT}\left(m,\frac{\epsilon}{4d}\right) + 24n_\eta+32b_r+18m + 4(2\eta-1) dm-4d-68 \\
\end{equation}
\end{proof}

Before continuing with the discussion of our implementation for the block-encoding of the potential $V$ we want to first point out that, 
at least for the $SU(4)$ symmetric interaction in Eq.~\eqref{eq:Vdiscrete}, we can express the diagonal matrix elements of the potential in the following form:
\begin{equation}
\begin{split}
    V(\vec{r}_0,\dots,\vec{r}_{\eta-1}) &= \frac{C}{2} \sum_{i=0}^{\eta-1} \sum_{\substack{j=0\\ j\ne i}}^{\eta-1} \delta_{\Vec{r}_i,\Vec{r}_j} + \frac{G}{6} \sum_{i=0}^{\eta-1} \sum_{\substack{j=0\\ j\ne i}}^{\eta-1} \sum_{\substack{k=0\\ k\ne i \\ k\ne j}}^{\eta-1} \delta_{\Vec{r}_i,\Vec{r}_j} \delta_{\Vec{r}_i,\Vec{r}_k} \\
    &= \frac{C}{2} \sum_{i=0}^{\eta-1} \sum_{\substack{j=0\\ j\ne i}}^{\eta-1} \delta_{\Vec{r}_i,\Vec{r}_j} + \frac{G}{6} \sum_{i=0}^{\eta-1} \sum_{\substack{j=0\\ j\ne i}}^{\eta-1} \delta_{\Vec{r}_i,\Vec{r}_j} \sum_{\substack{k=0\\ k\ne i}}^{\eta-1} \delta_{\Vec{r}_i,\Vec{r}_k} - \frac{G}{6} \sum_{i=0}^{\eta-1} \sum_{\substack{j=0\\ j\ne i}}^{\eta-1} \delta_{\Vec{r}_i,\Vec{r}_j} \;,
\end{split}
\end{equation}
where we used the fact that $\delta_{\Vec{r}_i,\Vec{r}_j}  \delta_{\Vec{r}_i,\Vec{r}_k} =  \delta_{\Vec{r}_i, \Vec{r}_j}$ if $j=k$. From here, we can write the potential as follows:
\begin{equation}
\label{eq:pot_compact}
    V(\vec{r}_0,\dots,\vec{r}_{\eta-1}) = \left( \frac{C}{2} - \frac{G}{6} \right) \sum_{i=0}^{\eta-1} \sum_{\substack{j=0\\ j\ne i}}^{\eta-1} \delta_{\Vec{r}_i,\Vec{r}_j} + \frac{G}{6} \sum_{i=0}^{\eta-1} \left( \sum_{\substack{j=0\\ j\ne i}}^{\eta-1} \delta_{\Vec{r}_i,\Vec{r}_j} \right)^2
\end{equation}
Using this expression to evaluate the matrix elements the potential energy operator would take the form
\begin{equation}
V=\left(\frac{C}{2}-\frac{G}{6}\right) \sum_{i=0}^{\eta-1} \sum_{\substack{j=0\\ j\ne i}}^{\eta-1}\sum_{\Vec{r}_i,\Vec{r}_j} \delta_{\Vec{r}_i,\Vec{r}_j}\Pi_i(\Vec{r}_i)\Pi_j(\Vec{r}_j)+\frac{G}{6} \sum_{i=0}^{\eta-1}\sum_{\Vec{r}_i}\Pi_i(\Vec{r}_i) \left( \sum_{\substack{j=0\\ j\ne i}}^{\eta-1} \sum_{\Vec{r}_j}\delta_{\Vec{r}_i,\Vec{r}_j} \Pi_j(\Vec{r}_j)\right)^2\;.
\end{equation}
In order to employ this efficiently as a basis for a block encoding, we first introduce the following oracle 
\begin{equation}
\label{eq:umatch}
U_{\text{match}}\ket{0}_S\ket{\Vec{r}_0}\cdots\ket{\Vec{r}_{\eta-1}}=\bigg\rvert\sum_{\substack{j=1}}^{\eta-1}\delta_{\Vec{r}_0,\Vec{r}_j}\bigg\rangle_S\ket{\Vec{r}_0}\cdots\ket{\Vec{r}_{\eta-1}}:=\ket{\Lambda}_S\ket{\Vec{r}_0}\cdots\ket{\Vec{r}_{\eta-1}}\;.
\end{equation}
The integer value $\Lambda$ counts the number of particles present on the same site as the one where we have the first particle $i=0$.Since we only have $4$ types of fermions, on any site we can at most have $4$ particles which means $\Lambda=0,1,2,3$ and we only need $S=2$ qubits to store its value without risking overflow. In order to implement the potential we first SWAP the coordinates of particle $i$, flagged by the value in unary in the $f$ register with $\eta$ qubit, with those of the first particle using the $cSN$ unitary used also for the kinetic energy (and given explicitly in Fig.~\ref{fig:controlled_SN_example_circuit}), apply $U_{match}$ and then SWAP back at the end. Using this oracle, the potential $V$ can then be written directly as
\begin{equation}
V=\sum_{i=0}^{\eta-1}\sum_{\Lambda=0}^3 \left(-\left(\frac{|C|}{2}+\frac{G}{6}\right)  \Lambda  +\frac{G}{6} \Lambda^2 \right)U^\dagger_{\text{match}}cSN\rvert i\rangle_f\langle i\lvert\rvert \Lambda\rangle\langle \Lambda\lvert cSNU_{\text{match}}\;,
\end{equation}
where we used that $C<0$ while $G>0$. If we make then explicit the bit sums as we did in Eq.~\eqref{eq:Tdiscrete_forLCU} we have
\begin{equation}
\begin{split}
\label{eq:v_for_b_enc}
V=&-\left(\frac{|C|}{2}+\frac{G}{6}\right)\sum_{i=0}^{\eta-1}\sum_{\Lambda=0}^3 \sum_{u=0}^1  2^u\Lambda_u U^\dagger_{\text{match}}cSN\rvert i\rangle_f\langle i\lvert\rvert \Lambda\rangle\langle \Lambda\lvert cSNU_{\text{match}}\\
&+\frac{G}{6} \sum_{i=0}^{\eta-1}\sum_{\Lambda=0}^3 \sum_{u=0}^1\sum_{v=0}^12^{u+v}\Lambda_u\Lambda_v U^\dagger_{\text{match}}cSN\rvert i\rangle_f\langle i\lvert\rvert \Lambda\rangle\langle \Lambda\lvert cSNU_{\text{match}}\;,
\end{split}
\end{equation}
We aim then to prepare a state given by
\begin{equation}
\label{eq:lcu_prep_statet_V}
\ket{\Phi_V}  =  \ket{+}_b \left(\frac{1}{\sqrt{\eta}}\sum_{i=0}^{\eta-1} \ket{i}_f \right)\ket{\phi}_l\ket{\kappa}_p\ket{\kappa}_q\;,
\end{equation}
where the $b$ and $f$ registers we can reuse those initialized for the kinetic energy, on $l$ we have
\begin{equation}
\label{eq:phi_state}
\ket{\phi}_l = \sqrt{\frac{6}{3|C|+2G}}\left(\sqrt{\frac{G}{6}}\ket{0}_l+\sqrt{\frac{|C|}{2}+\frac{G}{6}}\ket{1}_l\right)
\end{equation}
while on the $p$ and $q$ register we have
\begin{equation}
\label{eq:kappa_state}
\ket{\kappa} = \frac{1}{\sqrt{3}}\left(\ket{0}+\sqrt{2}\ket{1}\right)\;.
\end{equation}
We prepare all three of these states with arbitrary rotation synthesis using Theorem~\ref{theorem:Z_rotations}. The $SELECT$ is then implemented in a similar way as for the kinetic energy by applying a $Z$ gate to the $b$ register unless the correct bits of $\Lambda$ are set. We then have the following result

\begin{lemma}(LCU Potential Energy)
    Consider the full potential energy $V$ with $2$ and $3$-body operators $V_2$ and $V_3$ from Eq.~\eqref{eq:Vdiscrete}. {\  We can construct a quantum circuit implementing} a block encoding $U_V$ of $V$
    such that $\norm{ {}_a\langle0\lvert U_V\rvert0\rangle_a-V/\lambda_V} \le \epsilon$ for
    \begin{equation}
    \lambda_V =\eta\frac{3|C|+4G}{2}
    \end{equation}
    using the state preparation employed for the kinetic energy, additional $dm+4$ ancilla qubits and a $T$-gate count equal to:
    \begin{equation}
    \begin{split}
        T^{LCU}_V (\epsilon) &= 16(\eta-1) dm-8\eta+44+6T_{ROT}(\epsilon/12)\\
        &= O\left( \eta m + \log\left( \frac{1}{\epsilon} \right) \right)\;.
    \end{split}
    \end{equation}
    \label{lemma:LCU_V}
\end{lemma}

\begin{proof}
Assuming that we have already performed the state preparation for the kinetic energy, we only need to prepare the states $\ket{\phi}$ from Eq.~\eqref{eq:phi_state} on register $l$ and two copies of the state $\ket{\kappa}$ from Eq.~\eqref{eq:kappa_state} on registers $p$ and $q$. Since the $SELECT$ will be done exactly and the uniform superposition on the $f$ register comes already with an error $\epsilon/2$, we perform the $3$ rotations needed for these states with target error $\epsilon/12$. The resulting one norm we can use for the potential energy is thus given by
\begin{equation}
\lambda_V = \left(\frac{3|C|+G}{6}3+\frac{G}{6}9\right)\eta=\eta\frac{3|C|+4G}{2}
\end{equation}

For the $SELECT$ we first need to implement the unitary $U_{\text{match}}$ from Eq.~\eqref{eq:umatch}. The control on the label $i$ for the particle can be done controlling a SWAP network on register $f$ as we did for the kinetic energy. In order to implement $U_{\text{match}}$ we perform, for all the $(\eta-1)$ particles not in the first register a sequence of $dm$ CNOT gates with controls on the coordinate register of the first particle and targets of those of the other particle. We now perform $(\eta-1)$ incrementers on a register $S$ with two bits controlled on the coordinate bits of all the other particles, one at a time. We can set one ancilla bit with a, zero controlled, $C^{dm}X$ gate, we then add this to the two qubit register using a two qubit adder for $4$ T gates and using an additional ancilla and finish with a $C^{dm}X$ gate. Since for the first controlled incrementer the $S$ register start in $\ket{0}_S$, we can replace a full incrementer with just a CNOT and the total cost for implementing $U_{\text{match}}$ is then $4(\eta-2)$ T gates plus $2(\eta-1)$ $C^{dm}X$ gates. It also requires $2$ additional ancilla qubits, one for the adder and one for flagging equality. We show a complete circuit implementation of $U_{\text{match}}$ in a simple case with $\eta=3$ and $dm=3$ in Fig.~\ref{fig:Umatch_example_circuit}.

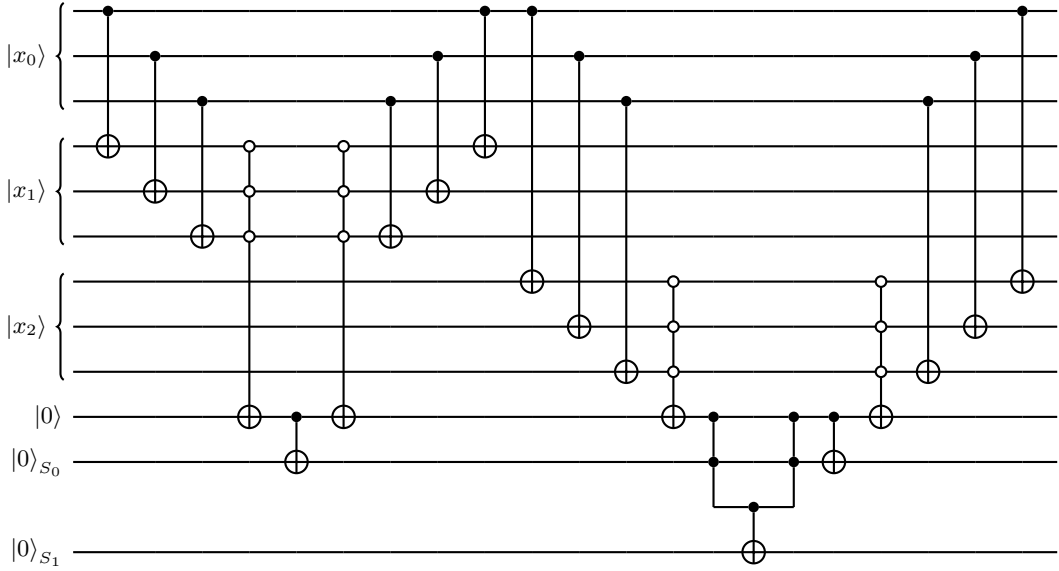
\begin{figure}
\centering
\begin{quantikz}[row sep={0.6cm,between origins},column sep={0.3cm}]
    \lstick[3]{$\ket{x_0}$}&\ctrl{3}&  \qw   &  \qw   &  \qw    &\qw     &\qw      &\qw     &\qw     &\ctrl{3}&\ctrl{6}&\qw     &\qw     &\qw      &\qw     &\qw     &\qw     &\qw     &\qw      &\qw     &\qw     &\ctrl{6}&\qw\\
    & \qw                           &\ctrl{3}&  \qw   &  \qw    &\qw     &\qw      &\qw     &\ctrl{3}&\qw     &\qw     &\ctrl{6}&\qw     &\qw      &\qw     &\qw     &\qw     &\qw     &\qw      &\qw     &\ctrl{6}&\qw     &\qw\\
    & \qw                           &  \qw   &\ctrl{3}&  \qw    &\qw     &\qw      &\ctrl{3}&\qw     &\qw     &\qw     &\qw     &\ctrl{6}&\qw      &\qw     &\qw     &\qw     &\qw     &\qw      &\ctrl{6}&\qw     &\qw     &\qw\\
    \lstick[3]{$\ket{x_1}$} &\targ{}&  \qw   &  \qw   &\octrl{1}&\qw     &\octrl{1}&\qw     &\qw     &\targ{} &\qw     &\qw     &\qw     &\qw      &\qw     &\qw     &\qw     &\qw     &\qw      &\qw     &\qw     &\qw     &\qw\\
    & \qw                           &\targ{} &  \qw   &\octrl{1}&\qw     &\octrl{1}&\qw     &\targ{} &\qw     &\qw     &\qw     &\qw     &\qw      &\qw     &\qw     &\qw     &\qw     &\qw      &\qw     &\qw     &\qw     &\qw\\
    & \qw                           &  \qw   & \targ{}&\octrl{4}&\qw     &\octrl{4}&\targ{} &\qw     &\qw     &\qw     &\qw     &\qw     &\qw      & \qw    &\qw     &\qw     &\qw     &\qw      &\qw     &\qw     &\qw     &\qw\\
    \lstick[3]{$\ket{x_2}$} &  \qw  &  \qw   &  \qw   &  \qw    &\qw     &\qw      &\qw     &\qw     &\qw     &\targ{} &\qw     &\qw     &\octrl{1}&\qw     &\qw     &\qw     &\qw     &\octrl{1}&\qw     &\qw     &\targ{} &\qw\\
    & \qw                           &  \qw   &   \qw  &  \qw    &\qw     &\qw      &\qw     &\qw     &\qw     &\qw     &\targ{} &\qw     &\octrl{1}&\qw     &\qw     &\qw     &\qw     &\octrl{1}&\qw     &\targ{} &\qw     &\qw\\
    & \qw                           &  \qw   &  \qw   &  \qw    &\qw     &\qw      &\qw     &\qw     &\qw     &\qw     &\qw     &\targ{} &\octrl{1}&\qw     &\qw     &\qw     &\qw     &\octrl{1}&\targ{} &\qw     &\qw     &\qw\\
    \lstick{$\ket{0}$}&\qw          &  \qw   &  \qw   &\targ{}  &\ctrl{1}&\targ{}  &\qw     &\qw     &\qw     &\qw     &\qw     &\qw     &\targ{}  &\ctrl{1}&\qw     &\ctrl{1}&\ctrl{1}&\targ{}  &\qw     &\qw     &\qw     &\qw\\
    \lstick{$\ket{0}_{S_0}$}&\qw    &  \qw   &  \qw   &\qw      &\targ{} &\qw      &\qw     &\qw     &\qw     &\qw     &\qw     &\qw     &\qw      &\ctrl{1}&\qw     &\ctrl{1}&\targ{} &\qw      &\qw     &\qw     &\qw     &\qw\\
                            &       &        &        &         &        &         &        &        &        &        &        &        &         &        &\ctrl{1}&\qw     &\\
    \lstick{$\ket{0}_{S_1}$}&\qw    &  \qw   &  \qw   &\qw      &\qw     &\qw      &\qw     &\qw     &\qw     &\qw     &\qw     &\qw     &\qw      &\qw     &\targ{} &\qw     &\qw     &\qw      &\qw     &\qw     &\qw     &\qw \\
    \end{quantikz}
    \caption{Circuit implementation of the unitary $U_{\text{match}}$ for a case with $\eta=3$ and $dm=3$. We use the same notation as Ref.~\cite{Gidney_2018} for the computation of the logical-AND of two qubits in an ancilla.}
    \label{fig:Umatch_example_circuit}
\end{figure}

After having applied $U_{\text{match}}$ we can finish the $SELECT$ using
\begin{itemize}
    \item we apply a $Z$ gate to the register $b$
    \item controlled on the $p$ register we copy the first or second bit of $\Lambda$ in an ancilla $a_1$ using $2$ Toffoli 
    \item controlled on the $q$ register we copy the first or second bit of $\Lambda$ in an ancilla $a_2$ using $2$ Toffoli 
    \item apply a CNOT with control on the $l$ qubit and target on the $b$ qubit in order to get the negative sign for the first term in Eq.~\eqref{eq:v_for_b_enc} 
    \item in order to finish applying the first line of Eq.~\eqref{eq:v_for_b_enc} we apply a $Z$ gate to $b$ but now controlled on the success flag for the state in the $f$ register being $\ket{0}$, on the state of the $l$ register being $\ket{1}$ and the $a_1$ ancilla being $\ket{1}$. This requires a $C^3X$ gate plus Clifford
    \item in order to apply the second line of Eq.~\eqref{eq:v_for_b_enc} we apply a $Z$ gate to $b$ but now controlled on the success flag for the state in the $f$ register being $\ket{0}$, on the state of the $l$ register being $\ket{0}$ and the $a_1$ and $a_2$ ancillas being $\ket{1}$. This requires a $C^4X$ gate plus Clifford
    \item we uncompute the $a_1$ and $a_2$ ancilla qubits with Clifford and measurements
    \item we undo the $U_{\text{match}}$ gate
\end{itemize}

We report the schematic implementation of this $SELECT$ unitary as a circuit in Fig.~\ref{fig:V_LCU_complete_circuit} using the same notation and conventions used in Fig.~\ref{fig:T_LCU_complete_circuit} above.

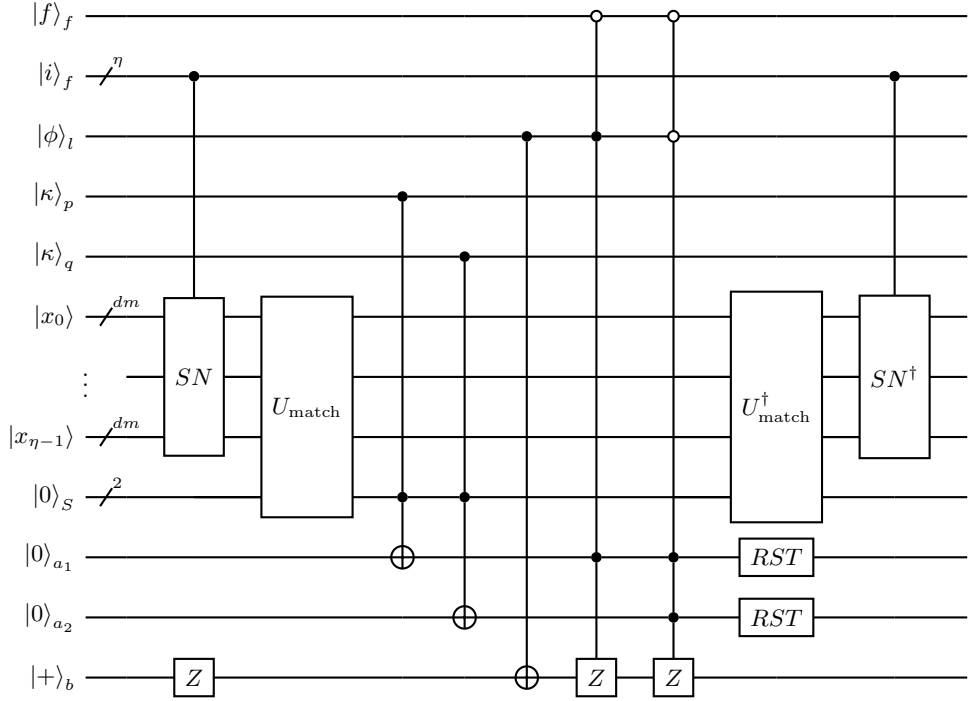
\begin{figure}
\centering
    \begin{quantikz}[row sep={0.8cm,between origins}]
    \lstick{$ \ket{f}_f$} & \qw & \qw                  & \qw                      & \qw    & \qw    &\qw     &\octrl{2}&\octrl{2}& \qw                              & \qw & \qw\\
    \lstick{$ \ket{i}_f$}&\qwbundle{\eta} &   \ctrl{4} & \qw                      & \qw    & \qw    & \qw    &  \qw    & \qw     & \qw                              & \ctrl{4}& \qw \\
    \lstick{$\ket{\phi}_l$}     & \qw & \qw            & \qw                      & \qw    & \qw    &\ctrl{9}&\ctrl{7} &\octrl{7}& \qw                              & \qw & \qw\\
    \lstick{$\ket{\kappa}_p$}& \qw & \qw               & \qw                      &\ctrl{5}& \qw    &\qw     &  \qw    & \qw     & \qw                              & \qw & \qw\\
    \lstick{$\ket{\kappa}_q$}& \qw & \qw               & \qw                      & \qw    &\ctrl{4}&\qw     &  \qw    & \qw     & \qw                              & \qw & \qw\\
    \lstick{$\ket{x_0}$} & \qwbundle{dm} & \gate[3]{SN}&\gate[4]{U_{\text{match}}}& \qw    & \qw    & \qw    &  \qw    & \qw     &\gate[4]{U^\dagger_{\text{match}}}&\gate[3]{SN^\dag} & \qw \\
    \vdots                                   & & \qw   & \qw                      & \qw    & \qw    & \qw    & \qw     &  \qw    & \qw                              &\qw& \qw \\
    \lstick{$\ket{x_{\eta-1}}$} & \qwbundle{dm}& \qw   & \qw                      & \qw    & \qw    & \qw    & \qw     &  \qw    & \qw                              &\qw & \qw \\
    \lstick{$\ket{0}_{S}$} & \qwbundle{2} & \qw        & \qw                      &\ctrl{1}&\ctrl{2}& \qw    &\qw      &  \qw    & \qw                               & \qw& \qw \\
    \lstick{$\ket{0}_{a_1}$} & \qw & \qw               & \qw                      &\targ{} & \qw    & \qw    &\ctrl{2} &\ctrl{1} & \gate{RST}                       & \qw & \qw\\
    \lstick{$\ket{0}_{a_2}$} & \qw & \qw               & \qw                      & \qw    &\targ{} & \qw    &\qw      & \ctrl{1}& \gate{RST}                       & \qw & \qw\\
    \lstick{$\ket{+}_b$} & \qw            & \gate{Z}   & \qw                      & \qw    & \qw    &\targ{} &\gate{Z} &\gate{Z} & \qw                              & \qw & \qw\\
    \end{quantikz}
    \caption{Circuit implementation of the $SELECT$ unitary for the potential energy operator.}
    \label{fig:V_LCU_complete_circuit}
\end{figure}

Summarizing we need $8\eta$ T gates, where we include those coming from the $4$ Toffoli used for $a_1$ and $a_2$ and for two $U_{\text{match}}$, $4(\eta-1)$ $C^{dm}X$ gates, one $C^3X$ and one $C^4X$ gate. For the last two we can reuse the ancilla employed for the $C^{dm}X$ gates. The total cost, including the two controlled SWAPS we could reuse from the kinetic energy when doing the full Hamiltonian, is then
\begin{equation}
16(\eta-1) dm-8\eta+36+6T_{ROT}(\epsilon/12)+8(\eta-1)dm
\end{equation}
T gates and a total of $dm+4$ ancilla, including those for state preparation.
\end{proof}

Combining these two results we finally arrive at

\begin{theorem}(LCU for Pionless-EFT)
    Consider the full Pionless-EFT Hamiltonian $H = T+V_2+V_3$ where the kinetic and potential terms are given in Eq.~\eqref{eq:Tdiscrete} and Eq.~\eqref{eq:Vdiscrete}. {\  We can construct a quantum circuit implementing} a block encoding $U_H$ for $H$ such that $\norm{ {}_a\langle0\lvert U_H\rvert0\rangle_a - H/\lambda_H}\le \epsilon$ using a number of $T$ gates given by:
    \begin{equation}
    \begin{split}
        T^{LCU}(\epsilon) =& 2T^P(\epsilon) + T^S(\epsilon) \\
        =& O\left( m\left(\eta+\min\left[\log\left( \frac{m}{\epsilon} \right),m\right]\right)+\log\left( \frac{1}{\epsilon} \right)\right) \;,\\
    \end{split}
    \end{equation}
    where
    \begin{equation}
    \begin{split}
        T^P\left( \epsilon \right) &= 12 n_\eta+16b_r+8m-64+4T_{ROT}\left( \frac{\epsilon}{32}\right) \;. \\
        T^S(\epsilon) &= 24(\eta-1)dm+4(d+2)(m-1)+16+2dT_{QFT}\left(m,\frac{\epsilon}{4d}\right)-8\eta+52
    \end{split}
    \end{equation}
    and the one-norm $\lambda_H$ is given by
    \begin{equation}
    \lambda_H = \eta\left(\frac{3|C|+4G}{2}+dK 2^{2(m-1)}\right)\;.
    \end{equation}
    
    The number of ancilla qubits employed for the block-encoding is $a=\eta+2m+10$ and we need an number of ancillas given by
    \begin{equation}
    b_r+b_{QFT}+\max\left(2b_{QFT}-1,(2d+1)m+6\right)\;,
    \end{equation}
    where
    \begin{equation}
    b_r = \left\lceil\frac{1}{2}\log_2\left(\frac{18\pi^2}{\epsilon}\right)\right\rceil\quad\text{and}\quad b_{QFT}=\min\left(m-1,\left\lceil\log_2\left(\frac{4dm}{\epsilon}\right)\right\rceil\right)+1
    \end{equation}
    including the $b_r$ and $b_{QFT}$ qubits needed to store the phase-gradient states.
\label{theorem:LCU_pionless}
\end{theorem}

\begin{proof}
    Let us start by defining the $PREPARE$ operator. First of all, we need a qubit to distinguish between $T$ and $V$, and use it to give them the right coupling constants. Calling that qubit as $\ket{\cdot}_a$, we will need the following superposition:
    \begin{equation}
    \label{eq:prep_tv}
        \sqrt{\frac{\lambda_{T}}{\lambda_{T} + \lambda_{V}}}\ket{0}_a + \sqrt{\frac{\lambda_{V}}{\lambda_{T} +\lambda_{V}}}\ket{1}_a\;,
    \end{equation}
    with the one-norms $\lambda_T$ and $\lambda_V$ from Lemma~\ref{lemma:LCU_T} and  Lemma~\ref{lemma:LCU_V} respectively. This can be prepared with a single-qubit rotation that can be approximated with $T_{ROT}(\epsilon_R)$ $T$-gates up to error $\epsilon_R$. The uniform superposition in the $f$ register will introduce an error $\epsilon_U$ when using a phase gradient state of size
\begin{equation}
b_r = \left\lceil\frac{1}{2}\log_2\left(\frac{9\pi^2}{4\epsilon_U}\right)\right\rceil\;,
\end{equation}
    while the $3$ rotations needed for the PREPARE of the potential energy will add an error $3\epsilon_R$ if we use $T_{ROT}(\epsilon_R)$ $T$-gates each. We aim for the total error in the PREPARE to be $\epsilon/4$ so that we can afford an error $\epsilon/2$ in the SELECT. We therefore need to choose the individual errors such that $4(\epsilon_U+4\epsilon_R)=\epsilon$. Here we take $\epsilon_U=\epsilon/8$ and $\epsilon_R=\epsilon/32$.
    
    Then, the $PREPARE$ requires $\eta + 2m+9$ qubits, for both $T$ and $V$ together, with a number of $T$ gates given by
    \begin{equation}
    T^P\left( \epsilon \right) = 12 n_\eta+16b_r+8m-64+4T_{ROT}\left( \frac{\epsilon}{32}\right) \;,
    \end{equation}
    where the phase gradient state uses a numbr of qubits given by
\begin{equation}
b_r = \left\lceil\frac{1}{2}\log_2\left(\frac{18\pi^2}{\epsilon}\right)\right\rceil\;.
\end{equation}    

    As for the $SELECT$, we need to apply the select operators for $V$ and $T$, adding a single control on the qubit $\ket{\cdot}_a$. This additional control can be implemented by adding one more control to the multi-controlled gates that act on the $b$ register: we have one for $T$ and 2 for $V$. In addition, the CNOT controlled on $l$ and with target $b$, that we used to get the right sign for the two body potential, will now need to be converted into a Toffoli. This can however be avoided exploiting the relation in Figure \ref{fig:CNOT_relation}, substituting the Toffoli with 2 control-$Z$ operations.
    \begin{figure}[t]
      \centering
      \begin{quantikz}[row sep={1cm,between origins}, column sep={0.4cm}, baseline=(current bounding box.center)]
          \lstick{$\ket{\cdot}_l$} & \qw & \ctrl{1} & \qw \\
          \lstick{$\ket{+}$}        & \gate{Z} & \targ{} & \qw
        \end{quantikz}
        \;=\;
       \begin{quantikz}[row sep={1cm,between origins}, column sep={0.4cm}, baseline=(current bounding box.center)]
          \lstick{$\ket{\cdot}_l$} & \gate{Z} & \qw \\
          \lstick{$\ket{+}$}        & \gate{Z} & \qw
        \end{quantikz}
      \caption{Relation for quantum circuits: exploiting the fact that we know the state of the second qubit, we can remove the CNOT by adding a single-qubit gate.}
      \label{fig:CNOT_relation}
    \end{figure}
    Additionally, in order for this block encoding to be easy to use with QSP later on, we would like the full $SELECT$ to be self-inverse. The individual $SELECT$ of $T$ and $V$ already are but they do not commute since the Quantum Fourier Transform and $U_{\text{match}}$ both act on the coordinate register of the first particle. Since, as shown in Fig.~\ref{fig:Umatch_example_circuit}, within $U_{\text{match}}$ this register has only controls, we can make the two individual $SELECT$ commute if we keep separate the $m+6$ qubits used as targets and perform the QFT on a copy of the register holding $\ket{x_0}$, which can be done adding $dm$ ancilla qubits and $2dm$ CNOT gates. Reusing the SWAP circuit for both $T$ and $V$ the total $T$-gate count for $SELECT$ becomes then
    \begin{equation}
    T^S(\epsilon) = 24(\eta-1)dm+4(d+2)(m-1)+16+2dT_{QFT}\left(m,\frac{\epsilon}{4d}\right)-8\eta+52
    \end{equation}
    T gates. We also need $6+(2d+1)m$ ancilla qubits in addition to the qubits required to perform the QFT and store its phase gradient state for a total of
    \begin{equation}
    b_r+b_{QFT}+\max\left(2b_{QFT}-1,(2d+1)m+6\right)\quad\text{with}\quad b_{QFT}=\min\left(m-1,\left\lceil\log_2\left(\frac{4dm}{\epsilon}\right)\right\rceil\right)+1\;.
    \end{equation}

    For the full block encoding we need then a number of $T$ gates equal to
    \begin{equation}
        T^{LCU}(\epsilon) = 2T^P(\epsilon) + T^S(\epsilon)
    \end{equation}
\end{proof}

We can now use the block encoding of the Hamiltonian from the previous Theorem to simulate time evolution generated by it using Quantum Signal Processing (QSP)~\cite{Low_2019,Gily_n_2019}. Here we briefly review the main steps required to use QSP. Consider the block-encoding of the Hamiltonian $U_H$ as a product of a $PREPARE$ unitary $P$ and a $SELECT$ unitary $U_S$ as follows
\begin{equation}
{}_a\langle0\lvert U_H\rvert0\rangle_a = {}_a\langle0\lvert \left(P^\dagger\otimes\mathbb{1}_s\right) U_S\left(P\otimes\mathbb{1}_s\right)\rvert0\rangle_a\equiv{}_a\langle P\lvert U_S\rvert P\rangle_a\;,
\end{equation}
where $a$ is the register containing the ancillas for the block encoding and $s$ is the system register. Then we construct two operators $W$ and $V_{\varphi}$ as follows
\begin{equation}
\label{eq:qubiterate}
    W = \left(\left(2\ket{P}\bra{P} - \mathbb{1} \right)_a \otimes \mathbb{1}_s \right) U_S
\end{equation}
\begin{equation}
    V_{\varphi} = \left( e^{-i\varphi Z/2} \otimes \mathbb{1}_s \right) V_0 \left( e^{i\varphi Z/2} \otimes \mathbb{1}_s \right)
\end{equation}
where $V_0 = \ket{+}\bra{+} \otimes \mathbb{1}_s + e^{i\pi/2}\ket{-}\bra{-} \otimes W$. The unitaries $V_\varphi$ can be implemented using the following circuit
\begin{equation}
\begin{quantikz}
    \lstick{$\ket{+}$} & \gate[3]{V_{\varphi}} & \qw \\
    \lstick{$\ket{P}_a$} & & \qw \\
    \lstick{$\ket{\psi}_s$} & & \qw \\
\end{quantikz} =
\begin{quantikz}
    \lstick{$\ket{+}$} & \gate{R_z(\varphi)} & \gate{H} &\gate{S}& \ctrl{1} & \gate{H} & \gate{R_z(\varphi)^\dag} & \qw \\
    \lstick{$\ket{P}_a$} & \qw& \qw & \qw & \gate[2]{W} & \qw & \qw & \qw \\
    \lstick{$\ket{\psi}_s$} & \qw & \qw& \qw & & \qw & \qw & \qw \\
\end{quantikz}
\end{equation}
where the $S$ gate is used to get the right phase of $e^{i\frac{\pi}{2}}$. For the controlled $W$ operation instead we can use
\begin{equation}
\begin{quantikz}
    \lstick{$\ket{+}$} & \ctrl{1} & \qw \\
    \lstick{$\ket{P}_a$} & \gate[2]{W} & \qw \\
    \lstick{$\ket{\psi}_s$} & & \qw \\
\end{quantikz} =
\begin{quantikz}
    \lstick{$\ket{+}$} & \ctrl{1} & \qw & \ctrl{0} & \qw & \qw \\
    \lstick{$\ket{P}_a$} & \gate[2]{U_S} & \gate{P^\dag} & \octrl{-1} & \gate{P} & \qw \\
    \lstick{$\ket{\psi}_s$} & & \qw & \qw & \qw & \qw \\
    \label{circuit:W_QSP}
\end{quantikz}
\end{equation}
By repeatedly applying $V_{\varphi}$ and its inverse for $Q$ times each, with possibly different rotation angles $\varphi_0,\dots,\varphi_{2Q}$, we can implement a large variety of polynomials of order $Q$ in the Hamiltonian~\cite{Low_2019,Gily_n_2019}. The total circuit is therefore of the following form
\begin{equation}
\label{eq:circ_qsp}
\begin{quantikz}
    \lstick{$\ket{+}$} & \qw & \gate[3]{V_{\varphi_1}} & \gate[3]{V_{\varphi_2+\pi}^\dag} & \gate[3]{V_{\varphi_3}}& \qw & & & \gate[3]{V_{\varphi_{2Q}+\pi}^\dag} & \qw & \qw\rstick{$\ket{+}$} \\
    \lstick{$\ket{0}_a$} & \gate{P} & && & \qw & \cdots & & & \gate{P^\dag} &\qw \rstick{$\ket{0}_a$} \\
    \lstick{$\ket{\psi}_s$} & \qw & && & \qw & & & & \qw \rstick{$e^{-iHt}\ket{\psi}_s$} \\
\end{quantikz}
\end{equation}
Note that this construction requires the unitary $W$ to be a reflection implying that, for the construction in Eq.~\eqref{eq:qubiterate} to work, the $SELECT$ unitary $U_S$ needs to be self-inverse. The $U_S$ unitary of the block encoding of Theorem~\ref{theorem:LCU_pionless} was constructed explicitly to satisfy this condition. 
For a fixed polynomial order $Q$, the circuit in Eq.~\eqref{eq:circ_qsp} requires $2Q$ controlled $U_S$, $Q+1$ uses of both $P$ and $P^\dagger$, $2Q+1$ arbitrary $Z$ rotations and $2Q$ multi-controlled $Z$ gates for the reflections.

\begin{theorem}[QSP for Pionless-EFT]
 Consider a system with $\eta$ particles on a $d$ dimensional lattice with $M=2^{m}$ sites per direction described by the pionless Hamiltonian $H = T+V$ with the kinetic and potential terms given in Eq.~\eqref{eq:Tdiscrete} and Eq.~\eqref{eq:Vdiscrete}. Then for any $t\in \mathbb{R}$ and $\frac{1}{2}>\epsilon>0$ {\  we can construct a quantum circuit implementing} a unitary operator $U(t)$ such that $\norm{U(t) - e^{-iHt}}< \epsilon$ with a one-norm equal to
\begin{equation}
\lambda_H = \eta\left(\frac{3|C|+4G}{2}+dK 2^{2(m-1)}\right)=O(\eta)\;
\end{equation}    
    using a number of $T$ gates given by
\begin{equation}
\begin{split}
    T^{QSP}(\epsilon, t) &= 2Q\left(t, \frac{\epsilon}{2}\right) \left( T^S\left( \frac{\epsilon}{4\lambda_H |t|} \right) + 2T^P\left( \frac{\epsilon}{4\lambda_H |t|} \right) + 3 \right) + T^P\left( \frac{\epsilon}{4\lambda_H |t|} \right) \\
    &+ 2\left(Q\left(t,\frac{\epsilon}{2}\right)+1\right) T_{ROT}\left( \frac{\epsilon}{4\left(Q\left(t,\frac{\epsilon}{2}\right)+1\right)} \right) \\
    &=\widetilde{O}\left(\left(\eta t+\log\left(\frac{1}{\epsilon}\right)\right) \left(\eta m + m\log\left(\frac{m t}{\epsilon}\right)\right)\right)\\
\end{split}
\end{equation}
where we introduced the total number of repetitions
\begin{equation}
Q(t,\epsilon_t) = \left\lceil 2\lambda_H|t|+3\log\left(\frac{6}{\epsilon_t}\right)\right\rceil\;.
\end{equation}
and used $T^S(\epsilon),T^P(\epsilon)$ from Theorem~\ref{theorem:LCU_pionless}.
The quantity $T_{ROT}$ is instead defined in Theorem~\ref{theorem:Z_rotations}. Finally, the algorithm requires a total number of qubits given by
\begin{equation}
    \eta+(2+d)m+12+b_r+b_{QFT}+\max\left(2b_{QFT}-1,dm+6,\eta+2m+9\right)=O\left(\eta+m+\log\left(\frac{t}{\epsilon}\right)\right)\;,
\end{equation}
where we have defined
\begin{equation}
b_r = \left\lceil\frac{1}{2}\log_2\left(\frac{72\pi^2\lambda_H |t|}{\epsilon}\right)\right\rceil\quad\text{and}\quad b_{QFT}=\min\left(m-1,\left\lceil\log_2\left(\frac{16d\lambda_H |t|m}{\epsilon}\right)\right\rceil\right)+1\;,
\end{equation}
in addition to the $n_s=\eta d m$ qubits used to encode the system.
\label{theorem:QSP}
\end{theorem}

\begin{proof}

To compute the cost of the algorithm, we first recall from Lemma~57 of Ref.~\cite{gilyen2018quantum} that using the Jacobi-Anger expansion we can build a polynomial $P_{Q}(t)$ of order $Q$ such that
\begin{equation}
\norm{e^{it}-P_Q(t)}_{[-1,1]}\leq2\sum_{q=Q+1}^\infty\left|J_q(t)\right|^{q}\;,
\end{equation}
where $J_q$ are Bessel polynomials of the first kind. Using a similar calculation than the one used in Ref.~\cite{gilyen2018quantum} we have the following sequence of bounds
\begin{equation}
2\sum_{q=K}^\infty\left|J_q(t)\right|^{q}\leq2\sum_{q=0}^\infty \frac{(|t|/2)^{K+q}}{(K+q)!}\leq 2\frac{(|t|/2)^{K}}{K!}\sum_{q=0}^\infty 2^{-q}=4\frac{(|t|/2)^{K}}{K!}\leq\sqrt{\frac{8}{\pi K}}\left|\frac{et}{2K}\right|^K\;.
\end{equation}
We now take $r(t,\delta)$ such that $\delta=(t/r(t,\delta))^r(t,\delta)$ and for $K= Q+1\geq r(e|t|/2,\delta)$ we have
\begin{equation}
2\sum_{q=K}^\infty\left|J_q(t)\right|^{q}\leq \sqrt{\frac{8}{\pi K}}\left|\frac{et}{2K}\right|^K\leq\frac{2}{\sqrt{\pi}}\left|\frac{et}{2K}\right|^K\leq \frac{2}{\sqrt{\pi}}\delta\;.
\end{equation}
If we want an error $\epsilon_P$ for this approximation then we can just take
\begin{equation}
Q=\left\lceil r\left(\frac{e|t|}{2},\frac{\sqrt{\pi}}{2}\epsilon_P\right)-1\right\rceil\;.
\end{equation}
After the truncation of the Jacobi-Anger expansion the resulting polynomial $P_Q(t)$ satisfies the requirements for being implementable with QSP only approximately. In particular
\begin{equation}
1-2\epsilon_P\leq\left|P_Q(t)\right|^2\leq 1+2\epsilon_P+\epsilon_P^2\;,
\end{equation}
instead of $\left|P_Q(t)\right|^2=1$, while $P_Q(0)=1$ as it should. Following the construction from Lemma~14 of Ref.~\cite{Low_2019} we can however find an implementable polynomial $\widetilde{P}_{Q}(t)$ of the same degree such that
\begin{equation}
\norm{e^{it}-\widetilde{P}_Q(t)}_{[-1,1]}\leq3\epsilon_P+\epsilon_P^2\leq A \epsilon_P\;,
\label{eq:poly_bound}
\end{equation}
for a parameter $A>3$ to be determined later.
We note that the same would hold for the polynomials that need to be built for applications of the generalized QSP as derived in Ref.~\cite{PhysRevA.110.012612}.

 Now we can notice that, using Theorem~\ref{theorem:LCU_pionless} for accuracy $\epsilon_H$, our block encoding of $H$ satisfies
\begin{equation}
\norm{ \widetilde{H} - H}\le \epsilon_H \lambda_H\quad\text{where}\quad \widetilde{H}=\lambda_H\;{}_a\langle0\lvert U_H\rvert0\rangle_a\;.
\end{equation}
If we then choose $\epsilon_H=\epsilon_t/|2\lambda_H t|$ we can use QSP to build a unitary $\widetilde{U}(t)$ such that $\norm{\widetilde{U}(t)-e^{it\widetilde{H}}} < \epsilon_t/2$ which will also be an approximation to the evolution operator $e^{itH}$ with $\norm{U(t)-e^{itH}} < \epsilon_t$. For this purpose we will need a polynomial of order
\begin{equation}
Q\geq\left\lceil r\left(\frac{e\lambda_H|t|}{2},\frac{\sqrt{\pi}}{4A}\epsilon_t\right)-1\right\rceil\;
\end{equation}
to use in QSP. However, for every $q\in \mathbb{R}$ we have
\begin{equation}
\begin{split}
 r\left(\frac{e\lambda_H|t|}{2},\frac{\sqrt{\pi}}{4A}\epsilon_t\right)-1&\leq e^q \frac{e\lambda_H|t|}{2}+\frac{1}{q}\log\left(\frac{4A}{\sqrt{\pi}\epsilon_t}\right)-1\\
 &\leq e^q \frac{e\lambda_H|t|}{2}+\frac{1}{q}\log\left(\frac{4A}{e^q\sqrt{\pi}\epsilon_t}\right)\leq 2\lambda_H|t|+3\log\left(\frac{6}{\epsilon_t}\right)\;,
\end{split}
\end{equation}
where we chose $A=3.5$, which satisfies Eq.~\eqref{eq:poly_bound} when $\epsilon_P\leq1/2$, and $q=1/3$. The final polynomial order we choose for the decomposition is therefore given by
\begin{equation}
Q(t,\epsilon_t) = \left\lceil 2\lambda_H|t|+3\log\left(\frac{6}{\epsilon_t}\right)\right\rceil\;.
\end{equation}

Remember from Theorem~\ref{theorem:LCU_pionless} that the number of ancillas in the register for block encoding is
\begin{equation}
n_a=\eta+2m+10\;,
\end{equation}
so that the total cost for reflections done using $C^{n_a}X$ gates amounts to
\begin{equation}
8Q(t,\epsilon_t)(n_a-1) = 8Q(t,\epsilon)(\eta+2m+9)\;
\end{equation}
T gates.
Since we alternate between controlled $W$ and controlled $W^\dagger$ we can cancel a $P$ with a $P^\dagger$ every pair resulting in a total of at most $2(Q(t,\epsilon_t)+1)$ uses of $P$ and $P^\dagger$ for a T gate cost of
\begin{equation}
2(Q(t,\epsilon_t)+1)T^P\left( \frac{\epsilon_t}{2\lambda_H |t|} \right)\;,
\end{equation}
where we used an error $\epsilon_H=\epsilon_t/|2\lambda_H t|$ as determined above, while $T^P(\epsilon)$ is defined in Theorem~\ref{theorem:LCU_pionless}. For the $SELECT$ instead, accounting for the $3$ more Toffoli gates required for controlling it, we have a total cost in T gates of
\begin{equation}
    2Q(t,\epsilon_t)\left( T^S\left( \frac{\epsilon_t}{2\lambda_H |t|} \right) + 3 \right)
\end{equation}
where again $T^S(\epsilon)$ is defined in Theorem~\ref{theorem:LCU_pionless}.

In summary, the total T gate count for the algorithm so far is given by
\begin{equation}
    T_{W}(\epsilon_t) = 2Q(t, \epsilon_t)\left( T^S\left( \frac{\epsilon_t}{2\lambda_H |t|} \right) + T^P\left( \frac{\epsilon_t}{2\lambda_H |t|} \right) + 3 \right) +2 T^P\left( \frac{\epsilon_t}{2\lambda_H |t|} \right)
\end{equation}

Finally, since we also need $2(Q(t,\epsilon_t)+1)$ rotations for the full algorithm we can do those such that the total accumulated error is below $\epsilon/2$ and therefore choose $\epsilon_t=\epsilon/2$ above. The total gate cost, adding everything together, is given by
\begin{equation}
\begin{split}
    T^{QSP}(\epsilon, t) &= T_{W}\left( \frac{\epsilon}{2} \right) + 2(Q(t,\epsilon_t)+1) T_{ROT}\left( \frac{\epsilon}{4(Q(t,\epsilon_t)+1)} \right) \\
    &=\widetilde{O}\left(\left(\eta t+\log\left(\frac{1}{\epsilon}\right)\right) \left(\eta m + m\log\left(\frac{m\eta t}{\epsilon}\right)\right)\right)\\
\end{split}
\end{equation}

The total ancilla cost is then $n_a+2=\eta+2m+12$ qubits for the full block encoding of the evolution operator, together with an additional register of size
\begin{equation}
    b_r+b_{QFT}+\max\left(2b_{QFT}-1,dm+6,\eta+2m+9\right)\;,
\end{equation}
where we defined
\begin{equation}
b_r = \left\lceil\frac{1}{2}\log_2\left(\frac{72\pi^2\lambda_H |t|}{\epsilon}\right)\right\rceil\quad\text{and}\quad b_{QFT}=\min\left(m-1,\left\lceil\log_2\left(\frac{16d\lambda_H |t|m}{\epsilon}\right)\right\rceil\right)+1\;,
\end{equation}
which includes the space to store the phase gradient states for the PREPARE and the QFT, plus an additional register of $dm$ qubits to perform the QFT in a separate register so that the SELECT is self-inverse.

\end{proof}

The QSP-based scheme just presented can be further improved by a factor of $\approx 2$ using the Generalized Quantum Signal Processing algorithm of Ref.~\cite{PhysRevA.110.012612}.

\begin{theorem}[G-QSP for Pionless-EFT]
Consider a system with $\eta$ particles on a $d$ dimensional lattice with $M=2^{m}$ sites per direction described by the pionless Hamiltonian $H = T+V$ with the kinetic and potential terms given in Eq.~\eqref{eq:Tdiscrete} and Eq.~\eqref{eq:Vdiscrete}.
Then for any $t\in \mathbb{R}$ and $\frac{1}{2}>\epsilon>0$ {\  we can construct a quantum circuit implementing} a unitary operator $U(t)$ such that $\norm{U(t) - e^{-iHt}}< \epsilon$ with a one-norm equal to
\begin{equation}
\lambda_H = \eta\left(\frac{3|C|+4G}{2}+dK 2^{2(m-1)}\right)=O(\eta)\;
\end{equation}    
    using a number of $T$ gates given by
\begin{equation}
\begin{split}
    T^{GQSP}(\epsilon, t) &= Q\left(t, \frac{\epsilon}{2}\right)T^{S}\left(\frac{\epsilon}{4\lambda_H |t|}\right) + 2\left( Q\left(t, \frac{\epsilon}{2}\right)+1 \right) T^{P} \left( \frac{\epsilon}{4\lambda_H |t|} \right) \\
    &+ 3\left( Q\left(t, \frac{\epsilon}{2}\right)+1 \right) T_{ROT} \left( \frac{\epsilon}{6\left( Q\left(t, \frac{\epsilon}{2}\right)+1 \right)} \right) +2Q\left(t, \frac{\epsilon}{2}\right)T_{MCX}(\eta+2m+10)\\
    &=O\left(\left(\eta t+\log\left(\frac{1}{\epsilon}\right)\right) \left(\eta m + m\log\left(\frac{m t}{\epsilon}\right)\right)\right)\\
\end{split}
\end{equation}
where we introduced the total number of repetitions
\begin{equation}
Q(t,\epsilon_t) = \left\lceil 2\lambda_H|t|+3\log\left(\frac{6}{\epsilon_t}\right)\right\rceil\;.
\end{equation}
and used $T^S(\epsilon),T^P(\epsilon)$ from Theorem~\ref{theorem:LCU_pionless}.
The quantities $T_{ROT}$ and $T_{MCX}$ are instead defined in Theorem~\ref{theorem:Z_rotations} and Theorem~\ref{theorem:MCX} respectively. Finally, the algorithm requires a total number of qubits given by
\begin{equation}
    \eta+(2+d)m+13+b_r+b_{QFT}+\max\left(2b_{QFT}-1,dm+6,\eta+2m+9\right)=O\left(\eta+m+\log\left(\frac{t}{\epsilon}\right)\right)\;,
\end{equation}
where we have defined
\begin{equation}
b_r = \left\lceil\frac{1}{2}\log_2\left(\frac{72\pi^2\lambda_H |t|}{\epsilon}\right)\right\rceil\quad\text{and}\quad b_{QFT}=\min\left(m-1,\left\lceil\log_2\left(\frac{16d\lambda_H |t|m}{\epsilon}\right)\right\rceil\right)+1\;,
\end{equation}
in addition to the $n_s=\eta d m$ qubits used to encode the system.
\end{theorem}

\begin{proof}
In this generalized version of QSP, we can define a new operator $V_{\alpha, \beta}$:
\begin{equation}
\begin{quantikz}
    \lstick{$\ket{+}$} & \gate[3]{V_{\alpha, \beta}} & \qw \\
    \lstick{$\ket{P}_a$} & & \qw \\
    \lstick{$\ket{\psi}_s$} & & \qw \\
\end{quantikz} =
\begin{quantikz}
    \lstick{$\ket{+}$} & \gate{R(\alpha, \beta)} & \ctrl{1} & \octrl{1} & \qw \\
    \lstick{$\ket{P}_a$} & \qw & \gate[2]{W^\dagger} & \gate[2]{W} & \qw \\
    \lstick{$\ket{\psi}_s$} & \qw & &  & \qw \\
\end{quantikz}
\end{equation}
where the operation $R(\alpha, \beta)$ introduced in the circuit is a generic $SU(2)$ single qubit gate.
The full circuit required to implement a polynomial of order $Q$ will require $Q$ applications of this $V_{\alpha, \beta}$ operator with different angles, plus one last rotation. However, note that we can exploit the following relation:
\begin{equation}
\begin{quantikz}
    & \ctrl{1} & \octrl{1} & \qw \\
    & \gate[2]{W^\dagger} & \gate[2]{W} & \qw \\
    & &  & \qw \\
\end{quantikz} =
\begin{quantikz}
    & \qw & \ctrl{1} & \qw & \qw & \qw & \octrl{1} & \qw & \qw  \\
    & \gate{P} & \ctrl{0} & \gate{P^\dagger} & \gate[2]{S} & \gate{P} & \ctrl{0} & \gate{P^{\dagger}} & \qw \\
    & \qw & \qw & \qw & & \qw & \qw & \qw & \qw \\
\end{quantikz}
\end{equation}
where we used the fact that control-$U$ times anticontrol-$U$ is just $U$. On the other hand, when concatenating multiple $V_{\alpha, \beta}$, one $P$ and one $P^{\dagger}$ simplify. This means that the total cost of implementing the time evolution is: $2Q+1$ applications of $P$ or $P^\dagger$, $Q$ applications of $S$, and $Q+1$ $SU(2)$ rotations, that can take up to $3$ $R_Z$ rotations and Hadamard gates, and finally $2Q$ reflections on the $ a=\eta+2m+10$ qubits where $P$ acts. The total number of $T$-gates is:
\begin{equation}
\begin{split}
    T^{GQSP}(\epsilon) &= Q\left(t, \frac{\epsilon}{2}\right)T^{S}\left(\frac{\epsilon}{4\lambda_H |t|}\right) + 2\left( Q\left(t, \frac{\epsilon}{2}\right)+1 \right) T^{P} \left( \frac{\epsilon}{4\lambda_H |t|} \right) \\
    &+ 3\left( Q\left(t, \frac{\epsilon}{2}\right)+1 \right) T_{ROT} \left( \frac{\epsilon}{6\left( Q\left(t, \frac{\epsilon}{2}\right)+1 \right)} \right) +2Q\left(t, \frac{\epsilon}{2}\right)T_{MCX}(\eta+2m+10)\\
\end{split}
\end{equation}

Where the value of $Q\left(t, \frac{\epsilon}{2}\right)$, the errors, and number of qubits are the same as in Theorem~\ref{theorem:QSP}. We also used $T^S(\epsilon),T^P(\epsilon)$ from Theorem~\ref{theorem:LCU_pionless}, $T_{ROT}$ from Theorem~\ref{theorem:Z_rotations} and $T_{MCX}$ from Theorem~\ref{theorem:MCX}. 

\end{proof}

\section{Cost of state preparation}
\label{sec:state_prep}

{\  Our main focus in this work is to provide an accurate estimate of the computational resources, gates and memory, required to perform time evolution with the simplest nuclear Hamiltonian. This is a key subroutine to enable the study of dynamical properties of  nuclear system on digital quantum computers but another, equally important, component is the procedure to prepare low energy nuclear states describing the reagents in the collision process. For example, in a typical lepton-nucleus reaction the goal is to prepare a good approximation of the ground-state of the target nucleus. }

{\  In the most general case determining the ground-state energy, and thus also preparing the ground state, for a $k$-local qubit Hamiltonian is known to be extremely hard, $QMA$-complete in fact~\cite{kempe2006complexity}, for any $k\geq2$. The main source of this complexity is the possibility that the ground state energy gap becomes exponentially small in system size. When preparing the ground state of a target with choice of basis of single particle we do not expect to be in this since the ground state energy gap will only go to a constant as we make the linear size of the simulation box larger while keeping the lattice spacing fixed.}

{\  Since the pioneering work by Dumitrescu et. al~\cite{PhysRevLett.120.210501}, a variety of techniques has been developed and applied to nuclear systems (see e.g.~\cite{https://doi.org/10.1002/qute.202300219,refId0,Ayral2023} for recent reviews). The overwhelming majority of these works uses a second quantized encoding for the full fermionic Fock space and only recently are state preparation techniques being applied to nucleons in first quantization~\cite{weiss2025solvingreactiondynamicsquantum,stetcu2025antisymmetrization}.
An accurate cost estimation for state preparation is hindered by two main factors. The first is technical: many popular algorithms rely on heuristic optimization schemes to prepare approximate ground states making cost estimates difficult to obtain. The second has to do with physics: different nuclear targets will have very different ground states depending on the constituent nucleons and the nuclear interaction being used. For example, closed shell nuclei like ${}^{16}O$ or ${}^{40}C$ tend to have larger energy gaps to their excited states than an open shell nucleus like ${}^{40}Ar$. The resources required for the initial state preparation required for simulations of nuclear reactions will then strongly depend on the process being studied. A comprehensive study of these costs is beyond the scope of the present work which focuses instead on the Hamiltonian evolution part of a nuclear reaction simulation. }

{\  In order to put the results shown in the next section into perspective, we can estimate the cost of preparing an initial Hartree-Fock state in our nuclear lattice model assuming access to an oracle that prepares the $\eta$ single-particle orbitals we want. Using the strategy of Ref.~\cite{babbush2023quantum} the Hartree Fock state could then be prepared in first quantization using $O(\eta\Omega)$ $T$ gates and $O(\eta)$ ancilla space. An alterative strategy proposed recently in Ref.~\cite{Rule2026recursivealgorithm} allows, for large lattice sizes, to trade $T$ gates for space achieves $O(\eta^2\sqrt{\Omega})$ gate complexity while requiring $O(\sqrt{\Omega})$ ancilla space. The authors of Ref.~\cite{Rule2026recursivealgorithm} estimated that of the order of $O(10^{5})$ $T$ gates would be needed for preparing a simple Hartree-Fock state on a $64\times 64\times 64$ lattice with up to $\eta\approx 70$. Compared to the cost of time evolution described in more detail in the next section, this would constitute only a small fraction of the overall cost of the simulation. Finally we want to point out that, for sufficiently well localized spatial orbitals, it could be possible to initialize an Hartree Fock state with a logarithmic cost in $\Omega$ using the protocol proposed in Ref.~\cite{PRXQuantum.6.020319} which achieves a $T$ complexity in $O(\eta^2\log(\Omega))$ with a prefactor that depends on the number of localized basis functions used to obtain the single particle orbitals.}

\section{Results}
\label{sec:results}

In this section, we present numerical results for both the gate count and the number of qubits using realistic model parameters. In particular, we use the coefficients reported in Tab.~\ref{tab:parameters} which are the same used in previous publications on pionless EFT~\cite{Roggero_2020,watson2023quantumalgorithmssimulatingnuclear}. These correspond to a $SU(4)$ symmetric contact interaction with a lattice spacing $a=1.4$ fm and a nucleon mass $\mu=939$ MeV. We provide resource estimates for consider two separate situations
\begin{itemize}
    \item simulations for a time $t_{cross}$ comparable to the time it would take a nucleon with energy $E=10$ MeV to cross the full lattice (as introduced in Ref.~\cite{watson2023quantumalgorithmssimulatingnuclear})
\begin{equation}
    t_{cross}(L) = \frac{aL}{\hbar c} \sqrt{\frac{\mu}{2E}}\;,
    \label{equation:crossing-time}
\end{equation}
    where $L=2^m$ is the dimensionless linear size of the lattice. 
    \item simulations for a time $t_r(\Delta \omega)$ required to reconstruct the nuclear response with a resolution of $\Delta \omega$
    \begin{equation}
    t = \left( \left\lceil \frac{\Delta H}{\Delta \omega} \right\rceil -1 \right) \frac{2\pi}{\Delta H}
    \label{equation:simulation_time}
\end{equation}
where $\Delta H = \norm{T} + \norm{V} + 18\eta$ (see Ref.~\cite{Roggero_2020} with the norm estimates is in Appendix \ref{appendix:norm}. In these calculations we consider $\Delta \omega = 100$ MeV.
\end{itemize}

The first type of calculation gives an idea of the effort required to study nuclear collisions or (semi-)-exclusive processes where nucleon propagation needs to be taken into account. The second type of calculation is more amenable to study nuclear cross sections within linear response~\cite{PhysRevC.100.034610,PhysRevA.102.022409,hartse2023faster,PhysRevE.105.055310,PhysRevLett.134.192701}.

\begin{table}[b]
    \centering
    \begin{tabular}{|c|c|c|}
    \hline
        $\frac{\hbar^2}{2\mu a^2}\; [$MeV$]$ & $C\; [$MeV$]$ & $G\; [$MeV$]$ \\
        \hline
        $10.58$ & $-98.23$ & $127.84$ \\
        \hline
    \end{tabular}
    \caption{Hamiltonian parameters for a lattice spacing $a=1.4$ fm extracted from Ref.~\cite{Rokash_2014}.}
    \label{tab:parameters}
\end{table}

As in the main part of the paper above, we will use $\Omega$ to denote the number of possible single-particle states, which will be $4$ (due to spin and isospin) per site of the lattice, meaning $\Omega = 4\cdot 2^{dm}=4L^d$, where $d=3$ is the number of spatial dimensions, and $m$ is the number of qubits used for every spatial dimension. 

\begin{figure}[th]
    \centering
    \includegraphics[width=0.48\linewidth]{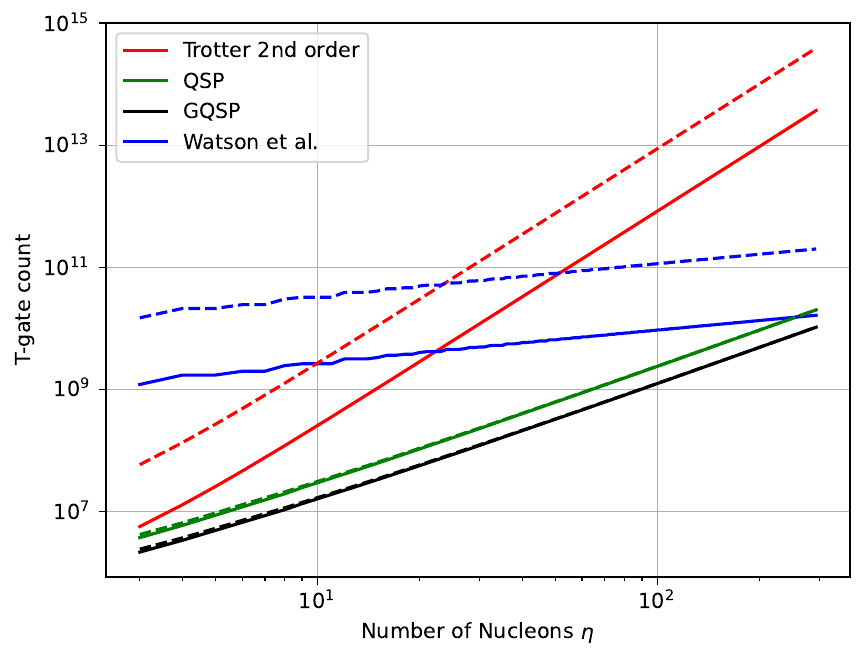}
    \includegraphics[width=0.48\linewidth]{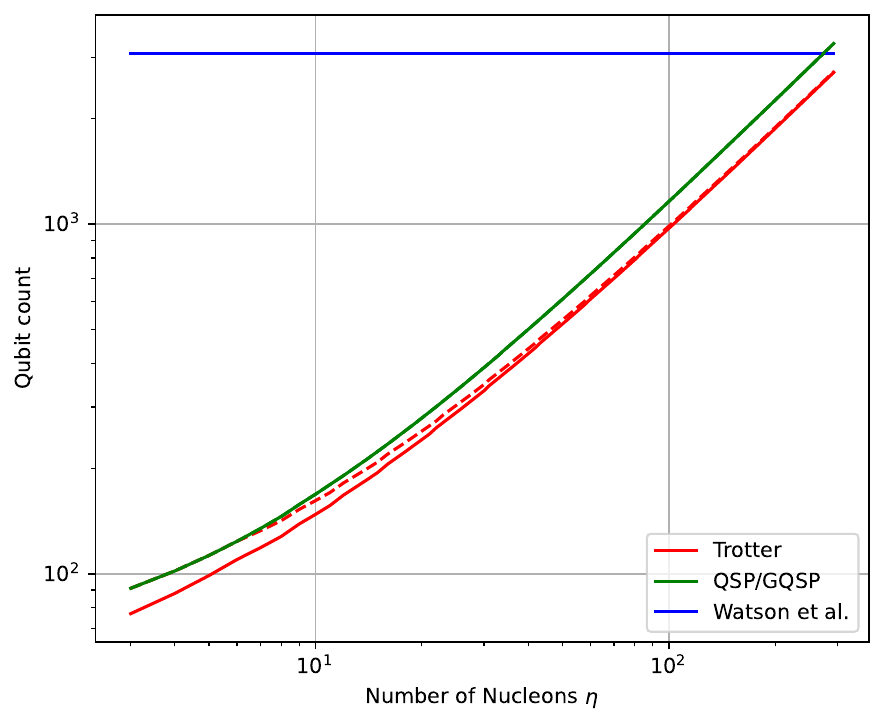}
    \caption{The left (right) panel shows the T-gate (qubit) count as a function of the number of nucleons for the crossing time $t_{cross}$ in Eq.~\eqref{equation:crossing-time} in a $8\times8\times8$ box. The red curves correspond to a second-order product formula in first quantization, the green and black curves correspond to the QSP and GQSP methods in first quantization and, finally, the blue curves are the second-order product formula as implemented in Ref.~\cite{watson2023quantumalgorithmssimulatingnuclear}. For each method we show results for {\  the time $t_{cross}$ from Eq.~\eqref{equation:crossing-time} and} two different target errors: $\epsilon=10^{-1}$ (solid line), and $\epsilon=10^{-3}$ (dashed line).}
    \label{figure:eta_comparison}
\end{figure}

We start from the first type of calculations and show in Fig~\ref{figure:eta_comparison} the $T$-gate and qubit counts in terms of the number of nucleons $\eta$ in the system for $m=3$, corresponding to a $8\times8\times8$ box. We show two different target errors corresponding to $\epsilon=10^{=1}$ (solid lines) and $\epsilon=10^{-3}$ (dashed lines). In this figure we compare three methods: the second-order product formula in either first quantization (red) or second quantization (blue) as well as the QSP and GQSP methods in first quantization (green and black respectively). The implementation for the second quantization scheme follows the previous work by  Watson \textit{et al.} Ref.~\cite{watson2023quantumalgorithmssimulatingnuclear}. We start by noticing that, as expected, the number of qubits required for simulations in first quantization is always smaller than in second quantization, even accounting for the additional ancilla space required to perform the former scheme. For a medium mass nucleus with $\eta=40$ first quantization requires around $420$ qubits while in second quantization one need $3072$ qubits regardless on the number of particles. The gate cost shows instead a more interesting behavior: for the Trotter method the gate cost in second quantization grows more gently with $\eta$ and becomes the preferred scheme for all but the smallest nuclei. This was expected since, as summarized in Tab.~\ref{tab:comparision_results}, for fixed time, error and box size we have
\begin{equation}
C_{1st} = \widetilde{O}\left(\eta^{7/2}\log(\eta)\right)\quad\quad C_{2nd}=O\left(\eta^{1/2}\log(\eta)\right)\;,
\end{equation}
for the first and second quantization product formula respectively.
On the other hand, the GQSP simulation in first quantization appears to outperform the other techniques for all nuclei considered (up to $\eta=294$) even for the largest error tolerance $\epsilon=10^{-1}$. For large enough systems the $O(\eta^2)$ of QSP/GQSP can start  to favor Trotter in second quantization, at least for large error tolerances. In our calculations this happens only for QSP where second-quantized Trotter starts to be preferred at $\eta=255$. Thanks to their logarithmic scaling with the error $\epsilon$, QSP and GQSP will progressively become advantageous as the target error $\epsilon$ is reduced. As we can see on the right panel, this does not come with an appreciable change in the number of qubits required with QSP requiring more logical qubits than the algorithm in second quantization only for $\eta\geq276$. In the right panel we have a single line for both QSP and GQSP since they bothy require exactly the same number of qubits.
We can conclude that, with the parameters we used ($\epsilon = 10^{-1},10^{-3}$ and a small $m=3$ lattice in three spatial dimensions), the GQSP algorithm in first quantization is the preferred scheme since it requires fewer gates, fewer qubits, and it allows us to easily go to smaller errors.

\begin{figure}[t]
    \centering
    \includegraphics[width=0.48\linewidth]{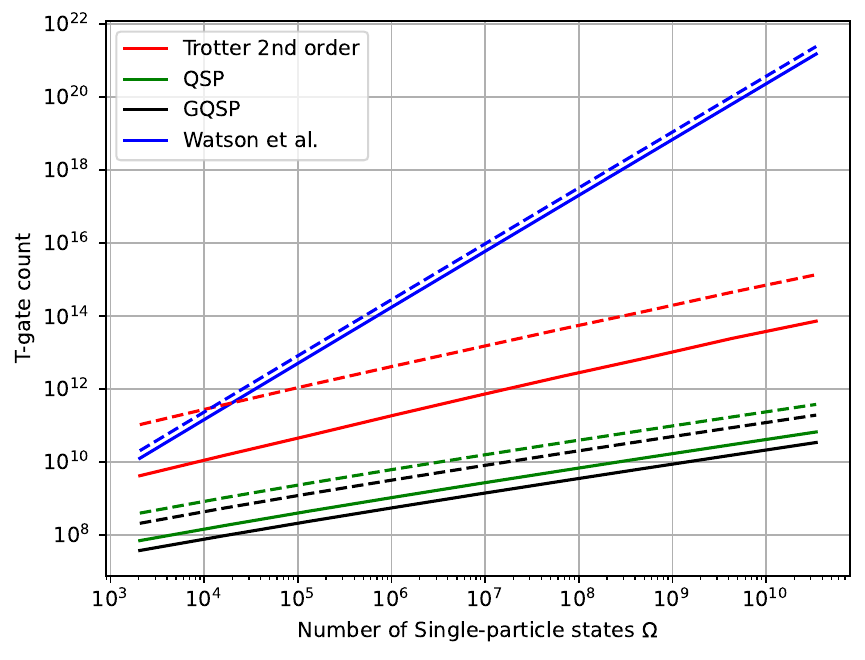}
    \includegraphics[width=0.48\linewidth]{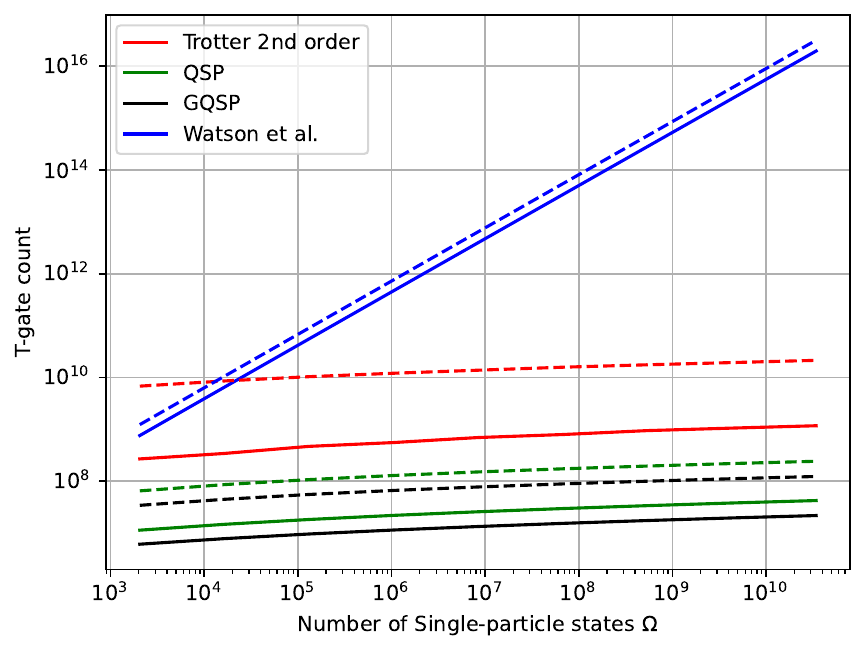}
    \caption{The two panels show the T-gate count as a function of the size of the single particle space $\Omega = 4\cdot 8^{m}$ for $3\leq m\leq 12$. The left panel shows simulations up to $t_{cross}$ from Eq.~\eqref{equation:crossing-time} while the right panel is for simulations up to $t_{r}(\Delta\omega)$ from Eq.~\eqref{equation:simulation_time} with $\Delta\omega=100$ MeV. The color coding is the same as in Fig.~\ref{figure:eta_comparison} but here the solid line correspond to a sytem with $\eta=16$ while the dashed line to $\eta=40$. In all cases we set $\epsilon=10^{-2}$.
    }
    \label{figure:omega_comparison}
\end{figure}

Of course, the main appeal of formulating this problem in first quantization is the favorable scaling of the cost with the size of the single particle space $\Omega$. To show the impact of this, we present on the left panel of Fig~\ref{figure:omega_comparison} the gate count as a function of the number of single-particles states $\Omega = 4\cdot 2^{dm}$ for $d=3$, total time equal to $t_{cross}$, target error $\epsilon=10^{-2}$ and two separate nuclei with $\eta = 16$ (solid lines) and $\eta=40$ (dashed lines). The results span a range going from $m=3$, corresponding to the $8^3$ lattice used for the previous results, to $m=12$ which correspond to a $4096^3$ lattice. For reference, for the first quantization algorithms we reach the same qubit register need for the second quantization scheme for $m=24$ at $\eta=40$ and at $\eta=16$ for $m=54$ (a $10^{16}$ cubed lattice) . It is important to notice that the time $t_{cross}$ used to produce these results grows with the single particle space size as $t_{cross}\propto \Omega^{1/3}$ in $d=3$ spatial dimensions. Using then the results from Tab.~\ref{tab:comparision_results} we see that, for fixed error and particle number the gate cost scales as
\begin{equation}
C_{1st}=\widetilde{O}\left(\Omega^{1/2}\right)\quad C_{2nd}=\widetilde{O}\left(\Omega^{3/2}\right)\quad C_{QSP/GQSP}=\widetilde{O}\left(\Omega^{1/3}\right)\;,
\end{equation}
where $C_{1st}$ and $C_{2nd}$ refer to the second-order Trotter in first and second quantization respectively. From the results shown in the left panel of Fig.~\ref{figure:omega_comparison} we clearly see this trend with the GQSP method outperforming all the others. On the right panel of Fig.~\ref{figure:omega_comparison} we show instead the T gate count for $\eta=16$ (solid lines) and $\eta=40$ (dashed lines) but for the total simulation time $t_r(\Delta \omega)$ with $\Delta \omega =100$ MeV which is largely independent on $\eta$ and approximately given by $t_r\approx 0.063$ MeV${}^{-1}$. The difference between methods is now stronger as we expect, from Tab.~\ref{tab:comparision_results}, the following scaling 
\begin{equation}
C_{1st}=\widetilde{O}\left(\log(\Omega)^{2}\right)\quad C_{2nd}=O(\Omega\log\left(\Omega)\right)\quad C_{QSP/GQSP}=O\left(\log(\Omega)^{2}\right)\;.
\end{equation}
The logarithmic scaling with $\Omega$ of the gate cost is now apparent with for instance the gate cost for the $\eta=16$ system in QSP/GQSP changing by less than a factor of $5$ when $\Omega$ increases by more than $8$ orders of magnitude. We can now also clearly see that the second quantization method has a linear scaling in the size of the lattice.

In summary, we conclude that for a purely pion-less nuclear model the second quantization approach brings an advantage only in the limit of high $\eta$ and small $m$, which is the case of a lattice system where the majority of sites is full. When considering scattering simulations, we are mostly interested in the opposite limit, since we have few particles, at most hundreds of nucleons, and very large lattices. In this setting, the first quantization approaches appear to be cheaper both in $T$-gate and qubit count in every parameter range (given $\eta/\Omega \ll 1$). In particular, the QSP/GQSP algorithm brings a significant advantage in the gate count even in regimes where both target errors are not overly small or total times overly large. In order to provide a more concrete resource assessment, we provide in Tab.~\ref{tab:results_numbers_cross} both T gate and logical qubit counts for a minimal simulation in a $8\times8\times8$ spatial lattice with error $\epsilon=0.1$ and two systems with $\eta=16$ and $40$ respectively. For all instances GQSP provides the best gate count at the price of a modest increase in the required number of logical qubits. The results shown for GQSP approaching values that suggest these simulations are reaching an efficiency that could make them candidates for early fault-tolerant regimes.

\begin{table}[h]
    \centering
    \setlength{\tabcolsep}{4pt}
    \begin{tabular}{|c|c|c|c|c|c|}
        \hline
        Method & $\eta$ & T-count [$t_{cross}$]& qubits [$t_{cross}$]& T-count [$t_{\Delta\omega}$]& qubits [$t_{\Delta\omega}$]\\
        \hline
        Trotter 2nd order & 16 & $3.64\cdot 10^9$ & 3072& $2.15\cdot 10^8$ & 3072\\\cline{2-6}
        2nd-quantization \cite{watson2023quantumalgorithmssimulatingnuclear}& 40 & $5.89\cdot 10^9$ & 3072& $3.47\cdot 10^8$ & 3072\\
        \hline
        Trotter 2nd order & 16 & {\  $1.3\cdot 10^9$} & {\  $196$} & {\  $8.2\cdot 10^7$} & {\  $190$} \\\cline{2-6}
        1st-quantization& 40 & {\  $3.3\cdot 10^{10}$} & {\  $416$} & {\  $2.1\cdot 10^9$} & {\  $410$} \\
        \hline
        QSP & 16 & {\  $7\cdot 10^7$} & $235$ & {\  $1.12\cdot 10^7$} & $234$ \\\cline{2-6}
        1st-quantization& 40 & {\  $4\cdot 10^8$} & $500$ & {\  $6.4\cdot 10^7$} & $498$ \\
        \hline
        GQSP & 16 & $3.75\cdot 10^7$ & $235$ & {\  $5.99\cdot 10^6$} & $234$ \\\cline{2-6}
        1st-quantization& 40 & $2.11\cdot 10^8$ & $500$ & $3.38\cdot 10^7$ & $498$ \\
        \hline
    \end{tabular}
    \caption{Summary of the gate and qubit cost for simulations up to $t_{cross}$ (columns 3 and 4) and $t_{\Delta\omega}$ (columns 5 and 6) for four methods and two different numbers of particle $\eta=16$ and $\eta=40$. All estimates use a $8\times8\times8$ lattice and target error $\epsilon=0.1$.}
    \label{tab:results_numbers_cross}
\end{table}

\section{Summary and conclusions}
\label{sec:conclusions}

In this work, we compared various algorithms for the simulation of nuclear dynamics at low energies using a minimal model of the nuclear Hamiltonian based on pion-less Effective Field Theory. When formulated on a spatial lattice, this model shares many similarities with the standard Fermi Hubbard model but it includes more fermionic species and three-body interactions. The other main difference with simulations of the Hubbard model, or the electron gas, is that for those models one is often interested in keeping the particle density fixed as the volume is made larger while when studying nuclear reactions we are instead interested in studying systems with a fixed particle number in the limit of large volumes. This provides a great opportunity for a formulation of the problem in first quantization where the memory requirement scales only as the logarithm of the volume. In our first two main results, Theorem~\ref{theorem:first_order_Trotter} and Theorem~\ref{theorem:2nd_and_fourth_Trotter}, we provided a complete resource estimation for first, second and fourth order product formulas. We achieve this by proposing an efficient strategy to apply the evolution under the potential energy operator using an equality testing approach in Lemma~\ref{theorem:complact_Trotter_potential_energy} and by carefully estimating the commutator bounds for product formulae using a generalization of fermionic semi-norms in first quantization. We have also provided a full block encoding of the pionless EFT Hamiltonian in Theorem~\ref{theorem:LCU_pionless} using a standard approach from Ref.~\cite{Su2021} for the kinetic energy but develop a novel block encoding for contact interactions that exploits the fermionic nature of the system allowing for an efficient implementation with only $O(\eta)$ gate cost even in the presence of two and three-body interactions. This block encoding is then used together with Quantum Signal Processing in Theorem~\ref{theorem:QSP} to provide a complete resources estimation for simulations of time evolution of pionless EFT in first quantization with a near optimal simulation method. In terms of the number of gates, the scaling is nearly linear in time as $\widetilde{O}(t)$, polylog in error as $O(\log(1/\epsilon)^2)$ and at most quadratic in the number of particles $\eta$. A full comparison of the gate and logical qubit scaling of a various techniques for the simulation of time-evolution in the pionless EFT model is provided in Tab.~\ref{tab:comparision_results} above. For a fixed total time $t$ and target error $\epsilon$, the main difference between first and second quantization schemes is the scaling with the system size: a product formula of order $2p$ requires $\widetilde{O}(\eta^{1/2p}\,\Omega)$ T-gates and $O(\Omega)$ logical qubits while in first quantization these become $\widetilde{O}(\eta^{1+1/2p}\log(\Omega)(\eta^2+\log(\Omega)))$ and $\widetilde{O}(\log(\Omega)(\eta+\log(\Omega)))$ respectively. Using QSP in first quantization brings the gate cost down to $\widetilde{O}(\eta^2\log(\Omega)$ on $O(\eta\log(\Omega))$ logical qubits. {\  All gate costs presented here are expressed in terms of $T$-gate count.  We expect however this to be a good estimate for the overall scaling when considering also entangling Clifford gates.}

In order to go beyond a study of the asymptotic scaling, in Sec.~\ref{sec:results} we compared directly the cost of the second order product formula in both first and second quantization as well as QSP in first quantization for system parameters appropriate for a minimal simulation of nuclear scattering at low energies. The results show that consistently formulating the problem in first quantization not only requires exponentially fewer qubits to encode the system, but it also requires fewer $T$-gates to be implemented. Due to a stronger dependence on the number of particles $\eta$ of the first quantization methods we expect that there is a regime where second quantization would be preferred. From our estimations, even for a small $8^3$ lattice, the 2nd order formula in first quantization becomes more expensive than the second quantization version for $\eta\approx 10-20$ while the crossing point for QSP happens only for large systems with $\eta\approx 200$.

The cost estimates reported in Tab.~\ref{tab:results_numbers_cross} suggest that early studies of low energy nuclear scattering on simple targets could be performed with around $10^7$ T gates on a few hundred logical qubits (e.g. 212 for ${}^{16}$O) which is possibly within reach for early fault-tolerant quantum devices. As a comparison, this is about one order of magnitude more expensive in terms of non-Clifford gates than Quantum Phase Estimation on the Fermi Hubbard model~\cite{Campbell_2021}, one of the paradigmatic early fault-tolerant applications of quantum simulations. Using the second quantization based implementation of Ref.~\cite{watson2023quantumalgorithmssimulatingnuclear} will require for the same systems at least one order of magnitude more $T$ gates and logical qubits and is then more comparable to ground state simulation of molecules~\cite{Lee_2021,doi:10.1073/pnas.2203533119} or factoring of 2048-bit RSA integers~\cite{Gidney2021howtofactorbit}. 

It is important to note that the estimates presented here deal only with the time evolution portion of the scattering simulation. The progress that we were able to demonstrate calls now for a full end-to-end resource estimation for nuclear scattering problems to assess the practical viability of studying these physical process in the early fault-tolerant era. Future steps towards this goal will need to include the cost of state preparation of the initial nuclear target and an estimate of the number of circuit repetitions required to extract physical observables. {\  We have provided in Sec.~\ref{sec:state_prep} a short discussion about the possible strategies that can form a starting point for future work on costing the state preparation part,at least in cases where one is interested in the linear response of the nuclear target. The number of circuit repetition will also significantly depend on the specific observables that are targeted in the simulation. For general exclusive processes with many particles involved in the reaction the number of final states becomes rapidly too large to allow for accurate quantification of the individual cross sections and more research will be needed to better understand which measurable observables could be computed efficiently for these processes. For inclusive cross sections instead, one possible strategy would be to measure directly either the Chebyshev or the Fourier moments of the response function as described e.g. in Refs.~\cite{PhysRevA.102.022409,PhysRevE.105.055310,PhysRevD.111.034504}. For typical applications in nuclear physics, the improved estimators from Ref.~\cite{hartse2023faster} suggest that $O(10^3)$ Fourier moments would be sufficient. Measuring these moments individually to a percent level precision would require $O(10^4)$ repetitions for each moment but in the fault tolerant regime it could be possible to trade repetitions for additional gate depth using strategies like the one proposed in Ref~\cite{PhysRevLett.129.240501}.}  It will be also important to explore the implementation of improved Hamiltonians with more complex interactions which are needed for higher accuracy and for higher collision energies~\cite{RevModPhys.92.025004}. Even though the results presented here are specific to the SU(4) invariant Hamiltonian from Eq.~\eqref{equation:Hamiltonian}, the construction we follow can be readily extended to more accurate models that also include an $SU(4)$ breaking term that distinguishes the spin singlet and triplet states {\ (see Appendix.~\ref{app:su4_break})}. 

Several further optimization are still available to reduce the final gate costs of the methods presented here. For instance, the exponentiation of the potential energy operator can be sped up from $O(\eta^3)$ to $O(\eta^2)$ using a variant of the trick used to construct the block encoding in Lemma~\ref{lemma:LCU_V} where we use the expression in Eq.~\eqref{eq:pot_compact} for the potential matrix elements and store the inner sum over $j$ in a two qubit register for each particle $i$ and finally perform a phase kickback with {\  its} value. This scheme will incur in an increase in ancilla space of $O(\log(\eta))$ and we didn't fully pursue it here because from preliminary estimates the QSP-based scheme was still performing better. Of course this could change if better bounds on the Trotter error are pursued~\cite{Campbell_2021,PhysRevB.108.195105,baysmidt2025faulttolerantquantumsimulationgeneralized}. The choice of kinetic energy operator will play a role in this: here we used a QFT based implementation of the kinetic energy $T$ but since we are working with finite UV regulators anyway it is also possible to use finite differences approximations which could lead to smaller Trotter errors~\cite{watson2023quantumalgorithmssimulatingnuclear} and might still be implemented efficiently even at high orders~\cite{Kivlichan_2017}.

In addition, our main focus in this has been the optimization of the total number of $T$-gates but for practical early fault-tolerant calculations it might be helpful to also minimize the gate depth. For instance, for the block encoding of the potential energy we use the $f$ register in Eq.~\eqref{eq:lcu_prep_statet_T} in unary and employ the swap network construction shown in Fig.~\ref{fig:controlled_SN_example_circuit} which uses $\eta$ ancilla qubits, and $(\eta-1)dm$ controlled-swap gates. In order to achieve the same task we could have used instead the $f$ register in bit encoding and use a swap network with an ancilla register of $\lceil \log_2 \eta \rceil$ qubits, and $2^{\lceil \log_2\eta \rceil}dm$ controlled-swap gates~\cite{Kivlichan_2017}. The alternative circuit could cost a factor of $2$ more than the previous one, but it would save ancilla qubits (from $\eta$ to $\lceil \log_2\eta \rceil$), and reduce the depth from $O(\eta)$ to $O(\log\eta)$. For the Trotter evolution of the potential we have to perform a large number of rotations with fixed angles, $O(\eta^2)$ per step for the two body and $O(\eta^3)$ per step for the three body. By parallelizing the application of these rotations, which will also reduce the gate depth, we can then leverage techniques like Hamming Weight Phasing~\cite{Gidney_2018,Kivlichan2020improvedfault} to reduce the non-Clifford cost of implementing the rotations. As a final example,in our algorithm, we use many $N$-controlled Toffoli gates, which with the implementation we considered (see Theorem~\ref{theorem:MCX} in Appendix~\ref{app:general_subroutines}) have linear depth. A recent proposal from Ref.~\cite{nie2024quantumcircuitmultiqubittoffoli} would require instead a single ancilla qubit and reduce the depth to logarithmic in $N$. We want to conclude by noting that the focus in our resource estimation has been concentrated solely on minimizing $T$ gates, as these operation become cheaper and other quantum error correcting besides the surface code receive increased attention~\cite{xu2024constant,Bravyi_2024,yoder2025tourgrossmodularquantum} future work in design of quantum algorithms for nuclear reactions could take a different approach for optimizing the overall cost.

\section{Author contribution}
All authors worked on the derivations and algorithm design described in this work, L. S. carried out the numerical simulations and A. R. formulated the problem and supervised the work. We acknowledge the use of generative AI in writing the python code used to generate figures.

\section{Acknowledgments}
We want to thank J. Carlson for stimulating discussions about simulating low energy nuclear scattering, C. Kane for discussions on the first draft of the paper and I. Chernyshev for pointing out a mistake in our original commutator bounds. This project was supported by the Provincia Autonoma di Trento, and by Q@TN, the joint lab between the University of Trento, FBK-Fondazione Bruno Kessler, INFN-National Institute for Nuclear Physics and CNR-National Research Council.



\appendix

\section{Subroutines}
\label{app:general_subroutines}
In this Appendix we provide a collection of useful results from other works that we will use repeatedly in our derivations.

\begin{theorem}[Squaring, Lemma 7~\cite{Su2021}]
    Consider an $N$-qubit register. The cost in T gates to compute its square in a $(2N)$-qubit register is given by
    \begin{equation}
        T_{SQU}(N) = 4N^2 - 4N
    \end{equation}
    and requires $N(N-1)$ ancilla qubit.
    \label{theorem:multiplication_better}
\end{theorem}

\begin{theorem}[Approximate Quantum Fourier Transform, \cite{Nam_2020}]
    Let us denote with $QFT$ the matrix implementing the Quantum Fourier Transformation on $N$ qubits, and let $U$ be a unitary such that $\norm{QFT - U} \le \epsilon$. The cost in number of T gates of implementing $U$ for $N>2$ and $\epsilon<3/4$ is given by
    \begin{equation}
    \begin{split}
        T_{QFT}(N,\epsilon) &=7N-11+\sum_{n=3}^{N-1}\left(8\min\left(\left\lceil\log_2\left(\frac{N}{\epsilon}\right)\right\rceil,n\right)-15\right)  \\
        &=O\left(N\min\left(\log\left(\frac{N}{\epsilon}\right),N\right)\right)
    \end{split}
    \end{equation}
    The algorithm also requires the availability of a phase-gradient state over $b=(\min(N-1,\lceil\log_2(N/\epsilon)\rceil)+1)$ qubits which needs to be implemented only once, employing $O(b\log(b/\epsilon))$ T gates, and can be reused. The total number of ancilla required is $3b-1$.
    \label{theorem:QFT}
\end{theorem}

\begin{theorem}[Z Rotations, \cite{Kliuchnikov_2023}]
    Let $R_Z(\theta) = e^{-i\theta Z/2}$. A $1$-qubit gate $U_Z(\theta)$ can be implemented such that $\norm{R_Z(\theta) - U_Z(\theta)} \le \epsilon$ with a $T$-gate cost of:
    \begin{equation}
        T_{ROT}(\epsilon) = 0.57\log_2\left( \frac{1}{\epsilon} \right) + 8.83
    \end{equation}
    \label{theorem:Z_rotations}
\end{theorem}

\begin{theorem}[Multi-controlled Toffoli]
    It is possible to implement an $N$-controlled Toffoli gate (which appies an $X$ gate if and only if all $N$ controls are in state $1$) with a $T$-gate count of
    \begin{equation}
        T_{MCX}(N) = 4(N-1)
    \end{equation}
    and $N-1$ ancilla qubits
    \label{theorem:MCX}
\end{theorem}

\begin{proof}
    Let us assume to have $N-1$ clean ancilla qubits. We can decompose the $N$-controlled Toffoli in a ladder of $N-1$ Toffoli gates, each having a different ancilla qubit as target. Then, we could use the Gidney construction \cite{Gidney_2018} to decode all ancilla qubits. This would mean that the $N$-controlled Toffoli requires $N-1$ Toffoli gates, and each of them can be done with $4$ $T$ gates. It follows that we need a total number of $T$-gates equal to
    \begin{equation}
        T_{MCX}(N) = 4(N-1)
    \end{equation}
    and $N-1$ ancilla qubits.
\end{proof}

\begin{theorem}[Multi-controlled Z Rotation]
    It is possible to implement an $N$-controlled $R_Z(\theta)$ gate up to error $\epsilon$ by using two $N$-controlled Toffoli gates and two rotations. Using Theorem~\ref{theorem:MCX} and Theorem~\ref{theorem:Z_rotations} this amounts to a total number of T gates equal to
    \begin{equation}
        T_{MCRZ} (\epsilon, N) = 8(N - 1) + 2T_{ROT}\left( \frac{\epsilon}{2} \right)
    \end{equation}
    where $T_{ROT}(\epsilon)$ is given in Theorem \ref{theorem:Z_rotations}. The required number of ancilla qubits is $N-1$.
    \label{theorem:MCRZ}
\end{theorem}

\section{Phase kick-back method}
\label{appendix:kickback}

The phase kick-back method is a standard procedure to implement diagonal unitaries when the phase can be calculated efficiently (see e.g.~\cite{Kassal_2008,Jones_2012}).
This method requires using one quantum Fourier transform (QFT) at the beginning of the simulation and a single oracle call per step to a unitary $Q_H$ defined below.
Let us consider a general qubit register with $b$ qubits. We initialize the register state as the Fourier transform of the $\ket{0\cdots001}_b$ state, leading to:
\begin{equation}
\label{eq:onebar}
    \rvert \bar{1}\rangle_b =\frac{1}{2^{b/2}}\sum_{j=0}^{2^{b}-1}e^{-i\frac{2\pi j}{2^b}}\rvert j\rangle\;.
\end{equation}
Performing the bitstring addition $\bar 1 \oplus K$ (modulo $2^b$) to this state, with $K$ an arbitrary integer, results in
\begin{equation}
\rvert \bar{1}\oplus K\rangle = \frac{1}{2^{b/2}}\sum_{j=0}^{2^{b}-1}e^{-i\frac{2\pi j}{2^b}}\rvert j\oplus K\rangle
    = \frac{e^{i\frac{2\pi K}{2^b}}}{2^{b/2}}\sum_{j=0}^{2^{b}-1}e^{-i\frac{2\pi j}{2^b}}\rvert j\rangle={e^{-i\frac{2\pi K}{2^b}}}\rvert\bar{1}\rangle\;.
\end{equation}

The goal is to approximate the evolution for a time $t$ under the diagonal Hamiltonian $H = \sum_m \rvert m \rangle \lambda_m \langle m \lvert$, where we assume the eigenvalues are ordered as $\lambda_m\leq\lambda_{m+1}$ for all $m$. In order to simplify the analysis, we will assume that $\lambda_m\geq0$ (this can always be ensured by adding a global phase). The goal is then to find a unitary $Q_H$ acting on both the system register and the $b$ ancillae such that
\begin{equation}
\label{eq:phkick_err}
\left\|e^{-itH} - \langle \bar{1}\lvert Q_H \rvert\bar{1}\rangle\right\|\leq \epsilon\;.
\end{equation}
Using the fact that $H$ is diagonal in the computational basis, we can then use a controlled unitary adder
\begin{equation}
\label{eq:qh_oracle}
Q_H \rvert m\rangle\rvert z\rangle = \rvert m\rangle\left\rvert z\oplus K_m\right\rangle\;,
\end{equation}
with a $b$-bit integer $K_m$ chosen such that
\begin{equation}
\left|t\lambda_m - \frac{2\pi}{2^b}K_m\right|\left(mod\; 2\pi\right)\leq \epsilon\;.
\end{equation}
This is a direct consequence of the zero$^{th}$-order bound $\|e^{iX}+\mathbb{1}\|\leq \|X\|$, valid for any Hermitian bounded operator $X$. This can be achieved if we select $K_m$ to be
\begin{equation}
K_m = \left\lfloor \frac{2^b}{2\pi}t\lambda_m\right\rceil\left(mod\; 2^b\right)\;,
\end{equation}
so that the error is bounded by
\begin{equation}
\epsilon\leq \frac{2\pi}{2^{b+1}}\;.    
\end{equation}
We can then choose the precision as
\begin{equation}
b\geq \left\lceil\log_2\left(\frac{\pi}{\epsilon}\right)\right\rceil\;.
\end{equation}

A slight generalization of the construction is the situation where the eigenvalue $\lambda_m$ itself is being approximated by a $b_\lambda$-bit integer $l_m$ so that
\begin{equation}
\label{eq:prec_lambdam}
\max_m \left|2^{b_\lambda}\frac{\lambda_m}{\Lambda}-l_m\right|< 1\;,
\end{equation}
for some $\Lambda\geq \max_m \lambda_m=\|H\|$. A naive way of proceeding in this case would be to first approximate the constant $\gamma=t\Lambda/(2\pi)$ using $b_\gamma$ bits, then form the multiplication $\mu_m=2^{b-b_\lambda}\gamma l_m$ and finally add the result to the $b$-bit register. As pointed out already in recent work~\cite{Sanders_2020}, a more efficient strategy is to  use instead the fact that the bit representation of the classical constant $\gamma$ is known to simply add directly the integer $l_m$, with the appropriate shift, for every non-zero component of $\gamma$ into the ancilla register containing the $\rvert\bar{1}\rangle$. We can then show the following result
\begin{lemma}[Phase kickback with approximate eigenvalue]
\label{lemma:phkick} Let $t>0$, $0<\epsilon<1$, $H$ a $q$-qubit diagonal Hamiltonian and $\Lambda\geq\|H\|$ an upperbound on its spectral norm. Let also $O_{\lambda;b_\lambda}$ be an oracle that performs
\begin{equation}
O_{\lambda;b_\lambda}\rvert m\rangle_q\rvert0\rangle_{b_\lambda} =\rvert m\rangle_q\rvert l_m\rangle_{b_\lambda}\;,
\end{equation}
with $\rvert m\rangle$ a $q$-qubit input state and $\rvert l_m\rangle$ a $b_\lambda$-qubit state satisfying Eq.~\eqref{eq:prec_lambdam}. We can implement a $Q_H$ oracle that can be used, together with a $b$-qubit phase register initialized in $\rvert\bar{1}\rangle$ as in Eq.~\eqref{eq:onebar}, to approximate $U(t)=\exp(-itH)$ with error less than $\epsilon$ choosing
\begin{equation}
b=\left\lceil\log_2\left(\frac{3t\Lambda}{\epsilon}\log_2\left(\frac{3t\Lambda}{\epsilon}\right)\right)\right\rceil\;,
\end{equation}
qubits for the phase register and
\begin{equation}
b_\lambda = \left\lceil\log_2\left(\frac{3t\Lambda}{\epsilon}\right)\right\rceil\;,
\end{equation}
qubits of precision for the $O_{\lambda;b_\lambda}$ oracle. If the eigenvalue can be represented exactly with $b_\lambda$ bits, then one can choose instead
\begin{equation}
b=\left\lceil\log_2\left(\frac{2t\Lambda}{\epsilon}\log_2\left(\frac{2t\Lambda}{\epsilon}\right)\right)\right\rceil\;,
\label{eq:b_exact}
\end{equation}
The algorithm requires one call to $O_{\lambda;b_\lambda}$, one to $O^\dagger_{\lambda;b_\lambda}$ together with
\begin{equation}
T_{DIAG}=4b\min\left(\left\lceil\frac{b_\lambda+1}{2}\right\rceil,w_H\right)=O\left(\log^2\left(\frac{t\Lambda}{\epsilon}\right)\right)\;,
\end{equation}
$T$ gates and $b$ ancilla qubits. Here $w_H$ is the Hamming weight of  $\gamma=t\Lambda/(2\pi)$ written with $b_\lambda$ bits.
\end{lemma}
\begin{proof}
We first write the constant $\gamma=t\Lambda/(2\pi)$ as\begin{equation}
\label{eq:gamma_bin}
\gamma = \sum_{k=1}^\infty \gamma_k 2^{d_\gamma-k},
\end{equation}
with $\gamma_k=\{0,1\}$ and $d_\gamma=\lceil\log_2(\gamma)\rceil$. We then have
\begin{equation}
\begin{split}
\mu_m &= 2^{b-b_\lambda}\gamma l_m = \sum_{k=1}^\infty \gamma_k l_m 2^{d_\gamma+b-b_\lambda-k} =\widetilde{\mu_m} + \sum_{k=n+1}^\infty \gamma_k l_m 2^{d_\gamma+b-b_\lambda-k}\;,
\end{split}
\end{equation}
where we have introduced $\widetilde{\mu_m}$ as the approximation obtained by truncating the binary expansion of $\gamma$ to only $n$ bits of precision. The error can be bounded as
\begin{equation}
\begin{split}
\left|\mu_m-\widetilde{\mu}_m\right|&=\sum_{k=n+1}^\infty \gamma_k l_m 2^{d_\gamma+b-b_\lambda-k}\\
&<\sum_{k=n+1}^\infty \gamma_k 2^{d_\gamma+b-k}\\
&<\sum_{k=n+1}^\infty 2^{d_\gamma+b-k}=2^{d_\gamma+b-n}\;.
\end{split}
\end{equation}
Finally, we perform the multiplication by truncating $l_m/2^{b_\gamma}$ to $b$ bits of precision and add that to the $b$-bit ancilla register containing $\rvert\bar{1}\rangle$. We therefore introduce
\begin{equation}
\widetilde{\widetilde{\mu_m}} = \sum_{k=1}^n \gamma_k 2^{d_\gamma}\lfloor 2^{b-b_\lambda}l_m\rfloor\;,
\end{equation}
whose error in approximating $\widetilde{\mu_m}$ is bounded by
\begin{equation}
\begin{split}
\left|\widetilde{\mu_m}-\widetilde{\widetilde{\mu_m}}\right|&=\sum_{k=1}^n \gamma_k 2^{d_\gamma}\left(2^{b-b_\lambda}l_m-\lfloor 2^{b-b_\lambda}l_m\rfloor\right)\\
&<n2^{d_\gamma}\;.
\end{split}
\end{equation}

We can now sum all the errors together using
\begin{equation}
\left|\frac{2^{b}}{2\pi}t\lambda_m-\mu_m\right|=\gamma2^b\left|\frac{\lambda_m}{\Lambda}-\frac{l_m}{2^{b_\lambda}}\right|<\gamma2^{b-b_\lambda}\;,
\end{equation}
and also the errors from truncations
\begin{equation}
\left|\mu_m-\widetilde{\widetilde{\mu_m}}\right|<2^{d_\gamma}\left(n+2^{b-n}\right)<\gamma\left(n+2^{b-n}\right)\;.
\end{equation}

If we then use $K_m=\widetilde{\widetilde{\mu_m}}$ in the unitary $Q_H$ from Eq.~\eqref{eq:qh_oracle} the total error in the phase kickback approximation is
\begin{equation}
\begin{split}
\left\|e^{-itH} - \langle \bar{1}\lvert Q_H \rvert\bar{1}\rangle\right\|&<\frac{2\pi}{2^b}\gamma\left(n+2^{b-b_\lambda}+2^{b-n}\right)\\
&=\frac{t\Lambda}{2^b}\left(n+2^{b-b_\lambda}+2^{b-n}\right)
\end{split}
\end{equation}

At this point, in order to guarantee a total error below $\epsilon$ it is sufficient to make the following choices
\begin{equation}
b_\gamma=n=\left\lceil\log_2\left(\frac{3t\Lambda}{\epsilon}\right)\right\rceil\;,
\end{equation}
together with a phase register of size
\begin{equation}
b=\left\lceil\log_2\left(\frac{3t\Lambda}{\epsilon}\log_2\left(\frac{3t\Lambda}{\epsilon}\right)\right)\right\rceil\;.
\end{equation}
This can be easily seen by noticing that, for this choices,
\begin{equation}
n+2^{b-b_\lambda}+2^{b-n}<1+\log_2\left(\frac{3t\Lambda}{\epsilon}\right)+2^b\frac{2\epsilon}{3t\Lambda}\;.
\end{equation}

In summary, the complete implementation of the phase kickback scheme proceeds then as follows: 
\begin{enumerate}
    \item we use the oracle $O_{\lambda,n}$ on a target register of size $n=\left\lceil\log_2\left(\frac{3t\Lambda}{\epsilon}\right)\right\rceil$ to produce
    \begin{equation}
    O_{\lambda;n} \rvert m\rangle_q\otimes\rvert0\rangle_n = \rvert m\rangle_q\otimes\rvert l_m\rangle_n\;.
    \end{equation}
    \item we precompute classically an approximation to $\gamma=t\Lambda/(2\pi)$ with $n=\left\lceil\log_2\left(\frac{3t\Lambda}{\epsilon}\right)\right\rceil$ bits of precision
    \begin{equation}
    \widetilde{\gamma}=\sum_{k=1}^n \gamma_k 2^{d_\gamma-l}\;,
    \end{equation}
    with $d_\gamma=\lceil\log_2(\gamma)\rceil$.
    \item for every non-zero bit in the above binary expansion of $\gamma$ we add to the $b$-qubit phase register the corresponding bit shifted version of the integer $l_m$ into the phase register. This requires a number of $b$-bit adders given by the Hamming weight $w_H=\sum_{k=1}^n\gamma_k$ of $\widetilde{\gamma}$. Using Gidney's adder~\cite{Gidney_2018}, this can be done using less than $4w_Hb$ $T$ gates, $b$ ancillas and additional Clifford gates.
    Finally, if the Hamming weight $w_H$ is larger than $\lceil(n+1)/2\rceil$ one can use the strategy described in Appendix A of~\cite{Sanders_2020} to requiring exactly this number of additions and subtractions.
    \item we complete the procedure by uncomputing the eigenvalue using $O_{\lambda,n}^\dagger$ to restore the ancillas.
\end{enumerate}

The total gate and ancilla cost follows by adding all the contributions described above.
\end{proof}

\section{Norms}
\label{appendix:norm}

In the following section we will prove the upper bounds to the norm of the commutators and operators that are used in the implementation of {\  Product} Formulas in Sec.~\ref{ssec:pf}. We start by defining the "fermionic semi-norm", which is the usual norm for operators with the difference that it is evaluated on the anti-symmetric part of the Hilbert space, not on the total Hilbert space. We will prove that it has the right properties and it satisfy the usual inequality relations. With the notion of this norm, we then proceed by setting upper bounds to the kinetic energy, the potential energy, and their commutators.

\subsection{Fermionic Semi-norm}
\label{app:fermionic_seminorm}
The fermionic semi-norm was defined in second quantization by Su et al.\cite{Su_2021seminorm}. It was used to quantify the error of the {\  product} formula algorithm, exploiting that the considered Hamiltonian was particle-preserving. This was done by evaluating the norm on the subspace of the Fock space with fixed particle number. In first quantization, every operator is particle-preserving, thus we do not have such a problem. However, the operators do not explicitly consider the statistics of the system. In fact, the {\  statistics are encoded in the wavefunction}, which is symmetric by exchange if the system is composed of bosons, and it is anti-symmetric if the system is fermionic. We want to extend the fermionic semi-norm in first quantization, and to do it we can follow an equivalent procedure to the one done in second quantization by Su et al.\cite{Su_2021seminorm}, by substituting the projectors on the fixed particle-number subspace, with the projectors on the anti-symmetric subspace.

Let's consider a single-particle Hilbert space $\mathcal{H}^1$, and consider its orthonormal basis $\{\ket{\phi} \}$.
An $\eta$ particle  Hilbert space $\mathcal{H}^{\eta}$ is tensor product of single-particle Hilbert spaces $\mathcal{H}^{\eta} = \bigotimes_{n=1}^{\eta}\mathcal{H}^1$, and its orthonormal basis is $\{\bigotimes_{n=1}^{\eta} \ket{\phi}_n \}$.
In general, a Hilbert space can be divided into symmetric and anti-symmetric subspace $\mathcal{H}^{\eta} = \mathcal{H}_S \oplus \mathcal{H}_A$.
Let's consider the permutation operator $P_{\pi}$, for each permutation $\pi \in S_{\eta}$ in the symmetric group of all permutations of the set $\{ 1,..., \eta\}$, which act as (see e.g.~\cite{watrous_2018})
\begin{equation}
    P_{\pi} \left(\ket{\phi}_1 \otimes \cdot\cdot\cdot \otimes \ket{\phi}_{\eta} \right) = \ket{\phi}_{\pi^{-1}(1)} \otimes \cdot\cdot\cdot \otimes \ket{\phi}_{\pi^{-1}(\eta)} 
\end{equation}
on any state $\ket{\psi} = \ket{\phi}_1 \otimes \cdot\cdot\cdot \otimes \ket{\phi}_{\eta} $ in $\mathcal{H}^{\eta}$. 
The symmetric and anti-symmetric subspaces are defined as
\begin{equation}
    \mathcal{H}_S = \left\{ \ket{\psi} \in \mathcal{H}^{\eta} : P_{\pi}\ket{\psi} = \ket{\psi} \text{ for every } \pi \in S_{\eta}\right\}
\end{equation}
\begin{equation}
    \mathcal{H}_A = \left\{ \ket{\psi} \in \mathcal{H}^{\eta} : P_{\pi}\ket{\psi} = \text{sign}(\pi)\ket{\psi} \text{ for every } \pi \in S_{\eta}\right\}
\end{equation}
and we will denote the corresponding states as $\ket{\psi_S}\in \mathcal{H}_S$ and $\ket{\psi_A}\in \mathcal{H}_A$. 
Consider a $\eta$-tuple $c=\left(a_1, ..., a_{\eta} \right)$, with $a_i\ne a_j$ for every $i,j$, 
we can write the unnormalized states
\begin{equation}
    \ket{\psi_S^c} = \sum_{\pi\in S_{\eta}} P_{\pi} \left(\ket{\phi_{a_1}}\cdot\cdot\cdot \ket{\phi_{a_{\eta}}} \right),
\end{equation}
\begin{equation}
    \ket{\psi_A^c} = \sum_{\pi\in S_{\eta}}\text{sign}(\pi) P_{\pi} \left(\ket{\phi_{a_1}}\cdot\cdot\cdot \ket{\phi_{a_{\eta}}} \right),
\end{equation}
Since $\ket{\psi_S}$ and $\ket{\psi_A}$ span the entire physical Hilbert space, every $\ket{\psi}\in \mathcal{H}^{\eta}$ can be written as a linear combination of them.

Let's focus only on fermionic systems because we are interested in them in this work. We can define the projector operator onto the anti-symmetric subspace $\mathcal{H_A}$ as (see e.g.~\cite{watrous_2018})
\begin{equation}
    \Pi_A = \frac{1}{\eta!}\sum_{\pi\in S_{\eta}} \text{sign}(\pi)P_{\pi}
\end{equation}
which has the following properties:
\begin{itemize}
    \item $\Pi_A = \Pi_A^{\dagger}$,
    \item $P_{\pi}\Pi_A = \text{sign}(\pi)\Pi_A$,
    \item for every distinct indices $i,j \in \{1,...,\eta \}$, $P_{(i j)}\Pi_A = -\Pi_A = \Pi_AP_{(i j)}$, with $(i j)$ being the permutation that swaps only $i$ and $j$.
\end{itemize}
On a symmetric wave function, it acts as
\begin{equation}
    \Pi_A\ket{\psi_S} = \frac{1}{\eta!}\sum_{\pi\in S_{\eta}} \text{sign}(\pi)P_{\pi}\ket{\psi_S} = \frac{1}{\eta!}\sum_{\pi\in S_{\eta}} \text{sign}(\pi\in S_{\eta})\ket{\psi_S} = 0
\end{equation}
since for $\eta>1$ there are as many odd permutations as even ones, while on an anti-symmetric state
\begin{equation}
    \Pi_A\ket{\psi_A} = \frac{1}{\eta!}\sum_{\pi\in S_{\eta}} \text{sign}(\pi)P_{\pi}\ket{\psi_A} = \frac{1}{\eta!}\sum_{\pi\in S_{\eta}} \text{sign}(\pi)^2\ket{\psi_A} = \ket{\psi_A}.
\end{equation}
Therefore, on a general $\ket{\psi} = \sum_i \alpha_i \ket{\phi_{A}}_i + \beta_i\ket{\phi_{S}}_i \in \mathcal{H}^{\eta}$, with $\alpha_i,\beta_i \in \mathbb{C}$, it acts as
\begin{equation}
    \Pi_A \ket{\psi} = \sum_i \alpha_i \ket{\phi_{A}}_i,
\end{equation}
and we will call $\Pi_A \ket{\psi} = \ket{\psi_A}$ which is the anti-symmetric part of $\ket{\psi}$.
Let's now consider a $\eta$-tuple $c=\left(a_1, ..., a_{\eta} \right)$ for which there exist two indices $i,j \in \{1, ..., \eta \}$ for which $a_i = a_j$, then 
\begin{equation}
\begin{split}
    &\Pi_A \ket{\psi_A^c} = \Pi_A P_{(i j)} \ket{\psi_A^c} = - \Pi_A\ket{\psi_A^c} \text{   thus}\\
    &\Pi_A \ket{\psi_A^c} = 0,
\label{eq:pauliexclusion}
\end{split}
\end{equation}
and as expected the Pauli exclusion principle is satisfied.

Let's define the fermionic semi-norm for an operator $O$ symmetric under particle permutation, that is such that $[O,P_\pi] = 0$ for all $\pi\in S_\eta$, as follows
\begin{equation}
    \norm{O}_A = \max_{\ket{\psi_A},\ket{\varphi_A}} \left|\bra{\psi_A}O\ket{\varphi_A} \right|.
\end{equation}
The fermionic semi-norm has the following properties, for a pair $O,O'$ of permutation invariant operators
\begin{enumerate}
    \item $\norm{\lambda O}_A = |\lambda|\norm{O}_A$ with $\lambda \in \mathbb{C}$,
    \item $\norm{O + O'}_A \le \norm{O}_A + \norm{O'}_A$,
    \item $\norm{OO'}_A \le \norm{O}_A\norm{O'}_A$,
    \item $\norm{\mathbb{1}}_A = 1$
    \item $\norm{UOU'}_A = \norm{O}_A$ for $U,U'$ permutation invariant unitaries,
    \item $\norm{O^{\dagger}}_A = \norm{O}_A$
\end{enumerate}
In order to show these results we will follow the same procedure described in Ref.~\cite{Su_2021seminorm}. 
In particular, we will prove the third property, as the remaining follow directly from the definition of the fermionic semi-norm. 
Let's consider
\begin{equation}
    \begin{split}
\norm{OO'}_A &= \max_{\ket{\psi_A},\ket{\varphi_A}} \left|\bra{\psi_A}OO'\ket{\varphi_A} \right| \\
&= \max_{\ket{\psi_A},\ket{\varphi_A}} \left|\bra{\psi_A}O\Pi_A \Pi_A O'\ket{\varphi_A} \right| \\
&\le \max_{\ket{\psi_A}} \norm{\Pi_A O^{\dagger}\ket{\psi_A}}\quad \max_{\ket{\varphi_A}} \norm{\Pi_A O'\ket{\psi_A}}
    \end{split}
\end{equation}
using the Cauchy-Schwarz inequality. 
Now consider an arbitrary state $\ket{\phi} \in \mathcal{H}^{\eta}$
\begin{equation}
    \begin{split}
\max_{\ket{\psi_A}} \norm{\Pi_A O^{\dagger}\ket{\psi_A}} &= \max_{\ket{\psi_A}} \max_{\ket{\phi}} \left|\bra{\phi}\Pi_A O^{\dagger} \ket{\psi_A} \right| \\
&= \max_{\ket{\psi_A}}\max_{\ket{\phi}} \norm{\Pi_A\ket{\phi}} \left| \frac{\bra{\phi}\Pi_A}{ \norm{\Pi_A\ket{\phi}}} O^{\dagger} \ket{\psi_A} \right| \\
&\le \norm{O^{\dagger}}_A = \norm{O}_A
    \end{split}
\end{equation}
assuming $\Pi_A\ket{\phi} \ne 0$ because the case $\Pi_A\ket{\phi}= 0$ never leads to maximality. 
If we also consider
\begin{equation}
    \begin{split}
\norm{O}_A = \norm{O^{\dagger}}_A &= \max_{\ket{\phi_A},\ket{\psi_A}} \left|\bra{\phi_A}O^{\dagger}\ket{\psi_A} \right| \\
&= \max_{\ket{\phi_A},\ket{\psi_A}} \left|\bra{\phi_A}\Pi_A O^{\dagger}\ket{\psi_A} \right| \\
&\le \max_{\ket{\phi},\ket{\psi_A}} \left|\bra{\phi}\Pi_A O^{\dagger}\ket{\psi_A} \right| = \max_{\ket{\psi_A}} \norm{\Pi_A O^{\dagger}\ket{\psi_A}} \\
    \end{split}
\end{equation}
so that $\max_{\ket{\psi_A}} \norm{\Pi_A O^{\dagger}\ket{\psi_A}} = \norm{O}_A$, and also $\max_{\ket{\psi_A}} \norm{\Pi_A O'\ket{\psi_A}} = \norm{O'}_A$, proving the third property. 

We can also prove that the fermionic semi-norm is a projected spectral norm. Therefore, that for any operator $O$ such that $[O,P_\pi]=0$ for all $\pi\in S_\eta$, it holds that
\begin{equation}
    \norm{O}_A = \norm{\Pi_A O \Pi_A}.
\end{equation}
Let's prove it.
Consider
\begin{equation}
    \begin{split}
\norm{O}_A &= \max_{\ket{\phi_A},\ket{\psi_A}} \left|\bra{\phi_A}O\ket{\psi_A} \right| \\
&= \max_{\ket{\phi_A},\ket{\psi_A}} \left|\bra{\phi_A}\Pi_A O \Pi_A \ket{\psi_A} \right| \\
&\le \max_{\ket{\phi},\ket{\psi}} \left|\bra{\phi}\Pi_A O \Pi_A \ket{\psi} \right| = \norm{\Pi_A O \Pi_A}.
    \end{split}
\end{equation}
On the other hand,
\begin{equation}
    \begin{split}
\norm{\Pi_A O \Pi_A} &= \max_{\ket{\phi},\ket{\psi}} \left|\bra{\phi}\Pi_A O \Pi_A \ket{\psi} \right| \\
&= \max_{\ket{\psi}}\norm{\Pi_A\ket{\psi}}\max_{\ket{\phi}}\norm{\Pi_A\ket{\phi}} \left|\frac{\bra{\phi} \Pi_A}{\norm{\Pi_A\ket{\phi}}} O \frac{\Pi_A\ket{\psi} }{\norm{\Pi_A\ket{\psi}}} \right| \\
&\le \max_{\ket{\phi_A},\ket{\psi_A}} \left|\bra{\phi_A}O\ket{\psi_A} \right| = \norm{O}_A
    \end{split}
\end{equation}
always assuming $\Pi_A\ket{\phi}\ne 0$ and $\Pi_A\ket{\psi}\ne 0$ since other choices do not lead to maximality. 
Then
\begin{equation}
    \norm{O}_A = \norm{\Pi_A O \Pi_A}
\end{equation}

We can specialize the result to our fermionic system, where a single particle wave function is composed of spatial, spin, and isospin parts. The space is discretized on a spatial lattice, the spin can be up or down as well as the isospin. 
Hence, a lattice site can host at most $4$ fermions. If we consider a total wave function $\ket{\psi_A}$ where there are more than $4$ fermions in the same lattice site, then $\Pi_A\ket{\psi_A} = 0$, as in Eq.~\eqref{eq:pauliexclusion}. This behavior agrees with the Pauli exclusion principle.

\subsection{Fermionic semi-norm of operators in the Hamiltonian}
\label{appendix_fermionic_seminorm_Trotter_error}

We can use the fermionic semi-norm to evaluate $\norm{H}_A$ where $H = T + V_2 + V_3$ is the pionless EFT Hamiltonian of our system with individual terms defined explicitly in Eq.~\eqref{eq:Tdiscrete} and Eq.~\eqref{eq:Vdiscrete} in the main text. 
We can remember that 
\begin{equation}
    \norm{H}_A  \le \norm{T}_A + \norm{V_2+V_3}_A \le \norm{T}_A + \norm{V_2}_A + \norm{V_3}_A
\end{equation}
and so let's estimate the fermionic semi-norm for each term of the Hamiltonian.
Regarding the kinetic energy, we can go into the momentum space without changing the norm, due to the fact that the Quantum Fourier Transform on all $d$ cartesian components and for all the $\eta$ is a permutation invariant unitary operator. The kinetic energy is diagonal in the momentum space, and its maximum possible value is obtained when the momentum of each particle is the maximum possible $p_{i,max}^2 = d\left(\frac{\pi}{a}\right)^2$.
Therefore we get
\begin{equation}
\label{eq:norm_T}
    \norm{T}_A = \max_{\ket{\psi_A}}\left|\bra{\psi_A}(QFT^{\otimes d\eta})^\dagger T (QFT^{\otimes d\eta})\ket{\psi_A} \right| \le \eta d \frac{\hbar^2}{2\mu} \left(\frac{\pi}{a} \right)^2 = dK\eta2^{2m-2}\;,
    \end{equation}
    where $a$ is the lattice spacing and $K$ is the total strength of the kinetic energy defined in Eq.~\eqref{eq:kkin}.
The potentials are diagonal in the coordinate space, so we can say that
\begin{equation}
    \norm{V}_A = \max_{\ket{\psi_A}}\left|\bra{\psi_A}V\ket{\psi_A} \right|.
\end{equation}
The maximum eigenvalue of the three-body potential can be evaluated as
\begin{equation}
\begin{split}
    \norm{V_3}_A &= \max_{\ket{\psi_A}} \left| \bra{\psi_A} \sum_{i=0}^{\eta-1}\sum_{j\ne i}^{\eta-1} \sum_{k\ne i \ne j}^{\eta-1} \sum_{\vec{r}_i,\vec{r}_j,\vec{r}_k} \frac{G}{6}\delta_{\Vec{r}_i,\Vec{r}_j}\delta_{\Vec{r}_i,\Vec{r}_k} \Pi_{i}(\Vec{r}_i)\Pi_{j}(\Vec{r}_j)\Pi_{k}(\Vec{r}_k) 
    \ket{\psi_A}  \right| \\
    &= \frac{|G|}{6} \max_{\ket{\psi_A}} \left| \bra{\psi_A} \sum_{i=0}^{\eta-1}\sum_{j\ne i}^{\eta-1} \sum_{k\ne i \ne j}^{\eta-1} \sum_{\vec{r}} \Pi_{i}(\Vec{r})\Pi_{j}(\Vec{r})\Pi_{k}(\Vec{r}) \ket{\psi_A}  \right|. \\
    &= |G|(\text{max number of {\  ordered} triplets})
\end{split}
\end{equation}
To find the maximum, we must maximize the number of triplets for a given $\eta$. 
Since we consider $4$ fermionic flavors, we cannot have more than $4$ particles in the same site due to the Pauli exclusion principle. 
Therefore we can arrange particles in quadruplets, triplets, couples, and single particles. These two last cases can be ignored for the three-body potential because they give a null contribution. 
Given $\eta$ particles, they can arrange in at most $\left \lfloor \frac{\eta}{4} \right \rfloor$ quadruplets, or in at most $\left \lfloor \frac{\eta}{3} \right \rfloor$ triplets. 
Noticing that each quadruplet contain $4$ triplets, if we consider to have $x \in \left[0,\left \lfloor \frac{\eta}{4} \right \rfloor  \right]$ quadruplets and $\left \lfloor \frac{\eta-4x}{3} \right \rfloor$ single-triplets, the total number of triplets is
\begin{equation}
\label{eq:tripletstomax}
    N_t=4x + \left \lfloor \frac{\eta-4x}{3} \right \rfloor.
\end{equation}
We can see that it is maximized when $x$ is maximum.
Therefore, the maximum number of triplets is
\begin{equation}
    N_t = \begin{cases}
        4\left \lfloor \frac{\eta}{4} \right \rfloor +1 & \text{if } \eta = 4n+3 \\
        4\left \lfloor \frac{\eta}{4} \right \rfloor & \text{otherwise}
    \end{cases}
\end{equation}
with $n$ integer. Noticing that $N_t\le \eta$ for every $\eta$, the three-body potential norm is bounded by
\begin{equation}
\label{eq:norm_V3}
    \norm{V_3}_A \le \eta |G|.
\end{equation}
For the two-body potential, instead, we can have groups of $2$, $3$, or $4$ particles on the same lattice site. 
In this case, we must evaluate the maximum number of couples. Every quadruplet contains $6$ couples, every triplet $3$, and every couple $1$. We must find the $x\in \left[0,\left \lfloor \frac{\eta}{4} \right \rfloor  \right]$ integer and the $y\in[0,\left \lfloor \frac{\eta-6x}{3}\right \rfloor]$ that maximizes
\begin{equation} 
\label{eq:couplestomax}
    N_c = 6x + 3y + \left \lfloor \frac{\eta-6x - 3y}{2} \right \rfloor
\end{equation}
Following the same idea as before, we maximize the number of quadruplets, and then the number of triplets.
The maximum number of couples the system can form is 
\begin{equation}
\label{eq:N_c}
N_c=\begin{cases}
    6\left \lfloor \frac{\eta}{4} \right \rfloor+ 3 & \text{if } \eta = 4n+3  \\
    6\left \lfloor \frac{\eta}{4} \right \rfloor+ 1 & \text{if } \eta = 4n+2 \\
    6\left \lfloor \frac{\eta}{4} \right \rfloor & \text{otherwise}\\
\end{cases}
\end{equation}
with $n$ integer.
Noticing that $N_c\le 3\eta/2$ for every $\eta$, the two-body potential norm is
\begin{equation}
\label{eq:norm_V2}
    \norm{V_2}_A  \le \frac{3}{2}\eta|C_0|
\end{equation}
A better bound on the potential can be found by realizing that the potential constants $C$ and $G$ have opposite signs, so their contribution partially cancels.
Therefore we want to evaluate $\norm{V_2 + V_3}_A$, but in this case, it is not for granted that the division in quadruplets maximizes the norm. In fact, it strongly depends on the value of the potential constants and on $\eta$. If we put together Eq.\eqref{eq:tripletstomax} and \eqref{eq:couplestomax}, we get
{\ 
\begin{equation}
\begin{split}
     \norm{V_2 + V_3}_A &\le \text{max}_{x,y}\left|N_cC_0 + N_tG \right| \le \text{max}\left\{a_0, a_1, a_2 \right\}.
\end{split}
\end{equation}
where
\begin{equation}
\begin{split}
    a_0 &= |C_0|\left \lfloor \frac{\eta}{2} \right \rfloor \\
    a_1 &= \begin{cases}
        |3C_0+G|\left \lfloor \frac{\eta}{3} \right \rfloor + |C_0| & \eta = 2 \mod 3 \\
        |3C_0+G|\left \lfloor \frac{\eta}{3} \right \rfloor & \text{otherwise} \\
    \end{cases} \\
    a_2 &= \begin{cases}
        |6C_0+4G|\left \lfloor \frac{\eta}{4} \right \rfloor + |C_0| & \eta = 2 \mod 4 \\
        |6C_0+4G|\left \lfloor \frac{\eta}{4} \right \rfloor + \max(|C_0|,|3C_0+G|) & \eta = 3 \mod 4 \\
        |6C_0+4G|\left \lfloor \frac{\eta}{4} \right \rfloor & \text{otherwise} \\
    \end{cases}
\end{split}
\end{equation}
}

\paragraph{Commutator bounds}
Let us call $P_q(t)$ the $q$-th order oracle for the product formulas. We can then define the following general formula:
\begin{equation}
    \norm{e^{-iHt} - P_q(t)} \le t^{q+1} \delta_q
\end{equation}
where
\begin{equation}
\begin{split}
    \delta_1 &= \frac{1}{2}\norm{[T,V]}_A \\
    \delta_2 &= \frac{1}{12} \norm{\left[V, [V,T]\right]}_A + \frac{1}{24} \norm{\left[T, [T,V]\right]}_A \\
    \delta_4 &= 0.0047\norm{[T,[T, [T, [V, T]]]]}_A + 0.0057\norm{[T,[T, [V, [V, T]]]]}_A + 0.0046 \norm{[T,[V, [T, [V, T]]]]}_A +\\
    &+ 0.0074 \norm{[T,[V, [V, [V, T]]]]}_A + 0.0097 \norm{[V,[T, [T, [V, T]]]]}_A + 0.0097 \norm{[V,[T, [V, [V, T]]]]}_A +\\
    &+ 0.0173 \norm{[V,[V, [T, [V, T]]]]}_A + 0.0284 \norm{[V,[V, [V, [V, T]]]]}_A \\
    \label{eq:total_error_Trotter}
\end{split}
\end{equation}
We reported here the definition of $\delta_q$ for simplicity, but one can find their full derivation in \cite{PhysRevX.11.011020}. From this relations we can see the norm of which commutators we have to estimate to find an upper bound to the total error of the product formula. Let us call $\alpha_q$ that estimate, which will be an upper bound to the real value of the errors $\delta_q$ (so that $\delta_q \le \alpha_q$ for every $q$). The full expressions for $\alpha_q$ can be found in Eq.~\eqref{eq:commutatorTV} for $\alpha_1$, Eq.~\eqref{eq:alpha2} for $\alpha_2$ and Eq.~\eqref{eq:alpha4} for $\alpha_4$.

Let us start by computing $\alpha_1$. In order to compute the bound for the commutators between kinetic and potential energy, we can exploit the fact that the potentials are symmetric by the exchange of two particles and that the kinetic energy of a given particle $i$ will not commute only with the potential terms where $i$ is involved. Hence, we have the following properties:  $[T_i, V_{ij}] = [T_i, V_{ji}]$, and  $[T_k, V_{ij}] \ne 0$ only if $k=i$ or $k=j$ due to the commutation relation of position and momentum. This means that we can simplify the commutator as
\begin{equation}
\begin{split}
    [T,V_2] &= \left[ \sum_{n=0}^{\eta-1}T_n, \sum_{i=0}^{\eta-1}\sum_{j\ne i}^{\eta-1} V_{ij} \right]  = \sum_{i=0}^{\eta-1}\sum_{j\ne i}^{\eta-1}\left([T_i, V_{ij}] + [T_j, V_{ij}]\right) = 2\sum_{i=0}^{\eta-1}\sum_{j\ne i}^{\eta-1} [T_i, V_{ij}]\;,
\end{split}
\end{equation}
while for the $3$-body potential a similar argument gives
\begin{equation}
\begin{split}
    [T,V_3] &= \left[ \sum_{n=0}^{\eta-1}T_n, \sum_{i=0}^{\eta-1}\sum_{j\ne i}^{\eta-1} \sum_{k\ne i,j}^{\eta-1} V_{ijk} \right]  \\
    &= \sum_{i=0}^{\eta-1}\sum_{j\ne i}^{\eta-1}\sum_{k\ne i,j}^{\eta-1} \left([T_i, V_{ijk}] + [T_j, V_{ijk}] + [T_k, V_{ijk}]\right) \\
    &= 3\sum_{i=0}^{\eta-1}\sum_{j\ne i}^{\eta-1}\sum_{k\ne i,j}^{\eta-1} [T_i, V_{ijk}]\;.
\end{split}
\end{equation}
By computing the fermionic semi-norm, we can use the triangular inequality to get the following expression:
\begin{equation}
\begin{split}
    \norm{[T,V]}_A &\le \left\|2\sum_{i=0}^{\eta-1}\sum_{j\ne i}^{\eta-1} [T_i, V_{ij}]+3\sum_{i=0}^{\eta-1}\sum_{j\ne i}^{\eta-1}\sum_{k\ne i,j}^{\eta-1} [T_i, V_{ijk}]\right\|_A  \\
    &\le 2 \sum_{i=0}^{\eta-1} \max_i\norm{T_i}_A \max_i\norm{2\sum_{j\ne i}^{\eta-1} V_{ij}+3\sum_{j\ne i}^{\eta-1}\sum_{k\ne i,j}^{\eta-1} V_{ijk}}_A
\end{split}
\end{equation}

So now we simply have to compute the maximum of the kinetic and potential terms. Starting with the kinetic energy, $\max_i\norm{T_i}_A$ is reached when particle $i$ has the maximum possible momentum, thus
\begin{equation}
    \max_i\norm{T_i}_A = dK2^{2m-2}\;
\end{equation}
For the potential term, if we fix a particle $i$, the interaction will depend on how many particles there are on that same site. The calculation is similar to the one done before to estimate $\norm{V_2+V_3}_A$, but since the sum over $i$ is outside the norm there are small differences. If we have $0$ or $1$ particles, the potential energy is zero. If there are $2$ particles we will have one couple, with $3$ particles we have $3$ couples and $1$ triplet, but the particle $i$ will be part of only $2$ couples and $1$ triplet. Finally, with $4$ particles, the $i$-th particle will be part of $3$ couples and $3$ triplets. Moreover, {\  the two and three-body coupling constants $C_0/2$ and $G/6$ include already factors to account for double counting but then we need to consider unordered triples which are twice as many as we counted before.} 
In the end, we have:
\begin{equation}
\max_i\norm{2\sum_{j\ne i}^{\eta-1} V_{ij}+3\sum_{j\ne i}^{\eta-1}\sum_{k\ne i,j}^{\eta-1} V_{ijk}}_A \le \max\left\{\left| C_0 \right|, \left| 2C_0+{\  G}\right|, \left| 3C_0 + {\  3G} \right|\right\}
\end{equation}
Putting all together, the bound for the norm of the commutator becomes:
\begin{equation}
\label{eq:commutatorTV}
    \frac{1}{2}\norm{[T,V]}_A \le \eta dK2^{2m-2}\max\left\{\left| C_0 \right|, \left| 2C_0+{\  G}\right|, {\  3\left| C_0 + G \right|}\right\} = \alpha_1 \;.
\end{equation}

In the same way, we can estimate the commutators for the second order Trotter formula:
\begin{equation}
\begin{split}
    [T,[T, V]] &= \sum_{l=0}^{\eta-1} \left[T_l, 2\sum_{i=0}^{\eta -1} \sum_{j\ne i}^{\eta-1} [T_i, V_{ij}] + 3\sum_{i=0}^{\eta -1} \sum_{j\ne i}^{\eta-1}\sum_{k\ne i\ne j}^{\eta-1} [T_i, V_{ijk}] \right] \\
    &= \sum_{i=0}^{\eta-1} 4\left[ T_i, \sum_{j\ne i}^{\eta-1} [T_i, V_{ij}] \right] + 9 \left[ T_i, \sum_{j\ne i}^{\eta-1}\sum_{k\ne i\ne j}^{\eta-1} [T_i, V_{ijk}] \right] \\
    &= \sum_{i=0}^{\eta-1} \left[ T_i, \left[ T_i, 4\sum_{j\ne i}^{\eta-1} V_{ij} + 9 \sum_{j\ne i}^{\eta-1}\sum_{k\ne i\ne j}^{\eta-1} V_{ijk} \right] \right]
\end{split}
\end{equation}
\begin{equation}
\begin{split}
    \norm{[T,[T, V]]}_A &\le 4 \sum_{i=0}^{\eta-1} \left( \max_i \norm{T_i}_A \right)^2 \max_i\norm{4\sum_{j\ne i}^{\eta-1} V_{ij} + 9\sum_{j\ne i}^{\eta-1}\sum_{k\ne i,j}^{\eta-1} V_{ijk}}_A \\
    &\le 4\eta \left( dK2^{2m-2} \right)^2 \max\left\{\left| 2C_0 \right|, \left| 4C_0+{\  3G}\right|, \left| 6C_0 + {\  9G} \right|\right\} \;
\end{split}
\end{equation}

\begin{equation}
\begin{split}
    [V,[V,T]] &= \sum_{i=0}^{\eta-1} \left[ 2\sum_{j\ne i}^{\eta-1} V_{ij} + 3\sum_{j\ne i}^{\eta-1}\sum_{k\ne i,j}^{\eta-1} V_{ijk}, \left[ 2\sum_{p\ne i}^{\eta-1} V_{ip} + 3\sum_{p\ne i}^{\eta-1}\sum_{q\ne i,j}^{\eta-1} V_{ipq}, T_i \right] \right]
\end{split}
\end{equation}
\begin{equation}
\begin{split}
    \norm{\left[V, \left[V,T \right] \right]}_A &\le 4 \sum_{i=0}^{\eta-1} \max_i \norm{T_i}_A \left( \max_i\norm{2\sum_{j\ne i}^{\eta-1} V_{ij} + 3\sum_{j\ne i}^{\eta-1}\sum_{k\ne i,j}^{\eta-1} V_{ijk}}_A \right)^2 \\
    &\le 4\eta dK2^{2m-2} \left( \max\left\{\left| C_0 \right|, \left| 2C_0+{\  G}\right|, \left| 3C_0 + {\  3G} \right|\right\} \right)^2 \;.
\end{split}
\end{equation}
This means that the upper bound to the error of the product formula will be:
\begin{equation}
\label{eq:alpha2}
\begin{split}
    \alpha_2 = &4\eta \frac{1}{24} \left(dK2^{2m-2} \right)^2 \max\left\{\left| 2C_0 \right|, \left| 4C_0+{\  3G}\right|, \left| 6C_0 + {\  9G} \right|\right\} \\
    + &4\eta \frac{1}{12} \left(dK2^{2m-2} \right) \left( \max\left\{\left| C_0 \right|, \left| 2C_0+{\  G}\right|, \left| 3C_0 + {\  3G} \right|\right\} \right)^2 \\
\end{split}
\end{equation}

In order to simplify the notation, let us define the quantity $M(s)$ as the upper bound to the potential term with modified coupling constants
\begin{equation}
    \max_i\norm{2^s\sum_{j\ne i}^{\eta-1} V_{ij} + 3^s\sum_{j\ne i}^{\eta-1}\sum_{k\ne i,j}^{\eta-1} V_{ijk}}_A \le \max\left\{\left| 2^s\frac{C_0}{2} \right|, \left| 2^sC_0+{\  3^s\frac{G}{3}}\right|, \left| 2^s\frac{3C_0}{2} + {\  3^sG} \right|\right\} = M(s)
\end{equation}

Since we have that $\max_i\norm{T_i}_A = dK2^{2m-2}$ the commutators for the fourth order becomes:
\begin{equation}
\begin{split}
    \norm{[T,[T,[T,[T,V]]]]}_A &\le 2^4 \sum_i \left(\max_i\norm{T_i}_A^4\right) M(4) \\
    \norm{[T,[T,[V,[V,T]]]]}_A &\le 2^4 \sum_i \left( \max_i\norm{T_i}^3_A \right) M(3)^2 \\
    \norm{[T,[V,[V,[V,T]]]]}_A &\le 2^4 \sum_i \left( \max_i\norm{T_i}_A^2 \right) M(2)^3 \\
    \norm{[V,[V,[V,[V,T]]]]}_A &\le 2^4 \sum_i \left( \max_i\norm{T_i}_A \right) M(1)^4 \\
\end{split}
\label{equation:4-order-Trotter-bounds}
\end{equation}

Furthermore, due to our method to bound those norms, it does not matter the position of $T$ and $V$ but only how many $T$ and $V$ are inside the norm. In particular, our upper bounds hold uniformly for every element in the following sets of commutator norms
\begin{equation}
\begin{split}
    S_{3T2V} &=\{\norm{[T,[T,[V,[V,T]]]]}_A, \norm{[T,[V,[T,[V,T]]]]}_A, \norm{[V,[T,[T,[V,T]]]]}_A \}\\
    S_{2T3V} &=\{\norm{[T,[V,[V,[V,T]]]]}_A, \norm{[V,[T,[V,[V,T]]]]}_A, \norm{[V,[V,[T,[V,T]]]]}_A \}\;.
\end{split}
\end{equation}

With those relations, we can explicitly write the expression of $\alpha_4$ as
\begin{equation}
\begin{split}
    \alpha_4 &= 2^4\eta \left( 0,0047 \right) \left( dK2^{2m-2} \right)^4 M(4) + 2^4\eta \left( 0,02 \right) \left( dK2^{2m-2} \right)^3 M(3)^2 \\
    &+ 2^4\eta \left( 0,01883 \right) \left( dK2^{2m-2} \right)^2 M(2)^3 + 2^4\eta \left( 0,0284 \right) \left( dK2^{2m-2} \right)^1 M(1)^4 \\
\end{split}
\label{eq:alpha4}
\end{equation}

\section{Higher order Trotter-Suzuki formulas}
\label{appendix:higher_order_Trotter}

In this section we will present the proof of Theorem~\ref{theorem:2nd_and_fourth_Trotter} in the main text.

\begin{lemma}[Second-order Trotter-Suzuki]
 Consider a system with $\eta$ particles on a $d$ dimensional lattice with $M=2^{m}$ sites per direction described by the pionless Hamiltonian $H = T+V$ with 
    the kinetic and potential terms
    given in Eq.~\eqref{eq:Tdiscrete} and Eq.~\eqref{eq:Vdiscrete}.
    For any $t\in \mathbb{R}$ and $\epsilon > 0$ {\  we construct a quantum circuit implementing} a unitary operator $U(t)$ such that $\norm{U(t) - e^{-iHt}} < \epsilon$ with a $T$-gate count equal to
    \begin{equation}
        T^{(2)}_{tr}(\epsilon, t) = rT_V \left( \frac{\epsilon}{4r} \right) + (r-1) T_T \left( \frac{\epsilon}{4r}, \frac{t}{r} \right) + 2T_T\left( \frac{\epsilon}{8r}, \frac{t}{2r} \right)\;.
    \end{equation}
    where $T_T$ and $T_V$ are given by Lemma~\ref{theorem:kinetic_energy_operator} and Lemma~\ref{theorem:complact_Trotter_potential_energy} respectively.
    The number of ancilla qubits will be
    \begin{equation}
        b_{QFT}^{(2)} + b_{DIAG}^{(2)} +\max \left( 2b_{QFT}^{(2)}-1, b_{DIAG}^{(2)}, m(m-1), dm-1 \right)
    \end{equation}
    with
    \begin{equation}
        b_{DIAG}^{(2)}=\left\lceil\log_2 \left( \frac{24d\eta t\lambda_T}{\epsilon}\log_2\left(\frac{24d\eta t\lambda_T}{\epsilon}\right)\right)\right\rceil\;,
    \end{equation}
    \begin{equation}
        b_{QFT}^{(2)} = \min\left( m-1,  \left\lceil\log_2\left(\frac{12d\eta mr}{\epsilon}\right)\right\rceil \right) +1\;;
    \end{equation}
    \begin{equation}
        b^{(2)}+\max \left( b^{(2)}, m(m-1), dm-1 \right)\quad\text{with}\quad b^{(2)}=\left\lceil\log_2\left(\frac{24d\eta t\lambda_T}{\epsilon}\log_2\left(\frac{24d\eta t\lambda_T}{\epsilon}\right)\right)\right\rceil\;.
    \end{equation}
    where $r = \left \lceil \sqrt{\frac{4t^3}{\epsilon}\alpha_2} \right \rceil$ and $\alpha_2$ can be found in Eq.~\eqref{eq:alpha2}.
    \label{lemma:2nd_order_trotter}
\end{lemma}

\begin{proof}
    Consider the error bound on the second order Trotter formula from Ref.~\cite{PhysRevX.11.011020}
    \begin{equation}
        \norm{e^{itH} - \left(e^{i\frac{t}{2r}T}e^{i\frac{t}{r}V}e^{i\frac{t}{2r}T}\right)^r} \le \frac{t^3}{r^2} \left(\frac{1}{12}\norm{\left[V, [V,T]\right]} + \frac{1}{24}\norm{\left[T, [T,V]\right]} \right) = r \epsilon_{Trotter}^{(2)}
    \end{equation}
    For simplicity, let us write
    \begin{equation}
        \epsilon_{Trotter}^{(2)} \le \frac{t^3}{r^3}\alpha_2
    \end{equation}
    where $\alpha_2$ is a constant whose value depends on the coupling constants $C, D$ and the number of particles $\eta$. As for the explicit calculation of $\alpha_2$, see Appendix \ref{appendix_fermionic_seminorm_Trotter_error} and in particular Equation \eqref{eq:total_error_Trotter}. Taking into account that we have access to oracles that only approximate $e^{iTt}$ and $e^{iVt}$ up to a finite error, we get
    \begin{equation}
        \norm{e^{itH} - \left( U\left(\frac{t}{r}\right) \right)^r} \le r\left( \epsilon_{Trotter}^{(2)} + 2\epsilon_T + \epsilon_V \right) \le \epsilon
    \end{equation}
    For simplicity, let us assume that $\epsilon_T + \epsilon_V = \epsilon/4r$, which leads to:
    \begin{equation}
        r\epsilon_{Trotter}^{(2)} \le \frac{\epsilon}{4}
    \end{equation}
    from which we can set $r$ as
    \begin{equation}
        r = \left \lceil \sqrt{\frac{4t^3}{\epsilon}\alpha_2} \right \rceil
    \end{equation}
    The total cost will be the number of oracle calls to $U_T$ times its cost, plus the number of oracle calls to $U_V$ times its cost. The oracle $U_V$ will be called $r$ times and its cost will be evaluated at time $t/r$. The oracle $U_T$ instead, will be called $r+1$ times, and its cost will be evaluated $r-1$ times at time $t/r$ and $2$ times at $t/2r$. 
    \begin{equation}
         T^{(2)}_{tr}(\epsilon, t) = rT_V \left( \frac{\epsilon}{4r} \right) + (r-1) T_T \left( \frac{\epsilon}{4r}, \frac{t}{r} \right) + 2T_T\left( \frac{\epsilon}{4r}, \frac{t}{2r} \right)\;.
    \end{equation}
    In order to keep a single phase register for the diagonal part when implementing the exponentials of the kinetic energy, we choose instead to perform the initial and final evolution for time $t/2r$ using a target error of $\epsilon/8r$. The cost of this implementation is then
    \begin{equation}
         T^{(2)}_{tr}(\epsilon, t) = rT_V \left( \frac{\epsilon}{4r} \right) + (r-1) T_T \left( \frac{\epsilon}{4r}, \frac{t}{r} \right) + 2T_T\left( \frac{\epsilon}{8r}, \frac{t}{2r} \right)\;.
    \end{equation}
    As for the number of ancilla qubits, we can assume to apply one oracle at a time between $U_T$ and $U_V$, and to be able to use the same ancilla qubits as many times as we want. This means that the number of ancilla qubits required for this algorithm will be the maximum between the number of ancilla qubits required for every oracle taken alone, which is equal to:
    \begin{equation}
        b_{QFT}^{(2)} + b_{DIAG}^{(2)} +\max \left( 2b_{QFT}^{(2)}-1, b_{DIAG}^{(2)}, m(m-1), dm-1 \right)
    \end{equation}
    where the size $b_{DIAG}^{(2)}$ of the phase register needed for Lemma~\ref{lemma:phkick} is given by
    \begin{equation}
    b_{DIAG}^{(2)}= \left\lceil\log_2 \left(\frac{24d\eta t\lambda_T}{\epsilon}\log_2\left(\frac{24d\eta t\lambda_T}{\epsilon}\right)\right)\right\rceil\;,
    \end{equation}
    while the size of the register needed for the Quantum Fourier Transform is
    \begin{equation}
        b_{QFT}^{(2)} = \min\left( m-1, \left\lceil\log_2 \left(\frac{12d\eta mr}{\epsilon}\right)\right\rceil \right) +1\;.
    \end{equation}
\end{proof}

\begin{lemma}[Fourth-order Trotter-Suzuki]
 Consider a system with $\eta$ particles on a $d$ dimensional lattice with $M=2^{m}$ sites per direction described by the pionless Hamiltonian $H = T+V$ with 
    the kinetic and potential terms 
    given in Equations \ref{eq:Tdiscrete} and \ref{eq:Vdiscrete}. 
    For any $t\in \mathbb{R}$ and $\epsilon > 0$ {\  we construct a quantum circuit implementing} a unitary operator $U(t)$ such that $\norm{U(t) - e^{-iHt}} < \epsilon$ with a $T$-gate count equal to
    \begin{equation}
    \begin{split}
    T^{(4)}_{tr}(\epsilon, t) = 5rT_V \left( \frac{\epsilon}{12r} \right) + 6r T_T \left( \frac{\epsilon}{12r}, \frac{t}{r} \frac{1}{4-4^{1/3}}\right)\;,
    \end{split}
    \end{equation}
    where $T_T$ and $T_V$ are given by Lemma~\ref{theorem:kinetic_energy_operator} and Lemma~\ref{theorem:complact_Trotter_potential_energy} respectively.
    The number of ancilla qubits will be
    \begin{equation}
       b_{QFT}^{(4)}+b_{DIAG}^{(4)}+ \max \left( 2b_{QFT}^{(4)}-1, b_{DIAG}^{(4)}, m(m-1), dm-1 \right)
    \end{equation}
    where $r = \left \lceil \left( \frac{12t^5}{\epsilon}\alpha_4 \right)^{1/4} \right \rceil$, $\alpha_4$ can be found in Eq.~\eqref{eq:alpha4}, the phase register size is
    \begin{equation}
    b_{DIAG}^{(4)}=\left\lceil \log_2\left(\frac{6d\eta t\lambda_T}{\epsilon}\frac{12}{4-4^{1/3}}\log_2\left(\frac{6d\eta t\lambda_T}{\epsilon}\frac{12}{4-4^{1/3}}\right)\right)\right\rceil\;,
    \end{equation}
    and the register for the number of ancilla qubits needed for the QFT is
    \begin{equation}
        b_{QFT}^{(4)} = \min\left( m-1, \left\lceil\log_2\left(\frac{36d\eta mr}{\epsilon}\right)\right\rceil \right)+1\;.
    \end{equation}
    \label{lemma:4th_order_trotter}    
\end{lemma}

\begin{proof}
    The oracle for the fourth order Trotter formula can be found in Appendix M of Ref.~\cite{PhysRevX.11.011020}:
    \begin{equation}
        S_4(t) = e^{-ita_6T} e^{-itb_5V} e^{-ita_5T} e^{-itb_4V} e^{-ita_4T} e^{-itb_3V} e^{-ita_3T} e^{-itb_2V} e^{-ita_2T} e^{-itb_1V} e^{-ita_1T}
    \end{equation}
    where the definitions of $a_i$ and $b_i$ are, again, in Appendix M of Ref.~\cite{PhysRevX.11.011020}. If we denote the fourth-order Trotter error by $\epsilon_{Trotter}^{(4)}$, we will have that
    \begin{equation}
        \norm{e^{iHt} - S_4(t)} \le \epsilon_{tr}^{(4)} \le t^5\alpha_4
    \end{equation}
    where $\alpha_4$ can be found in Appendix \ref{appendix_fermionic_seminorm_Trotter_error} and in particular in Eq.~\eqref{eq:alpha4}. Moreover, let us call $U(t)$ the implemented unitary obtained by taking $S_4(t)$ and substituting $e^{iTt}$ and $e^{iVt}$ with unitaries $U_T(t)$ and $U_V(t)$ respectively such that
    \begin{equation}
     \max_{k=\{1,\dots,6\}}\norm{U_T(ta_k)-e^{-ita_kT}}<\epsilon_T\quad\text{and}\quad\max_{k=\{1,\dots,5\}}\norm{U_V(tb_k)-e^{-itb_kV}}<\epsilon_V\;.
    \end{equation}
    Then, we can say that
    \begin{equation}
        \norm{S_4(t) - U(t)} \le 6\epsilon_T + 5\epsilon_V
    \end{equation}
    and by discretizing the time interval we get the following total error
    \begin{equation}
        \norm{e^{iHt} - U \left( \frac{t}{r} \right)^r} \le r(\epsilon_{tr}^{(4)} + 6\epsilon_T + 5\epsilon_V) \le \epsilon
    \end{equation}
    For simplicity, we can assume $\epsilon_T = \epsilon_V = \epsilon/(12r)$, so that we can solve for $r$ obtaining
    \begin{equation}
        r \ge \left( \frac{12 t^5}{\epsilon} \alpha_4 \right)^{\frac{1}{4}}
    \end{equation}
    The cost for the implementation of the evolution under the kinetic energy operator using Lemma~\ref{theorem:kinetic_energy_operator} depends on the evolution time. In order to be able to use a single phase register for all the different times corresponding to the various values of the coefficients $a_j$ we prepare and use a phase register for the longest of these times and use it for all of them. From Ref.~\cite{PhysRevX.11.011020} we have
    \begin{equation}
    2a_1=a_2=a_5=2a_6=\frac{1}{4-4^{1/3}}\quad\text{and}\quad a_3=a_4 = \frac{1-3a_2}{2}\;,
    \end{equation}
    so that $\max_j|a_j| = a_2$. The cost of this implementation will then be bounded by
    \begin{equation}
    T^{(4)}_{tr}(\epsilon, t) = 5rT_V \left( \frac{\epsilon}{12r} \right) + 6r T_T \left( \frac{\epsilon}{12r}, \frac{t}{r} \frac{1}{4-4^{1/3}}\right)\;.
    \end{equation}
    As for the number of ancilla qubits, we can assume to apply one oracle at a time between $U_T$ and $U_V$, and to be able to use the same ancilla qubits as many times as we want. This means that the number of ancilla qubits required for this algorithm will be the maximum between the number of ancilla qubits required for every oracle taken alone. As explained above, the size of the phase gradient register  needed for Lemma~\ref{lemma:phkick} is obtained using the largest one corresponding to the longest times ($a_2$ and $a_4$) , which is equal to:
    \begin{equation}
    b_{DIAG}^{(4)}=\left\lceil \log_2\left(\frac{6d\eta t\lambda_T}{\epsilon}\frac{12}{4-4^{1/3}}\log_2\left(\frac{6d\eta t\lambda_T}{\epsilon}\frac{12}{4-4^{1/3}}\right)\right)\right\rceil\;,
    \end{equation}
    while the size of the register for the QFT is
    \begin{equation}
        b_{QFT}^{(4)} = \min\left( m-1, \left\lceil\log_2\left(\frac{36d\eta mr}{\epsilon}\right)\right\rceil \right)+1\;,
    \end{equation}
    leading to a total number of ancilla qubits required equal to
    \begin{equation}
        b_{QFT}^{(4)} + b_{DIAG}^{(4)}+\max\left(2b_{QFT}^{(4)}-1, b_{DIAG}^{(4)}, m(m-1), dm-1 \right)\;.
    \end{equation}
\end{proof}

\section{Including $SU(4)$ symmetry breaking}
\label{app:su4_break}

{\  An important feature of nuclear interactions is that the $SU(4)$ Wigner symmetry is only approximate and a more realistic model is required to distinguish singlet and triplet states of spin and isospin. At leading order, pionless EFT includes an explicit symmetry breaking term for the spin part at the two body level resulting in a pair interaction that can be written as
\begin{equation}
\label{eq:vlo_twob}
    V^{LO}_2 = \frac{1}{2}\sum_{i=0}^{\eta-1}\sum_{j\ne i}^{\eta-1} \sum_{\Vec{r}_i,\Vec{r}_j} \delta_{\Vec{r}_i,\Vec{r}_j} \Pi_{i}(\Vec{r}_i)\Pi_{j}(\Vec{r}_j) \left(C_0+C_1\vec{\sigma}_i\cdot\vec{\sigma}_j\right)\;,
\end{equation}
with $C_1\ll C_0$ and where the spin operators $\vec{\sigma}_i=(X_i,Y_i,Z_i)$ act on the spin register of the $i-th$ particle. The exponential a single two body term can be implemented following a similar strategy to the one used in Sec.~\ref{sec:Exponential_of_Hamiltonian_terms} of the main text. The difference is that one now as to act directly on the spin registers of particle $i$ and $j$ as shown in Fig.~\ref{fig:pair_exp_circuit_su4_breaking}. In the circuit diagram, $s_i$ and $s_j$ are the two qubits associated to the spin of the two interacting particles. The exponential of the symmetry breaking term shown in the figure can then be implemented using $2$ $T$ gates, two arbitrary rotations around the y axis, a triple controlled $R_Z$ and one Toffoli with four controls plus Clifford gates (see eg. Eq.24 in Ref.~\cite{PhysRevD.107.023007}). A major difference with respect to the protocol described in Lemma~\ref{theorem:Trotter_potential_energy} is that now with the addition of the symmetry braking terms the individual terms in Eq.~\eqref{eq:vlo_twob} no longer commute with each other generating an additional source of Trotter error that must be taken into account when costing the full implementation with product formulae.}

{\  On the other hand, the implementation of the block encoding of the potential in Eq.~\eqref{eq:vlo_twob} does not incur necessarily in a significant overhead. To see this we first use the fact that
\begin{equation}
C_0+C_1\vec{\sigma}_i\cdot\vec{\sigma}_j = C_0+C_1\left(2\text{SWAP}-\mathbb{1}\right)=(C_0-C_1)+2C_1\text{SWAP}\;,
\end{equation}
where $\text{SWAP}$ is a swap operation on the two registers containing $s_i$ and $s_j$. The diagonal part can then be implemented following exactly the same strategy used in Lemma~\ref{lemma:LCU_V} by setting $C=C_0-C_1$. The off-diagonal part instead will need to be implemented separately but one can exploit the fact that $\text{SWAP}$ is already a unitary to find a compact LCU expansion to block encode that part. A direct optimization for this last part is beyond the scope of the present paper and left for future work.}

   \begin{figure}
\centering
    \begin{quantikz}[row sep={0.6cm,between origins}]
    \lstick[3]{$\Vec{r}_i$} & \ctrl{3} & \qw     & \qw    & \qw    & \qw            & \qw    & \qw    & \qw    & \qw & \ctrl{3} & \qw \\
    & \qw                              & \ctrl{3}& \qw    &  \qw   &\qw             & \qw    & \qw    & \qw    & \ctrl{3} & \qw & \qw \\
    & \qw                              & \qw     &\ctrl{3}&  \qw   &\qw             & \qw    & \qw    &\ctrl{3}& \qw & \qw & \qw \\
    \lstick[3]{$\Vec{r}_j$}& \targ{}   & \qw     & \qw    & \qw    & \octrl{1}      & \qw    & \octrl{1}& \qw    & \qw & \targ{} & \qw \\
    & \qw                              & \targ{} & \qw    & \qw    & \octrl{1}      & \qw    & \octrl{1}& \qw    & \targ{} & \qw & \qw \\
    & \qw                              & \qw     & \targ{}&\gate{X}&\gate{R_Z(2C_0t)}&\gate{X}&\octrl{1}& \targ{}& \qw & \qw & \qw \\
    \lstick{$s_i$}                 &\qw&\qw      &\qw     &\qw     &\qw             &\qw     & \gate[2]{\exp(-iC_1t\vec{\sigma}_i\cdot\vec{\sigma}_j)}    &\qw     &\qw  &\qw  &\qw\\
    \lstick{$s_j$}                 &\qw&\qw      &\qw     &\qw     &\qw             &\qw     &     &\qw     &\qw  &\qw  &\qw\\
    \end{quantikz}
    \caption{Circuit implementation of the exponential of a single two body term for $m=1$ and $d=3$ with explicit $SU(4)$ symmetry breaking.}
    \label{fig:pair_exp_circuit_su4_breaking}
\end{figure}

\end{document}